%% file: JBES_formatting.tex
\documentclass[12pt]{article}
\usepackage{amsmath}
\usepackage{graphicx,psfrag,epsf}
\usepackage{enumerate}
\usepackage{natbib}
\usepackage{url} 

\usepackage[bookmarks=false,hypertexnames=false]{hyperref}\hypersetup{colorlinks=true, linkcolor=blue, filecolor=gray, urlcolor=blue, citecolor=blue, }
\usepackage[ruled]{algorithm2e}
\usepackage{multirow}
\usepackage{threeparttable}
\usepackage{array}
\usepackage{fancyhdr}
\usepackage{caption}
\usepackage{graphicx}
\usepackage{pdfpages}
\usepackage{soul}
\usepackage{amsmath}
\usepackage{amssymb}
\usepackage{amsthm}
\usepackage{mathrsfs}
\usepackage{cases}
\usepackage{setspace}
\usepackage{enumerate}
\usepackage{subfigure}
\usepackage{makecell}
\usepackage{float}
\usepackage{listings}
\usepackage{xcolor}
\usepackage{appendix}
\usepackage{pdfpages}
\usepackage{abstract}
\usepackage{longtable,booktabs}
\usepackage{rotating}
\usepackage{lscape}
\usepackage{multicol}
\usepackage{makecell}
\usepackage{natbib}
\usepackage{mathtools}
\usepackage{algpseudocode}
\usepackage{lipsum}

\usepackage{xr}
\makeatletter

\newcommand*{\addFileDependency}[1]{
\typeout{(#1)}
\@addtofilelist{#1}
\IfFileExists{#1}{}{\typeout{No file #1.}}
}\makeatother

\newtheorem{theorem}{Theorem}
\newtheorem{lemma}{Lemma}
\newtheorem{example}{Example}
\newtheorem{definition}{Definition}

\newtheorem{assumption}{Assumption}
\newtheorem{proposition}{Proposition}

\makeatletter
\DeclareRobustCommand\widecheck[1]{{\mathpalette\@widecheck{#1}}}
\def\@widecheck#1#2{%
    \setbox\z@\hbox{\m@th$#1#2$}%
    \setbox\tw@\hbox{\m@th$#1%
       \widehat{%
          \vrule\@width\z@\@height\ht\z@
          \vrule\@height\z@\@width\wd\z@}$}%
    \dp\tw@-\ht\z@
    \@tempdima\ht\z@ \advance\@tempdima2\ht\tw@ \divide\@tempdima\thr@@
    \setbox\tw@\hbox{%
       \raise\@tempdima\hbox{\scalebox{1}[-1]{\lower\@tempdima\box
\tw@}}}%
    {\ooalign{\box\tw@ \cr \box\z@}}}
\makeatother

\renewcommand{\hat}{\widehat}
\renewcommand{\tilde}{\widetilde}
\renewcommand{\check}{\widecheck}
\newcommand{\pigamma}{\dfrac{\|\pi\|_2}{\|\gamma\|_2}}
\newcommand{\pisgamma}{\|\pi\|_2/\|\gamma\|_2}
\newcommand{\expect}{\mathbb{E}}

\newcommand{\var}{{\rm Var}}

\newcommand{\R}{\mathbb{R}}

\newcommand{\IA}{{\rm I}_{A}}
\newcommand{\IAst}{{\rm I}_{A^*}}
\newcommand{\QA}{{\rm Q}_{A}}
\newcommand{\QAst}{{\rm Q}_{A^*}}
\newcommand{\hIA}{\widehat{\rm I}_{A}}
\newcommand{\hQA}{\widehat{\rm Q}_{A}}

\newcommand{\uGamma}{u_{\Gamma}}
\newcommand{\ugamma}{u_{\gamma}}

\newcommand{\hugamma}{\widehat{u}_{\gamma}}

\newcommand{\hu}{\widehat{u}}
\newcommand{\hbeta}{\widehat{\beta}}
\newcommand{\hpi}{\widehat{\pi}}
\newcommand{\hphi}{\widehat{\varphi}}
\newcommand{\hpsi}{\widehat{\psi}}
\newcommand{\hSigma}{\widehat{\Sigma}}
\newcommand{\hOmega}{\widehat{\Omega}}
\newcommand{\hGamma}{\widehat{\Gamma}}
\newcommand{\hgamma}{\widehat{\gamma}}
\newcommand{\hPsi}{\widehat{\Psi}}

\newcommand{\convd}{\stackrel{d}{\longrightarrow}}
\newcommand{\convp}{\stackrel{p}{\longrightarrow}}
\newcommand{\eqd}{\stackrel{d}{=}}

\newcommand{\bOne}{\boldsymbol{1}}

\newcommand{\sumn}{\sum_{i=1}^n}
\newcommand{\lep}{\lesssim_p}
\newcommand{\gep}{\gtrsim_p}

\renewcommand{\hat}{\widehat}
\renewcommand{\tilde}{\widetilde}

\newtheorem{remark}{Remark}

\newcommand{\blind}{0}

\begin{document}

\def\spacingset#1{\renewcommand{\baselinestretch}%
{#1}\small\normalsize} \spacingset{1}


\if0\blind
{
  \title{\bf A Heteroskedasticity-Robust Overidentifying Restriction Test with High-Dimensional Covariates}
\author{Qingliang Fan\\    Department of Economics, The Chinese University of Hong Kong\\
    \\
    Zijian Guo \\
    Department of Statistics, Rutgers University\\  
    \\
    Ziwei Mei \\
    Department of Economics, The Chinese University of Hong Kong\\}
  \maketitle
} \fi

\if1\blind
{
  \bigskip
  \bigskip
  \bigskip
  \begin{center}
    {\LARGE\bf A Heteroskedasticity-Robust Overidentifying Restriction Test with High-Dimensional Covariates}
\end{center}
  \medskip
} \fi

\bigskip
\begin{abstract}
This paper proposes an overidentifying restriction test for high-dimensional linear instrumental variable models. The novelty of the proposed test is that it allows the number of covariates and instruments to be larger than the sample size. The test is scale-invariant and is robust to heteroskedastic errors. To construct the final test statistic, we first introduce a test based on the maximum norm of multiple parameters that could be high-dimensional. The theoretical power based on the maximum norm is higher than that in the modified Cragg-Donald test \citep{kolesar2018minimum}, the only existing test allowing for large-dimensional covariates. Second, following the principle of power enhancement \citep{fan2015power}, we introduce the power-enhanced test, with an asymptotically zero component used to enhance the power to detect some extreme alternatives with many locally invalid instruments. Finally, an empirical example of the trade and economic growth nexus demonstrates the usefulness of the proposed test.
\end{abstract}
\noindent%
\vfill

\newpage
\spacingset{1.5} 

\section{Introduction}\label{sec:intro}
\par Instrumental variable (IV) regression is popular for inference of endogenous effects, whose validity relies on the IV exclusion restrictions. With increasing access to large-scale data, the model with high-dimensional covariates or instruments has drawn considerable attention from the theoretical and empirical literature.  This paper develops a test for IV exclusion restrictions in a high-dimensional model. More precisely, consider the following instrumental variable model, for $i \in\{1, \ldots, n\}$, 
\begin{equation} 
\begin{aligned}
Y_i &= D_i\beta + X_{i\cdot}^{\top}\varphi + Z_{i\cdot}^{\top}\pi + e_i, \quad &\expect(e_i|Z_{i\cdot},X_{i\cdot})=0,\\
D_i &= X_{i\cdot}^{\top}\psi + Z_{i\cdot}^{\top}\gamma + \varepsilon_{i,D},\quad &\expect(\varepsilon_{i,D}|Z_{i\cdot},X_{i\cdot})=0,
\end{aligned}
\label{eq:model}
\end{equation}
where $Y_i\in \R$ is the dependent variable, $D_i\in \R$ is an endogenous variable, $X_{i\cdot}\in \R^{p_x}$ is a vector of exogenous covariates, $Z_{i\cdot}\in \R^{p_z}$ is a vector of instruments and $e_i$, $\varepsilon_{i,D}$ are random errors that may be correlated. In this paper, we allow both $p_x$ and $p_z$ to be larger than $n$ and assume the vectors $\varphi$, $\pi$, $\psi$ and $\gamma$ are sparse, which is specified by Assumption \ref{as:asym}(i) in Section \ref{subsec: asymp norm}. The paper develops a test of the null hypothesis,
\begin{equation} \label{eq:h0}
    \mathbb{H}_0:\pi=0,
\end{equation}
against the alternative $\mathbb{H}_a: \pi \neq 0$. The IVs are \emph{valid} if $\pi=0$. 
\par The classic Sargan test \citep{sargan1958estimation} and J test \citep{hansen1982large} consist of two steps: (1) Compute a two-stage least square (TSLS) estimator of $\beta$, denoted as $\hat\beta_{\rm TSLS}$; (2) Regress $Y-D\hat\beta_{\rm TSLS}$ on the covariates $X$ and IVs $Z$, and test the joint significance of the coefficients of IVs. Our new test follows similar ideas. We first construct a debiased Lasso-based estimator of the parameter $\beta$, denoted as $\hat\beta_A$ in \eqref{eq:hatbetaA}. The estimator is $\sqrt{n}$-consistent and asymptotically normal under the null hypothesis \eqref{eq:h0}. We further run the Lasso regression of $Y_i-D_i \hbeta_A$ on $X_{i\cdot}$ and $Z_{i\cdot}$, and store the debiased estimators of the coefficients of $Z_{i\cdot}$ as $\tilde \pi_A$. Under the null hypothesis \eqref{eq:h0}, $\tilde \pi_A$ is asymptotically equal to the sample average of mean-zero random vectors. The test rejects the null hypothesis if the maximum norm of a scaled version of  $\tilde\pi_A$ exceeds the critical value obtained from a high-dimensional central limit theorem by \cite{chernozhukov2013gaussian}. 

\subsection{Main Results and Contributions}
\par We first design a maximum test (M test) based on the \emph{maximum norm} of the coefficient vector $\tilde \pi_A$ that may be high-dimensional.  In closely related literature, some recent overidentification tests consider a model with a large number of IVs \citep{lee2012hahn,chao2014testing,carrasco2021testing,kolesar2018minimum}, and we refer to these tests based on a limiting $\chi^2$ distribution as ``$\chi^2$-type tests''. None of the $\chi^2$-type tests above allow $p:=p_x+p_z>n$ and $p_x \to \infty$, while our proposed M test covers this scenario. 
Under some commonly imposed sparsity assumptions \citep{Belloni2012sparse,belloni2014}, the M test has the correct asymptotic size. Moreover, when $p$ grows with the sample size and $p<n$, the M test has better power than $\chi^2$-type tests under the sparse regime. 
\par We further propose an add-on asymptotically zero \emph{quadratic} statistic (Q statistic) to improve the power when the model includes many ``locally invalid" IVs, where the individual violation of IV validity is weak; see Section \ref{sec: power enhance} for details. The resulting test, called the \textit{power-enhanced M test} (PM test), rejects the null hypothesis when either the M or the Q statistic is greater than the critical value of the significance level $\alpha$. Our paper extends the principle of power enhancement developed by \cite{fan2015power} and \citet{kock2019power} to the popular IV model and overidentification test. The PM test always has non-inferior power compared to the original M test. In simulations (see Section \ref{sec:simu} of the supplement), we show that the power of the PM test is at least as good as the M test and substantially improved when many IVs are locally invalid.  


\par In the empirical study, we revisit the effect of trade on economic growth. We perform overidentification tests on an IV model with a large number of covariates. The set of instruments includes several possibly invalid instruments, such as energy usage and business environment. The PM test strongly rejects the null hypothesis under the 1\% level. In contrast, the M test rejects the null hypothesis only under 5\%. The modified Cragg-Donald (MCD) test by \citet{kolesar2018minimum}, a representative $
\chi^2$-type test feasible for $p_x\to\infty$ with $p<n$, fails to reject at the 5\% level, indicating the potential power gains of the PM test under high-dimensional IV models. 
\par We summarize the main contributions as follows:
\begin{enumerate}
    \item We propose an overidentification test for IV models with high-dimensional data. To our knowledge, this is the first overidentification test for $p>n$ and $p_x \to\infty$. It is more powerful than $\chi^2$-type tests under certain sparsity restrictions when $p<n$. 
    \item Our test is robust to heteroskedasticity. Our paper extends the current high-dimensional statistical literature on the maximum norm or quadratic form inference to heteroskedastic data.
    \item We develop a power enhancement procedure in the IV validity test context. We use an asymptotically zero statistic to improve the  power for many locally invalid IVs.
\end{enumerate}
\subsection{Other Related Literature}
\par Our test relates to the maximum test in a linear regression model \citep{chernozhukov2013gaussian, zhang2017simultaneous}. The asymptotically zero Q statistic follows an inferential procedure for a quadratic form of high-dimensional parameters \citep{guo2019optimal,tony2020semisupervised,guo2021group}. In terms of high-dimensional IV regression, 
 \cite{Belloni2012sparse, belloni2014} and \cite{chernozhukov2015post} proposed post-selection inference for the endogenous treatment effects. The post-selection method requires covariate and IV selection consistency. Nevertheless, IV selection often suffers from errors in finite samples \citep{guo2023causal}. The treatment effect estimator used in our overidentifying restriction test adopts bias-corrected estimators of quadratic forms that are free from variable selection bias. Recently, \cite{belloni2022high} and \cite{gold2020inference} developed bias-corrected estimators for 
 high-dimensional IV models. These estimators are asymptotically normal when all IVs are valid ($\pi=0$), and thus an overidentifying restriction test with a correct asymptotic size using these estimators is possible. Nevertheless, in the presence of high-dimensional covariates, it is unclear how to derive the limiting distributions of the abovementioned estimators under the alternative $\pi\neq 0$, which brings challenges to deducing the power. In contrast, we identify $\beta$ under the null $\pi=0$ based on quadratic forms of reduced-form coefficients, paving a more transparent way to power analysis.  
\par Another strand of literature has studied the estimation and inference of endogenous treatment effects with potentially invalid instruments \citep{kang2016instrumental,guo2018confidence,windmeijer2019use,fan2022endogenous,gautier22}. In order to identify the treatment effect $\beta$, these methods required model identification conditions, e.g., the \emph{majority rule}, which means more than half of the IVs are valid \citep{kang2016instrumental}.
On the other hand, our test does not require these identification conditions, such as the majority rule, since our primary goal is to test the IV validity. Therefore, our method is complementary to the above studies. 

  \par  Other literature \citep{liao2013adaptive,cheng2015select,caner2018adaptive,chang2021culling} has studied moment condition selection under the GMM framework, requiring prior knowledge of some valid moment conditions. \cite{Chang2021} considered the overidentification test in high-dimensional settings using marginal empirical likelihood ratios and a selective subset of moment conditions. \citet{mikusheva2022inference} studied robust inference with many weak IVs. 


\par \textbf{Notations.} We consider $p=p(n)$ as a function of $n$ and discuss the asymptotics where $n$ and $p$ jointly diverge to infinity. The phrase ``with probability approaching one as $n\to\infty$" is abbreviated as ``w.p.a.1". An
\emph{absolute constant} is a positive, finite constant that is invariant with the sample size. We use ``$\convp$" and ``$\convd$" to denote convergence in probability and distribution, respectively. For any positive sequences $a_n$ and $b_n$, ``$a_n\lesssim b_n$'' means there exists some absolute constant $C$ such that $a_n\leq Cb_n$, ``$a_n\gtrsim b_n$'' means $b_n\lesssim a_n$, and ``$a_n\asymp b_n$'' means $a_n\lesssim b_n$ and $b_n\lesssim a_n$. Correspondingly, ``$\lesssim_p$'', ``$\gtrsim_p$'' and ``$\asymp_p$'' indicate that the aforementioned relations ``$\lesssim$'', ``$\gtrsim$'' and ``$\asymp$'' hold w.p.a.1. ``$a
_n\gg b_n$'' means $a_n/b_n\to\infty$ as $n\to\infty$. We use $[n]$ for some positive integer $n$ to denote the integer set $\{1,2,\cdots,n\}$. For a $p$-dimensional vector $x=(x_{1,}x_{2},\cdots,x_{p})^{\top}$, the number of nonzero entries is $\|x\|_0$, the $L_{2}$ norm is $\left\Vert x\right\Vert _{2}=\sqrt{\sum_{j=1}^{p}x_{j}^{2}}$, the $L_{1}$ norm is $\left\Vert x\right\Vert_{1}=\sum_{j=1}^{n}\left|x_{j}\right|$, and the maximum norm is $\|x\|_{\infty}=\max_{j\in[p]}|x_{j}|$. For a $p\times p$ matrix $A=(A_{ij})_{i,j\in[p]}$, we define the $L_1$ norm $\|A\|_1 = \max_{j\in[p]}\sum_{i\in[p]}|A_{ij}|$ and the maximum norm $\|A\|_\infty = \max_{i,j\in[p]}|A_{i,j}|$. ``$A\succ 0$'' means the matrix $A$ is positive definite. For any $p\times p$ matrix $A\succ 0$ with spectral decomposition $U\Lambda U^\top$, we define $\lambda_{\min}(A)$ and $\lambda_{\max}(A)$ as the minimum and maximum eigenvalues of $A$, and $A^{1/2}=U\Lambda^{1/2} U^\top$ with $\Lambda^{1/2}$ being the diagonal matrix composed of the square roots of the corresponding diagonal elements of $\Lambda$. We use ${\rm diag}(A)$ to denote the diagonal matrix composed of the diagonal elements of $A$. We define $\IA(x,y)=x^\top Ay$ and $\QA(x)=\IA(x,x)$ for any vectors $x,y\in\mathbb{R}^{p}$. We use $0_p$ to denote the $p\times 1$ null vector, $1_p$ to denote the $p\times 1$ vector of ones, and $I_p$ to denote the $p$-dimensional identity matrix. The indicator function is ${\textbf 1}(\cdot)$. Finally, for any $a,b\in\mathbb{R}$, we use $a\vee b$ and $a\wedge b$ to denote $\max(a,b)$ and $\min(a,b)$, respectively. 
\par The remainder of the paper is organized as follows. In Section \ref{sec:model}, we introduce the model and a treatment effect estimator. Section \ref{sec:overid} discusses the M test and its power-enhanced version with their asymptotic properties. We present an empirical example in Section \ref{sec:emp}. Section \ref{sec:conc} concludes the paper. Technical proofs, additional empirical study details, and Monte Carlo simulations are given in the supplement. 

\section{The Model and Treatment Effect Estimation} 
\label{sec:model}
The test operates via a random sample $\{Y_i, D_i, X_{i\cdot}, Z_{i\cdot}\}_{1\leq i\leq n}$ in model \eqref{eq:model}. 
Heteroskedastic errors are allowed so that ${\var}(e_i|Z_{i\cdot},X_{i\cdot})$ and ${\var}(\varepsilon_{i,D}|Z_{i\cdot},X_{i\cdot})$ could vary with $i$. To fix ideas, we assume high-dimensional covariates with $p_x\to\infty$ so that $p=p_x+p_z\to\infty$ with $p_z$ either fixed or growing. Our test, therefore, accommodates the studies with either high-dimensional or a fixed number of instruments.
 
In Section \ref{sec: identification}, we present the treatment effect identification using equation \eqref{betaA}. This identification motivates the data-dependent treatment effect estimator in the following Section \ref{subsec:validity}. We further establish the asymptotic normality of this estimator in Section \ref{subsec: asymp norm}. 
\subsection{Identification and Scale Invariance}
\label{sec: identification}
 Denote $Y=(Y_1,Y_2,\cdots,Y_n)^\top$, $D=(D_1,D_2,\cdots,D_n)^\top$, $X=(X_{1\cdot},X_{2\cdot},\cdots,X_{n\cdot})^\top$ and $Z=(Z_{1\cdot},Z_{2\cdot},\cdots,Z_{n\cdot})^\top$.  The reduced form of model (\ref{eq:model}) is 
\begin{equation}\label{eq: reduced form}
\begin{aligned}
Y = X\Psi + Z\Gamma + \varepsilon_Y,\\
D = X\psi + Z\gamma + \varepsilon_D, 
\end{aligned} 
\end{equation}
\noindent
where $\Psi=\psi\beta+\varphi$, $\Gamma=\gamma\beta+\pi$ and $\varepsilon_Y=\varepsilon_D\beta+e=(\varepsilon_{1,Y},\varepsilon_{2,Y},\cdots,\varepsilon_{n,Y})^\top$ with $e=(e_1,e_2,\cdots,e_n)^\top$ and $\varepsilon_D=(\varepsilon_{1,D},\varepsilon_{2,D},\cdots,\varepsilon_{n,D})^\top$. We write $W_{i\cdot}=(X_{i\cdot}^\top,Z_{i\cdot}^\top)^\top$ for $i=1,2,\cdots,n$ and $W=(W_{1\cdot},W_{2\cdot},\cdots,W_{n\cdot})^\top$. Define the population \emph{Gram matrix} $\Sigma := \expect(W_{i\cdot}W_{i\cdot}^\top)$, and the \emph{precision matrix} $\Omega := \Sigma^{-1}$. Furthermore, define $\sigma_{i,Y}^2:=\var(\varepsilon_{i,Y}|W_{i\cdot})$, $\sigma_{i,D}^2:=\var(\varepsilon_{i,D}|W_{i\cdot})$ and $\sigma_{i,YD}:=\text{cov}(\varepsilon_{i,Y},\varepsilon_{i,D}|W_{i\cdot})$. Let $\widehat{\Sigma} := n^{-1}\sum_{i=1}^n W_{i\cdot}W_{i\cdot}^\top = n^{-1}W^\top W$.

In the literature \citep{chao2014testing}, it is common to use weighted norms of the unknown parameters for the construction of estimators and tests. From model \eqref{eq: reduced form}, we have $\Gamma = \gamma \beta+\pi$. Thus, for any $p$-dimensional square matrix $A$ such that $\QA(\gamma)=\gamma^\top A\gamma > 0$, we have 
\begin{equation}\label{eq: beta identify}
    \beta = \dfrac{\gamma^\top A(\Gamma - \pi)  }{\gamma^\top A\gamma } =\dfrac{\IA(\gamma,\Gamma)-\IA(\gamma,\pi)}{\QA(\gamma)},
\end{equation}
where $\IA(\gamma,\Gamma)=\gamma^{\top} A\Gamma$ and $\IA(\gamma,\pi)=\gamma^{\top} A\pi.$ 
Since $\Gamma = \gamma \beta+\pi$, (\ref{eq: beta identify}) holds when $A$ is either random or fixed. In order to achieve the scale-invariant property, we choose
\begin{equation}\label{eq: def A}
    A:= {\rm diag}\left(\dfrac{Z^\top Z}{n}\right) = {\rm diag}\left(\hat\sigma_{1z}^2,\hat\sigma_{2z}^2,\cdots,\hat\sigma_{p_zz}^2\right),
\end{equation}
where $\hat\sigma_{jz}^2 := n^{-1}\sumn Z_{ij}^2$, for $j=1,2,\dots, p_z$. When the $j$-th IV $Z_{ij}$ is scaled with some number $m>0$, the corresponding coefficient $\gamma_j$ is multiplied by $1/m$ since $Z_{ij}\gamma_j = (mZ_{ij})(\gamma_j/m)$; similar arguments apply to 
$\Gamma_j$ and $\pi_j$. Thus, with the weighting matrix $A$ in \eqref{eq: def A}, the quadratic forms and inner products in \eqref{eq: beta identify} remain unchanged if we scale the instruments by some number $m>0$. 
\par It is easy to show that $\QA(\gamma)>0$ w.p.a.1 under the assumptions in Section \ref{subsec: asymp norm}, and thus we assume $\QA(\gamma)>0$ throughout the theoretical discussions. We define the parameter
\begin{equation}\label{betaA}
    \beta_{A} := { \IA(\gamma,\Gamma) }/{\QA(\gamma)}.
\end{equation}
In Section \ref{subsec:validity}, we apply \eqref{betaA} to derive a data-dependent estimator $\hbeta_A$ of $\beta_{A}$\footnote{We slightly abuse the terminology to say $\hbeta_A$ is an estimator of $\beta_A$ even when the matrix $A$ is random. The same applies to the notation $\widehat{\pi}_{A}$ in Section \ref{subsec: M test}.}. Since $\beta_{A}-\beta={\IA(\gamma,\pi)}/{\QA(\gamma)}$, we have $\beta_{A}=\beta$ under the null hypothesis of $\pi=0,$  
\begin{remark}[Connection to the Sargan test]\label{rem: sargan} When $p_x=0$, the TSLS estimator $\hat\beta_{\rm TSLS} = \frac{D^\top Z(Z^\top Z)^{-1} Z^\top Y}{D^\top Z(Z^\top Z)^{-1} Z^\top D }$ is the estimator for $\beta_A$ with the empirical Gram matrix $A=n^{-1}Z^\top Z$. Write the residual $\hat e_{\rm TSLS}=Y-D\hat\beta_{\rm TSLS}$, the sum of squared residuals $\hat\sigma_{\rm TSLS}^2 = n^{-1}\|\hat e_{\rm TSLS}\|_2^2$, and $\hat\pi_{\rm TSLS}=(Z^\top Z)^{-1}Z^\top\hat e_{\rm TSLS}$. The Sargan test statistic $\hat\sigma_{\rm TSLS}^{-2}\hat\pi_{\rm TSLS}^\top (Z^\top Z/n)\hat\pi_{\rm TSLS}$ weights the quadratic form by $A=n^{-1}Z^\top Z$. However, the sample Gram matrix of $Z$ is random and of large size when $p_z$ is large. It induces excessively large variances to the bias-corrected estimators in Section \ref{subsec:validity}, like (\ref{eq:hatQA}). Therefore, we employ the diagonal weighting matrix $A={\rm diag}(n^{-1}Z^\top Z)$ that is sparse and thus substantially reduces the variance. 
\end{remark}

\subsection{A Debiased Lasso-Based Estimator of $\beta$}
\label{subsec:validity}
   We now introduce an estimator of $\beta_A$ defined in \eqref{betaA}, where $\beta_A = \beta$  when all IVs are valid ($\pi=0$). This estimator is useful to construct the test statistic in Section \ref{sec:overid}. 
\par With the estimators $\hQA(\gamma)$ and $\hIA(\gamma,\Gamma)$  specified later in \eqref{eq:hatQA} and  \eqref{eq:hatIA}, $\beta_A$ can be estimated by 
  \begin{equation}\label{eq:hatbetaA}
   \hbeta_A  = \dfrac{\hIA(\gamma,\Gamma)}{\hQA(\gamma)}{\bf 1}(\hQA(\gamma)>0).
  \end{equation}
In the following, we provide details for estimating $\hIA(\gamma,\Gamma)$ and $\hQA(\gamma).$ We use Lasso \citep{tibshirani1996regression} to get the initial estimates of $\Gamma$ and $\gamma$ in \eqref{eq: reduced form}: 
\begin{equation}
\{\widehat{\Psi}, \widehat{\Gamma}\}=\arg\min_{\Psi,\Gamma} \dfrac{1}{n}\|Y-X\Psi-Z\Gamma\|_2^2+{\lambda_{1n}}(\|\Psi\|_1+\|\Gamma\|_1),
\label{Theta}
\end{equation}
\begin{equation}
\{\widehat{\psi}, \widehat{\gamma}\}=\arg\min_{\psi,\gamma} \dfrac{1}{n}\|D-X\psi-Z\gamma\|_2^2+{\lambda_{2n}}(\|\psi\|_1+\|\gamma\|_1),
\label{theta}
\end{equation}
where $\lambda_{1n},\lambda_{2n}$ are positive tuning parameters that are selected by cross-validation in practice. The plug-in estimator of $\beta_A$ given by $\IA(\hat\gamma,\hat\Gamma)/\QA(\hat\gamma)$ suffers from regularization bias and invalidates asymptotic normality. Therefore, we introduce a debiasing procedure for $\hat \beta_{A}$ through constructing debiasing estimators of $\IA(\gamma,\Gamma)$ and $\QA(\gamma)$. Here, we generalize the debiasing method for the quadratic form of high-dimensional parameters presented in recent literature \citep{guo2019optimal,guo2021group} to heteroskedastic errors. 
\par We specify our bias correction procedure in the following. First, for $\QA(\gamma)$, the denominator of $\beta_A$, the estimation error of the plug-in estimator $\QA(\hgamma)$ is
\begin{equation}
\begin{aligned}
\sqrt{n}\left(\QA(\hgamma) - \QA(\gamma)\right) &=  2\sqrt{n}\hgamma^\top A(\hgamma-\gamma) -\sqrt{n}\QA(\hgamma-\gamma).
\end{aligned}
\label{eqproj1}
\end{equation}

  The second term on the right-hand side (RHS) of \eqref{eqproj1} is asymptotically negligible. The bias of the plug-in estimator $\QA(\hat\gamma)$ is mainly induced by the first term on the RHS of \eqref{eqproj1}, specifically, the regularization bias in the initial LASSO estimator, $\hat\gamma-\gamma$. Thus, we need a bias-corrected estimator of $\gamma$ for an asymptotically normal estimator of $\QA(\gamma)$. Following the idea of \citet{javanmard2014confidence}, a bias-corrected estimator of $\gamma$ is given as 
  \begin{equation}\label{eq:DB_gamma}
\left(\begin{array}{c}
    \tilde{\psi}    \\
   \tilde{\gamma}    
\end{array}\right) = \left(\begin{array}{c}
    \hat{\psi}    \\
   \hat{\gamma}    
\end{array}\right) + \dfrac{1}{n}\hOmega W^\top (D - X\hpsi - Z\hgamma),
\end{equation}
where $\hat\Omega$ is the constrained $L_1$-minimization for inverse matrix estimation (CLIME, \citealp{cai2011constrained}) of the precision matrix $\Omega$. Specifically, let $\hat\Omega^{(1)}$ be the solution of the problem 
\begin{equation}\label{eq:clime1}
    \begin{aligned} \min_{\Omega \in\mathbb{R}^{p\times p}} \|\Omega\|_1,\ \text{s.t. }\|\widehat{\Sigma}\Omega-I_p\|_\infty \leq  \mu_{\omega},
\end{aligned}
\end{equation}
where $I_p$ is the $p$-dimensional identity matrix and $\mu_{\omega}$ is a positive tuning parameter. The CLIME estimator is defined as 
\begin{equation}\label{eq:clime}
\begin{aligned}
\hOmega &= (\hOmega_{jk})_{j,k\in[p]}\text{ where }\hOmega_{jk} = \hOmega_{jk}^{(1)}\boldsymbol{1}(|\hOmega_{jk}^{(1)}|\leq |\hOmega_{kj}^{(1)}|) + \hOmega_{kj}^{(1)}\boldsymbol{1}(|\hOmega_{jk}^{(1)}| > |\hOmega_{kj}^{(1)}|).
\end{aligned}
\end{equation}
The above definition \eqref{eq:clime} guarantees that $\hOmega$ is a symmetric matrix, even if $\hOmega^{(1)}$ is not necessarily symmetric. Particularly, for two different values in $\{\hOmega_{jk}^{(1)}, \hOmega_{kj}^{(1)}\}$, we choose the one with a smaller absolute value, and assign  $\hOmega_{jk}$ as this particular value. This value assignment results in $\hat\Omega_{jk}=\hat\Omega_{kj}$ and thus $\hOmega$ is symmetric. 
 We use the \texttt{fastclime} R package \citep{pang2014fastclime} for efficient computation of CLIME. The difference between \eqref{eq:DB_gamma} and the analog in \citet{javanmard2014confidence} is that we minimize the $L_1$-norm, instead of $L_2$-norm. The $L_1$ minization is also used in \citet{gold2020inference}. Lemma \ref{lem:CLIME} in the supplement establishes convergence rates of the CLIME estimator in  (\ref{eq:clime}). 

\par A bias-corrected estimator of $\QA(\gamma)$ is then given as 
\begin{equation}\label{eq:hatQA}
\hQA(\gamma) = \QA(\hgamma) + 2\hat\gamma^\top A(\tilde\gamma - \hat\gamma),
\end{equation}
where $\hat\gamma$ and $\tilde\gamma$ are respectively defined  in \eqref{theta} and \eqref{eq:DB_gamma}. The estimation error of the debiased estimator $\hQA(\gamma)$ is decomposed as 
\begin{equation}\begin{aligned}
\label{eq:u1hat}
\sqrt{n}(\hQA(\gamma) - \QA(\gamma)) &= 2\sqrt{n}\hat\gamma^\top A (\tilde\gamma-\gamma)  -\sqrt{n}\QA(\hgamma-\gamma). 
\end{aligned}
\end{equation}
The first term on the RHS of \eqref{eq:u1hat} is asymptotically normal since $\tilde\gamma$ is debiased, and the second term is asymptotically negligible. Thus, we can deduce the asymptotic normality of the estimator $\hQA(\gamma)$. 
\begin{remark}\label{rem: quad debias}Note that we do not use $\QA(\tilde\gamma)=\tilde\gamma^\top A\tilde\gamma$, the quadratic form of the debiased estimator $\tilde\gamma$. Though the estimator $\tilde\gamma$ is asymptotically unbiased, $\QA(\tilde\gamma)$  is not a consistent estimator of $\QA(\gamma)$ when $p_z$ is large. Instead, each $\tilde\gamma_j$ is asymptotically normal with a variance of order $1/n$. Thus, $\QA(\tilde\gamma)$ is the sum of $p_z$ squared normal random variables with an order at least $p_z/n$, thereby not necessarily a consistent estimator of $\QA(\gamma)$ when $p_z > n$. 
\end{remark} 
\par Similarly, the estimation error of the plug-in estimator $\IA(\hat\gamma,\hat\Gamma)$ is decomposed as 
\begin{equation}\label{eq: error plug in IA}
\begin{aligned}
\sqrt{n}\left(\IA(\hat\gamma,\hat\Gamma) - \IA(\gamma,\Gamma)\right) &=  \sqrt{n}\hgamma^\top A(\hGamma-\Gamma) + \sqrt{n}\hGamma^\top A(\hgamma-\gamma) -\sqrt{n}\IA(\hgamma-\gamma,\hat\Gamma-\Gamma).
\end{aligned} 
\end{equation}
With a similar motivation as \eqref{eq:hatQA}, we propose the following debiased estimator of $\IA(\gamma,\Gamma)$, 
\begin{equation}\label{eq:hatIA}
    \hIA(\gamma,\Gamma) = \IA(\hgamma,\hGamma) + \hgamma^\top A(\tilde\Gamma-\hGamma) + \hGamma^\top A(\tilde\gamma-\hgamma),
\end{equation}
where $\tilde\Gamma$ is the debiased estimator of $\Gamma$ defined as 
\begin{equation}\label{eq:DB_Gamma}
\left(\begin{array}{c}
    \tilde{\Psi}    \\
   \tilde{\Gamma}    
\end{array}\right) = \left(\begin{array}{c}
    \hat{\Psi}    \\
   \hat{\Gamma}    
\end{array}\right) + \dfrac{1}{n}\hOmega W^\top (Y - X\hat\Psi - Z\hat\Gamma),
\end{equation}
with $\widehat{\Omega}$ defined in \eqref{eq:clime}. 
We can then establish a bias-corrected estimator $\hat\beta_A$ as  \eqref{eq:hatbetaA} using the estimators in \eqref{eq:hatQA} and \eqref{eq:hatIA}. 
\subsection{Asymptotic Property of  $\hbeta_A$}
\label{subsec: asymp norm}
Under Assumptions \ref{as:DecisionMatrix}-\ref{as:tuning} below, we can show that $\hQA(\gamma)>0$ w.p.a.1, and thus the estimation error of $\hbeta_A$ in \eqref{eq:hatbetaA} is decomposed as
\begin{equation}\label{eq:IQdecom}
\hbeta_A - \beta_A = \dfrac{\hIA(\gamma,\Gamma)-\IA(\gamma,\Gamma)-\beta_A\cdot (\hQA(\gamma) - \QA(\gamma))}{\hQA(\gamma)}.
\end{equation}
We establish the asymptotic normality of $\sqrt{n}(\hbeta_A-\beta_{A})$  based on  the decomposition  \eqref{eq:IQdecom} 
and the asymptotic normality of $\sqrt{n}(\hIA(\gamma,\Gamma)-\IA(\gamma,\Gamma))$ and $\sqrt{n}(\hQA(\gamma)-\QA(\gamma))$. To state the theoretical results, we first recall the definition of sub-Gaussian norm \citep{vershynin2010introduction}. 
 \begin{definition}[Sub-Gaussian norm] 
 \label{def: subG}
 The sub-Gaussian norm of any random variable $x$ is
 \begin{equation}\label{eq:def-subG}
     \|x\|_{\psi_2} := \sup_{q\geq 1}\dfrac{1}{\sqrt{q}}[\mathbb{E}|x|^q]^{1/q}.
 \end{equation}
  For any random vector $X\in\mathbb{R}^p$, we define its sub-Gaussian norm as  
   \begin{equation}\label{eq:def-subG-vec}
     \|X\|_{\psi_2} := \sup_{b\in\R^p:\|b\|_2=1} \|b^\top X\|_{\psi_2}. 
 \end{equation}
 \end{definition}
\par We impose the following assumptions to derive the asymptotic properties of $\hbeta_A$.
\begin{assumption}
Suppose that $\{W_{i\cdot}\}_{i\in[n]}$ are independent and identically distributed random vectors with a bounded sub-Gaussian norm. The population Gram matrix $\Sigma$ satisfies $c_\Sigma\leq\lambda_{\min}(\Sigma)\leq\lambda_{\max}(\Sigma)\leq C_\Sigma$ for absolute positive constants $C_\Sigma\geq c_\Sigma >0$. \label{as:DecisionMatrix} 
\end{assumption}
\begin{assumption}
\label{as:error}Suppose that $(e_{i},\varepsilon_{i,D})_{i\leq n}$ are independent across $i$, where $e_i$ and $\varepsilon_{i,D}$ are centered with a bounded sub-Gaussian norm. Assume $\mathbb{E}(e_{i}|W_{i\cdot})=0$, $\mathbb{E}(\varepsilon_{i,D}|W_{i\cdot})=0$ and $ \sigma^2_{\min}\leq \sigma_{i,Y}^2,\sigma_{i,D}^2 \leq \sigma^2_{\max}$ for some absolute constants $\sigma_{\max}\geq \sigma_{\min}>0$. In addition, there exist some absolute constants $c_0$ and $C_0$ such that $\mathbb{E}(|\varepsilon_{i,Y}|^{2+c_0}+|\varepsilon_{i,D}|^{2+c_0}|W_{i\cdot})\leq C_0$. Further assume that $|\sigma_{i,YD}|/(\sigma_{i,Y}\sigma_{i,D})\leq\rho_\sigma<1$. 
\end{assumption}

    Assumption \ref{as:DecisionMatrix} is a sub-Gaussianity condition for the covariates and IVs, with eigenvalue bounds for the population Gram matrix. Assumption \ref{as:error} imposes sub-Gaussianity and bounded conditional moment conditions on the error terms. We rule out the perfect correlation between error terms by bounding the correlation coefficient away from one.

\begin{assumption}\label{as:SI} Define the class of population precision matrices
\begin{equation}
\mathcal{U}\left(m_\omega, q, s_\omega\right):=\left\{\Omega=\left(\omega_{j k}\right)_{j, k=1}^p \succ 0:\|\Omega\|_1 \leq m_\omega, \max _{1 \leq j \leq p} \sum_{k=1}^p\left|\omega_{j k}\right|^q \leq s_\omega\right\} ,
\end{equation}
where $0 \leq q<1$. Suppose that $\Omega \in \mathcal{U}\left(m_\omega, q, s_\omega\right)$ with $m_\omega \geq 1$ and $s_\omega \geq 1$. 
\end{assumption}

    Assumption \ref{as:SI} assumes an approximately sparse precision matrix, which is required to establish the rate convergence of the CLIME estimator (\ref{eq:clime}). Such a sparse precision matrix assumption is widely used for inferential procedures in high-dimensional models \citep{van2014asymptotically,gold2020inference}. 

We now specify the sparsity assumption on model \eqref{eq:model} and its reduced form \eqref{eq: reduced form}. Define the sparsity index $s=\max\{\|\varphi\|_0+\|\pi\|_0,\|\psi\|_0+\|\gamma\|_0, \|\Psi\|_0+\|\Gamma\|_0\}$, and the probability limit of the weighting matrix $A$ as 
\begin{equation}\label{eq:def Astar}
    A^* := {\rm diag}\left(\mathbb{E}(Z_{i\cdot}Z_{i\cdot}^\top)\right) =  {\rm diag}\left(\sigma_{1z}^2,\sigma_{2z}^2,\cdots,\sigma_{p_zz}^2\right),
\end{equation}
where $\sigma_{jz}^2 := \mathbb{E}(Z_{ij}^2)$ for $j=1,2,\dots , p_z$. 
\begin{assumption}\label{as:asym}
Define $r_n := \frac{s_\omega m_\omega^{3-2q} s^{(3-q) / 2}(\log p)^{(7+\nu-q) / 2}}{n^{(1-q) / 2}}$, where $\nu \in (0,1)$ is an absolute constant. Suppose that 
\begin{enumerate}[(i)]
    \item $r_n \to 0$ as $n\to\infty$; 
    \item (IV Strength) $\sqrt{\QAst(\gamma)} \gg r_n$.
\end{enumerate}
\end{assumption} 

    Assumption \ref{as:asym}(i) imposes the sparsity conditions by requiring an upper bound on $s$. Assumption \ref{as:asym}(i) further implies $(\log p)^{7} = o(n^{c_\nu})$ with $c_\nu=7/(7+\nu)\in(0,1)$, which is required for the Gaussian approximation property used for the M test in the next section.  Assumption \ref{as:asym}(ii)  provides an asymptotic lower bound for the global IV strength $\sqrt{\QAst(\gamma)} \asymp \|\gamma\|_2$. In classical low-dimensional IV models, strong IVs satisfy $\|\gamma\|_2\gg n^{-1/2}$. Under an exact sparse precision matrix with $q=0$ and constant sparsity indices $s_\omega$, $m_\omega$, and $s$,  Assumption \ref{as:asym}(ii) becomes $\QAst(\gamma) \gg (\log p)^{7+\nu}/n$ and is almost equivalent to the strong IV condition $\|\gamma\|_2^2 \gg 1/n$ under low dimensions up to a logarithmic term. Here, we only need global, not individual, strength for high-dimensional $\gamma$; the latter is required for post-selection inference \citep{guo2018confidence,guo2018testing}.  

\begin{assumption}[Tuning Parameters]\label{as:tuning}Suppose the following conditions hold:
\begin{enumerate}[(i)]
\item  The Lasso tuning parameters satisfy $\lambda_{\ell n}=C_{\ell} \sqrt{\log p / n}$ for $\ell=1,2$, where $\min \left\{C_{1}, C_{2}\right\} \geq C_\lambda$ with a sufficiently large absolute constant $C_\lambda$. 
\item The tuning parameters for the CLIME estimator in (\ref{eq:clime}) satisfy $\mu_{\omega} = C_\omega\sqrt{\log p/n}$ with a sufficiently large absolute constant $C_\omega$.
\end{enumerate}
\end{assumption}

Assumption \ref{as:tuning} specifies the theoretical rates for the tuning parameters. Similar restrictions are commonly used in Lasso-based estimation and inference methods \citep{bickel2009simultaneous,javanmard2014confidence,gold2020inference,belloni2022high}. These rates are necessary for theoretical analysis and merely technical. We use the data-driven tuning parameter selection for practical implementation. Details are available in Section \ref{sec:simu} for simulations in the supplement. 

The following theorem shows the asymptotic normality of $\hbeta_A$. 
 \begin{theorem}\label{thm:betaAsympNorm} Suppose that Assumptions \ref{as:DecisionMatrix}-\ref{as:tuning} hold and $\pi=0$. Then,
 \begin{equation}
    (\hat{\rm V}_\beta)^{-1/2}\sqrt{n}(\hbeta_A - \beta)  \convd N(0,1),
 \end{equation}  
 where $\hat{\rm V}_\beta=\hQA(\gamma)^{-2}n^{-1}\sum_{i=1}^n(W_{i\cdot}^\top\hugamma)^2(\hat{\varepsilon}_{i,Y}-\hbeta_A\hat{\varepsilon}_{i,D})^2$ and $\hugamma = \hat\Omega(0_{p_x}^\top, (A\hat\gamma)^\top )^\top$. 
\end{theorem} 

Theorem \ref{thm:betaAsympNorm} shows that we can use $\hbeta_A$ for inference on the treatment effect when all IVs are valid. Under the null hypothesis (\ref{eq:h0}) where all IVs are valid, the estimator $\hbeta_A$ is an alternative to the existing post-selection procedures \citep{belloni2014,chernozhukov2015post} without depending on variable selection consistency. The suitability of $\hbeta_A$ is further demonstrated by the simulation results in Section \ref{app:betaSimul} of the supplement. In the next section, we use this initial estimator $\hbeta_A$ to construct the overidentification test for its convenience in deriving the asymptotic properties of the test statistic. In the proof of Theorem \ref{thm:betaAsympNorm} in Section \ref{app:b22}, we also deduce the asymptotic normality of $\hat\beta_A-\beta_A$ for $\pi\neq0$ under the alternative set defined as (\ref{eq: HAstar}) below, which is useful in analyzing the power of our test.  
\section{Overidentifying Restriction Test}
\label{sec:overid}
So far, we have developed an estimator $\hat\beta_A$ in \eqref{eq:hatbetaA}. In this section, we develop testing procedures for the IV exclusion restriction \eqref{eq:h0} using this estimator.
Mainly, we test the weighted version of restriction $A^{1/2}\pi=0$ with $A = {\rm diag}(Z^\top Z/n)$. First, subtracting $D\beta_A$ from both sides of (\ref{eq:model}) yields
\begin{equation}
Y - D\beta_A = X\varphi_A + Z\pi_A + e_A, 
\label{deY}
\end{equation}
where $\varphi_A=\varphi-\psi(\beta_A-\beta)$, $\pi_A=\pi-\gamma(\beta_A-\beta)$ and $e_A=\varepsilon_Y-\varepsilon_D\beta_A$. 
Note that we identify $\pi_A$, not the true $\pi$, from \eqref{deY}. When $\gamma\neq 0$, $\pi = 0$ implies $\beta_A=\beta$ and hence $\pi_A=0$. 
\par Next, we derive the if and only if condition for equivalence between $\pi_A=0$ and $\pi=0$. To see this, we define the weighted quadratic forms of $\pi_A$ and $\pi$ as $\QA(\pi_A)=\pi_A^\top A\pi_A$ and $\QA(\pi)=\pi^\top A\pi$.  Following from the definition of $\beta_A$ in \eqref{betaA}, we establish the following condition between $\QA(\pi_A)$ and $\QA(\pi)$
\begin{equation}\label{eq:Q}
    \QA(\pi_A)=  \QA(\pi)\left[1-{\rm R}_A^2(\pi,\gamma)\right],
\end{equation}  
where ${\rm R}_A(\pi,\gamma)=\dfrac{\IA(\pi,\gamma)}{\sqrt{\QA(\pi)\QA(\gamma)}}\boldsymbol{1}\{\QA(\pi)>0,\QA(\gamma)>0\}$ is the relatedness between $A^{1/2}\pi$ and $A^{1/2}\gamma$. By (\ref{eq:Q}), if $|{\rm R}_A(\pi,\gamma)|\neq 1$, $\pi_A = 0$ if and only if $\pi = 0$ and hence it is equivalent to work with the following hypothesis for testing the null in \eqref{eq:h0},
\begin{equation}
A^{1/2}\pi_A=0.
\label{eq: weighted pi A}
\end{equation} 
 We interpret the condition $|{\rm R}_A(\pi,\gamma)|\neq1$ in the following Remark \ref{rem: R neq 1}. 
\begin{remark}\label{rem: R neq 1}
The inequality $|{\rm R}_A(\pi,\gamma)|\neq 1$ means that the weighted vectors $A^{1/2}\pi$ and $A^{1/2}\gamma$ are not perfectly parallel. A specific counterexample is $p_z=1$, which entails that $|{\rm R}_A(\pi,\gamma)| = 1$. This is why our test, like any other test for IV validity, requires overidentifying conditions. In Section \ref{app: pi piA} of the supplement, we provide more detailed discussions with several examples concerning ${\rm R}_A(\pi,\gamma)$ and the relation between $A^{1/2}\pi$ and $A^{1/2}\pi_A$. In later discussions about the power of the tests, we assume  $|{\rm R}_{A^*}(\pi,\gamma)|$ is bounded away from 1 in the alternative sets \eqref{eq: HAstar} and \eqref{eq: H Q}, where ${\rm R}_{A^*}(\pi,\gamma)$ is defined in (\ref{eq: RAstar}), and $A^*$ defined in (\ref{eq:def Astar}) is a population version of $A$.
\end{remark}
\par In the following subsections, we propose the testing procedure for the null hypothesis in \eqref{eq: weighted pi A}. Section \ref{subsec: M test} introduces a testing procedure for \eqref{eq: weighted pi A} using the maximum norm $\|A^{1/2}\pi_A\|_\infty$. Intuitively, the maximum test is powerful when $\pi$ is sparse but with relatively large absolute value of $\pi_j$. However, when there are many locally invalid IVs, the maximum test might be less powerful than a quadratic form based test. Inspired by the principle of power enhancement \citep{fan2015power, kock2019power}, in Section \ref{sec: power enhance}, we construct an asymptotically zero quadratic statistic by an estimator of $\QA(\pi_A)$ and use it to enhance the power of the original M test. 
 
\subsection{The M Test}\label{subsec: M test}

We start with constructing an estimator of $\pi_A$ and apply it to construct our proposed maximum test. Substituting $\beta_A$ by $\hbeta_A$ in equation \eqref{deY}, we have
\begin{equation} 
\label{eq:pluginbetahat}
Y-D\widehat{\beta}_A = X\widecheck{\varphi}_A + Z\widecheck{\pi}_A + \widecheck{e}_A, 
\end{equation} 
where\footnote{Throughout the paper, the subscript $A$ stands for a transformed variable or parameters using the unknown $\beta_A$. In addition, for generic notation $\theta$, $\check \theta_A$ stands for the transformed variables or parameters using the estimator $\hat\beta_A$, $\hat \theta$ denotes Lasso estimators or residuals, and $\tilde \theta$ represents debiased Lasso estimators. } $\widecheck{\varphi}_A=\varphi-\psi(\widehat{\beta}_A-\beta)$, $\widecheck{\pi}_A=\pi-\gamma(\widehat{\beta}_A-\beta)=\pi_A-\gamma(\widehat{\beta}_A-\beta_A)$ and $\widecheck{e}_A=\varepsilon_Y-\varepsilon_D\widehat{\beta}_A=e_A-\varepsilon_D(\beta_A-\hbeta_A)$. The left hand side, $Y-D\hbeta_A$, is analogous to the ``residual" in the Sargan test.
We apply Lasso to estimate $\widecheck \pi_A $ from \eqref{eq:pluginbetahat}, 
\begin{equation}\label{eq: Lasso pi}
\{\widehat{\varphi}_A,\widehat{\pi}_A\} = \arg\min_{\widecheck{\varphi}_A,\widecheck{\pi}_A} \dfrac{1}{n}\|Y-D\widehat{\beta}_A-X\widecheck{\varphi}_A-Z\widecheck{\pi}_A\|_2^2 + \lambda_{3n}(\|\widecheck{\varphi}_A\|_1 + \|\widecheck \pi_A\|_1),
\end{equation}
where $\lambda_{3n}$ is a positive tuning parameter selected by cross-validation in practice. The bias-corrected estimator for $({\varphi}_A^{\top}, \pi_A^{\top})^{\top}$ is given by  
\begin{equation}\label{eq:DB}
\left(\begin{array}{c}
    \tilde{\varphi}_A    \\
   \tilde{\pi}_A    
\end{array}\right) = \left(\begin{array}{c}
    \hat{\varphi}_A    \\
   \hat{\pi}_A    
\end{array}\right) + \dfrac{1}{n}\hOmega W^\top (Y - D\hbeta_A - X\hphi_A - Z\hpi_A),  
\end{equation}
where $\hOmega$ is defined by (\ref{eq:clime}). We use this bias-corrected $\tilde{\pi}_A$ in the maximum test.

Next, we give the approximate distribution of $\tilde{\pi}_A$. Let $\hOmega_z$ be the $p_z\times p$ submatrix composed of the last $p_z$ rows of $\hOmega$. 
We can deduce the following approximation under the null hypothesis $\pi=0$:
\begin{equation}\label{eq:DB-error2} 
\begin{aligned}
  \sqrt{n} A^{1/2}\tilde{\pi}_A \approx A^{1/2}\left(I_{p_z}-\dfrac{\hgamma \hgamma^\top A}{\hQA(\gamma)}\right)\dfrac{\hOmega_z W^\top e_A }{\sqrt{n}}.
\end{aligned}
\end{equation} 
By the form of the RHS of (\ref{eq:DB-error2}), the asymptotic covariance matrix of $ \sqrt{n} A^{1/2}\tilde{\pi}_A $ can be approximated by
\begin{equation}\label{eq:hatVA}
    \widehat{\rm V}_A =  \dfrac{\widehat{A}_0\hOmega_z\sum_{i=1}^n W_{i\cdot}W_{i\cdot}^\top\hat{e}_{iA}^2\hOmega_z^\top \hat{A}_0^\top}{n},
\end{equation}
where $\widehat{A}_0 = A^{1/2}\left(I_{p_z}-\dfrac{\hgamma \hgamma^\top A}{\hQA(\gamma)}\right)$ and $\widehat{e}_{iA}=Y_i-D_i\widehat{\beta}_A-X_{i\cdot}^\top\widehat{\varphi}_A-Z_{i\cdot}^\top\widehat{\pi}_A$. By  \citet{chernozhukov2013gaussian}, the distribution of $\sqrt{n}\|A^{1/2}\tilde\pi_A\|_\infty$
 can be well approximated by that of $\|\eta\|_\infty$, where $\eta\sim N(0, \widehat{\rm V}_A )$ conditionally on the observed data. 
 \par The \emph{M statistic} is defined as
 \begin{equation}\label{eq: def Mn}
M_n(A):=\sqrt{n}\|A^{1/2}\tilde\pi_A\|_\infty.
 \end{equation} 
Then, under any significance level $\alpha$, the M test rejects the null hypothesis when $M_n(A) > {\rm cv}_A(\alpha)$,  
 where the critical value ${\rm cv}_A(\alpha)$ is given as 
 \begin{equation}\label{eq: cv no hat}
     {\rm cv}_A(\alpha) = \text{inf}\{x\in\mathbb{R}:{\rm Pr}(\|\eta\|_\infty \leq x | \widehat{\rm V}_A ) \geq 1-\alpha\}. 
 \end{equation}
In practice, ${\rm cv}_A(\alpha)$ can be approximated by simulating independent draws $\eta\sim N(0,\widehat{\rm V}_A)$, following \citet{chernozhukov2013gaussian} and \citet {zhang2017simultaneous}. 

\par We then define the alternative set of $\pi$ for theoretical justification of the M test. Recall $A^*$ defined in (\ref{eq:def Astar})
is the probability limit of the weighting matrix $A$. Define the relatedness between $A^{*1/2}\pi$ and $A^{*1/2}\gamma$ as 
\begin{equation}\label{eq: RAstar}
     {\rm R}_{A^*}(\pi,\gamma) := \dfrac{\IAst(\pi,\gamma)}{\sqrt{\QAst(\pi)\QAst(\gamma)}}\bOne\{\QAst(\pi)>0,\QAst(\gamma)>0\},
\end{equation} 
similar to the relatedness in (\ref{eq:Q}) with the weighting matrix $A$. 
Treating all other parameters such as $\beta$, $\gamma$ as given, we define the alternative set of $\pi$, for any $t>0$, as
\begin{equation}\label{eq: HAstar}
    \mathcal{H}_{M}(t) := \{\pi\in\mathbb{R}^p:\|{A^*}^{1/2}\pi_{A^*}\|_\infty = t\sqrt{\log p_z / n},\ |{\rm R}_{A^*}(\pi,\gamma)|\leq c_r \},
\end{equation}
for some absolute constant $c_r\in(0,1)$, where 
\begin{equation}\label{eq: def pi A star}
    \pi_{A^*} := \pi -\gamma (\beta_{A^*} - \beta),
\end{equation}
and $\beta_{A^*}:=\IAst(\gamma,\Gamma)/\QAst(\gamma)$ are defined similarly to $\pi_A$ below (\ref{deY}) and $\beta_A$ in (\ref{betaA}) with $A$ replaced by $A^*$. We have the following technical assumptions, which are important for the theoretical properties of the M test.
\begin{assumption}\label{as: Lasso tuning pi}
The Lasso tuning parameter for (\ref{eq: Lasso pi}) satisfies $\lambda_{3n} = C_{3}\sqrt{\log p/ n}$, where $C_{3}\geq C_\lambda \left(1+\|\pi\|_2/\|\gamma\|_2\right)$ with some sufficiently large absolute constant $C_\lambda$.
\end{assumption}
\begin{remark}
The rate specified in Assumption \ref{as: Lasso tuning pi} is the same as in Assumption \ref{as:tuning}(i). Note that the lower bound for the constant $C_3$ is determined by $\|\pi\|_2/\|\gamma\|_2$ since the ``residual" $Y-D\hbeta_A$ in \eqref{eq:pluginbetahat} depends on the estimator $\hat\beta_A$, and the estimation error $\hat\beta_A - \beta_A$ relates to $\|\pi\|_2/\|\gamma\|_2$ when $\pi\neq0$.  
\end{remark}

\par Recall that $\widehat{\rm V}_A$ defined in (\ref{eq:hatVA}) estimates the asymptotic variance of $\sqrt{n}A^{1/2}\tilde\pi_A$, whose limiting form ${\rm V}_{A^*}$ is defined in  (\ref{def:VAstar}) in the supplement. 
The following assumption is needed to establish that the diagonal elements of ${\rm V}_{A^*}$ are lower-bounded away from zero, which is required for the theoretical justification of the maximum test.  
\begin{assumption}\label{as: two strong IV}
Suppose that there exists some absolute constant $C_\gamma \in (0,1)$ such that 
\[\dfrac{\max_{j\in[p_z]} \sigma_{jz}^2 \gamma_{j}^2}{\sum_{j\in[p_z]} \sigma_{jz}^2\gamma_{j}^2} \leq C_\gamma < 1,\] for all $j\in[p_z]$, where $\sigma_{jz}^2$ is defined in (\ref{eq:def Astar}). 
\end{assumption}
Assumption \ref{as: two strong IV} can be interpreted as an overidentification condition: the global weighted IV strength $\sqrt{\QAst(\gamma)}$ cannot be dominated by only one of the IVs. In other words, the model needs to be overidentified by two dominating IVs with the same order of strength. 
\begin{theorem}[asymptotic size and power of the M test]\label{thm:M test}Suppose that Assumptions \ref{as:DecisionMatrix}-\ref{as: two strong IV} hold. Then, the statistic $M_n(A)$ defined by (\ref{eq: def Mn}) satisfies the following:
\begin{enumerate}[(a)]
\item When $\pi = 0$, 
 \begin{equation}\label{eq: M test size}
 \sup_{\alpha\in (0,1) } \left| \Pr\left(  M_n(A) > {\rm cv}_A(\alpha)\right) - \alpha \right| \to 0,
\end{equation}
where ${\rm cv}_A(\alpha)$ is defined in (\ref{eq: cv no hat}). 
\item Suppose that $p_z\to\infty$ as $n\to\infty$. There exists some absolute constant $C_\pi$ such that for any constant $\epsilon > 0$ and $\alpha\in(0,1)$, 
 \begin{equation} \label{eq: M test power}
\inf_{\pi \in \mathcal{H}_{M}(C_\pi + \epsilon)} \Pr\left(  M_n(A) > {\rm cv}_A(\alpha) \right)  \to 1, 
\end{equation}
where $\mathcal{H}_{M}(\cdot)$ is defined by (\ref{eq: HAstar}). 
\end{enumerate} 
\end{theorem}

\begin{remark}[Power for low dimensional IVs]\label{rem:pz}
    In Theorem \ref{thm:M test}(b), we assume $p_z\to\infty$ for simplicity. When $p_z$ is fixed, the $\sqrt{\log p_z}$ in the alternative set (\ref{eq: HAstar}) can be replaced by any sequence that diverges to infinity. Hence, the alternative can be detected at the rate $n^{-1/2}$ when $p_z$ is fixed, which is aligned with the Sargan test under a fixed $p$. 
\end{remark}

\begin{remark}[The range of $\pi$ for power analysis]  \label{rem:range pi} For conciseness of exposition, we only display the local power of the M test in Theorem \ref{thm:M test}(b) under the alternative set (\ref{eq: HAstar}). Our test has asymptotic power 1 not only for a vector $\pi$ satisfying $\|A^{*1/2}\pi_{A^*}\|_\infty = C_\pi\sqrt{\log p_z/n}$ as specified in (\ref{eq: HAstar}), but also any $\pi\neq0$ such that $\|A^{*1/2}\pi_{A^*}\|_\infty \gg \sqrt{\log p_z/n}$, as long as $\|\pi\|_2/\|\gamma\|_2$ is bounded so that the variance of the error term $e_A$ in the regression (\ref{deY}) is finite. Under the lower bound of IV strength by Assumption \ref{as:asym}(ii), the bound of $\|\pi\|_2/\|\gamma\|_2$ holds for the alternative set (\ref{eq: HAstar}). This result also applies to the power analysis for the Q statistic in Theorem \ref{thm:Q}(b).  
\end{remark}

\begin{remark}[Power comparison to $\chi^2$-test] \label{rem:power max chi2}
Note that when $p>n$ and $p_x\to\infty$, the $\chi^2$-type tests are infeasible. We thus focus on $p<n$ and $p\to\infty$ for power comparison, under which both the $\chi^2$-type test and our M test are feasible. 
The previous studies \citep{donald2003empirical,Okui2011,chao2014testing,kolesar2018minimum} have established that the $\chi^2$-type tests have asymptotic power 1 if the vector $\pi_{A^*}$ defined below (\ref{eq: HAstar}) satisfies  $\|\pi_{A^*}\|_2 \gg p_z^{1/4}/\sqrt{n}$. By Theorem \ref{thm:M test}, when $\|\pi_{A^*}\|_\infty  \gg \sqrt{\log p_z/n} $, our proposed M test has asymptotic power 1. Under the sparsity condition $s_\pi\log p_z = o(\sqrt{p_z})$, $\|\pi_{A^*}\|_2 \gg p_z^{1/4}/\sqrt{n}$ implies $\|\pi_{A^*}\|_\infty  \gg \sqrt{\log p_z/n}$. That means our proposed M test achieves power 1 for the regime under which the $\chi^2$-type tests achieve power 1. On the other hand, there exist certain cases (e.g., $s_\pi=1,p_z\to\infty,\|\pi_{A^*}\|_\infty=\|\pi_{A^*}\|_2=\log p_z/\sqrt{n}$) under which the M test achieves asymptotic power 1, but the $\chi^2$ test does not. Thus, if $s_\pi\log p_z = o(\sqrt{p_z})$, the M test has higher power than the $\chi^2$ test even when $p<n$. Note that, when $p_z\gtrsim n^{2/3}$, the sparsity condition $s_\pi\log p_z = o(\sqrt{p_z})$ is implied by Assumption \ref{as:asym}(i). 
\end{remark}


\subsection{Power Enhancement}
\label{sec: power enhance}
\par As discussed earlier, the M test might not be powerful enough when there are many locally invalid IVs. In this case, a test statistic used to estimate the weighted quadratic form $\QA(\pi_A)=\pi_A^\top A\pi_A$ can be leveraged for power enhancement. 
\par Theorem \ref{thm:M test} shows that the M statistic $M_n(A)$ defined by (\ref{eq: def Mn}) satisfies $\Pr(M_n(A) > {\rm cv}_A(\alpha) ) \to \alpha$ as $n\to\infty$. Suppose that we have another statistic $q_n(A)\convp 0$ as $n\to\infty$ under the null hypothesis. Define $PM_n(A) := M_n(A)\vee q_n(A)$. Then, the PM test,  
\begin{equation}\label{eq: PM test}
    \mathcal{PM}_A(\alpha) =  \bOne\{PM_n(A) > {\rm cv}_A(\alpha)\},
\end{equation}
also has asymptotic size $\alpha$ with power at least the same as that of the M test $\bOne\{M_n(A) > {\rm cv}(\alpha) \}$. We then construct an asymptotically zero statistic $q_n(A)$ named as the Q test statistic in (\ref{eq:def Qn}) that measures the magnitude of $\QA(\pi_A)$. This Q test is only for power enhancement and we do not perform this test individually.
\par Following the same idea about the debiased estimators of $\QA(\gamma)$ and $\IA(\gamma,\Gamma)$ in (\ref{eq:hatQA}) and (\ref{eq:hatIA}), we construct the following bias-corrected estimator of $\QA(\pi_A)$:  
\begin{equation}
\hQA(\pi_A) = \QA(\widehat{\pi}_A) + 2\hat\pi_A^\top A (\tilde\pi_A-\hat\pi_A), 
\label{eq: bias correction QA}
\end{equation}
where $\hat\pi_A$ and $\tilde\pi_A$ are defined in (\ref{eq: Lasso pi}) and (\ref{eq:DB}) respectively. We then define the Q statistic as 
\begin{equation}\label{eq:def Qn}
    q_n(A) := \sqrt{n}\log p\hQA(\pi_A).
\end{equation} 
  For ease of discussion, we define a new alternative set 
  \begin{equation}\label{eq: H Q}
      \mathcal{H}_Q(t) := \{\pi\in\mathbb{R}^{p_z}:\|\pi_{A^*}\|_2 = tn^{-1/4},\ |{\rm R}_{A^*}(\pi,\gamma)|\leq c_r\},
  \end{equation}  with $|{\rm R}_{A^*}(\pi,\gamma)|$ defined by (\ref{eq: RAstar}) and the absolute constant $c_r\in(0,1)$ used in (\ref{eq: HAstar}). We have the following results in favor of the asymptotically zero Q statistic $q_n(A)$.
\begin{theorem}
    \label{thm:Q}Suppose that Assumptions \ref{as:DecisionMatrix}-\ref{as: Lasso tuning pi} hold. Then the estimator $\hQA(\pi_A)$ has the following decomposition: 
\begin{equation}\label{eq: lem QA decom}
\hQA(\pi_A)  = \QA(\pi_A) + \Delta_Q + \dfrac{2u_{\pi_A}^\top W^\top e_A}{n},
\end{equation}
where $u_{\pi_A} = \Omega(0_{p_x}^\top, (A\pi_A)^\top )^\top$ and $|\Delta_Q| = o_p\left(\dfrac{1+\epsilon^2}{\sqrt{n}\log p}\right)$ when $\pi\in\mathcal{H}_Q(\epsilon)$ for any $\epsilon>0$ with $\mathcal{H}_Q(\epsilon)$ defined in (\ref{eq: H Q}). Therefore, the Q statistic $q_n(A)$ 
defined by (\ref{eq:def Qn}) satisfies the following:  \begin{enumerate}[(a)]
    \item When $\pi=0$, $q_n(A)\convp{0}$, and hence for any $\alpha\in(0,1)$,
    \[\Pr\left(q_n(A) > {\rm cv}_A(\alpha)\right)\to 0,\] as $n\to\infty$, where ${\rm cv}_A(\alpha)$ is defined by (\ref{eq: cv no hat}). 
    \item When $\|\pi_{A^*}\|_2 \gtrsim n^{-1/4}$, $q_n(A) - c\sqrt{\log p} \convp \infty$ for any absolute constant $c$, and hence for any $\alpha\in(0,1)$ and constant $\epsilon > 0$, 
    \[\inf_{\pi\in\mathcal{H}_Q(\epsilon)}\Pr\left(q_n(A) > {\rm cv}_A(\alpha)\right)\to 1, \quad \text{as} \quad n\to\infty.\] \label{th3b}
\end{enumerate}     
\end{theorem}  
\begin{remark}\label{rem: Q chi2}
    We briefly discuss the power performances here. For the Q statistic, Theorem \ref{thm:Q}(\ref{th3b}) shows it has asymptotic power 1 when (a) $\|\pi_{A^*}\|_2 \gtrsim n^{-1/4}$. To achieve asymptotic power 1, the $\chi^2$-type tests need (b) $ \sqrt{n}\|\pi_{A^*}\|_2 / p_z^{1/4} \to\infty$.  When $p_z\gtrsim n$, condition (b) implies condition (a), and thus the Q test requires a weaker condition to achieve power 1 compared to the $\chi^2$-type tests. Hence, the asymptotically zero Q statistic guarantees higher asymptotic power than the $\chi^2$-type tests when $p_z \gtrsim n$. Here, we emphasize again that our test is feasible when $p > n$ with $p_x\to\infty$,  while the $\chi^2$-type tests break down. 
\end{remark}
 \par Under the alternative with $\pi\neq 0$, we still need the sparsity condition Assumption \ref{as:asym}(i) to prove the consistency of $\hat\pi_A$ in (\ref{eq: Lasso pi}), which is required to establish the asymptotic properties of the test statistic $q_n(A)$. Under this particular sparsity condition, the conditions required for Q statistic $q_n(A)$ achieving asymptotic power 1 are not weaker than those required for the M test. Thus, the power enhancement by Q statistic $q_n(A)$ compared to the original M test is not visible from the theoretical point of view. 
 \par Nevertheless, the power enhancement procedure is still favorable in practice. As mentioned in the paragraph right after Remark \ref{rem: R neq 1}, in practice, there can be many locally invalid IVs with small $|\pi_j|$.  Our numerical studies in Section \ref{sec:simu} show that power enhancement is evident for many locally invalid IVs, with the type I error almost unaffected. 
 
 Practitioners can easily implement our test with a high-dimensional dataset\footnote{The \texttt{R} code for implementing the above method is available at \url{https://github.com/ZiweiMEI/PMtest}.}. The steps for the PM test are summarized in Algorithm \ref{algorithm1}.
 
\begin{algorithm}[H]
\footnotesize
\caption{\label{algorithm1} Power-enhanced Maximum (PM) test}
\hspace*{0.01in} 
\begin{algorithmic}[1]
\State Estimate the reduced-form model parameters in \eqref{eq: reduced form} using Lasso by  (\ref{Theta}) and (\ref{theta}).
\State    Get the debiased estimator $\hbeta_A$ in \eqref{eq:hatbetaA}, following the procedure in Section \ref{subsec:validity}.
\State    Regress the ``residual", $Y-D\hbeta_A$, against $X$ and $Z$, using Lasso as in \eqref{eq: Lasso pi}.
\State    Get the debiased $\tilde \pi_A$ as in \eqref{eq:DB}.
\State    Compute $\hat{\rm V}_A$ in (\ref{eq:hatVA}) and the M statistic $M_n(A)$ in (\ref{eq: def Mn}).
\State    Compute the critical value ${\rm cv}_A(\alpha)$ in \eqref{eq: cv no hat} by simulating $\eta\sim N(0,\hat{\rm V}_A)$.
\State    Construct the debiased quadratic form $\hQA(\pi_A)$ as in \eqref{eq: bias correction QA}.
\State    Compute the Q statistic $q_n(A)$ defined by \eqref{eq:def Qn}.
\State    Perform the PM test. Reject the null hypothesis if $M_n(A)\vee q_n(A) > {\rm cv}_A(\alpha)$. 
\end{algorithmic}
\end{algorithm}

\section{Empirical Example}
\label{sec:emp}
To illustrate the usefulness of the proposed test with high-dimensional data, we revisit the empirical analysis of the effect of trade on economic growth (\citealp{frankel1999does}, FR99 hereafter). \cite{fan2018nonparametric} searched for instruments from all geographical variables following the celebrated gravity theory of trade. In this paper, we update all data to 2018 and expand the set of IVs from \citet{fan2018nonparametric} to include potentially invalid IVs from World Bank economic data. Following the literature, the dependent variable $Y$ is the logarithm of GDP. There are $n=159$ countries, and $p=p_z+p_x=58$, which includes (1) the constructed trade $\widehat{T}$ proposed by FR99 under the guidance of the gravity theory of trade, (2) the logarithms of population $X_1$ and land area describing the sizes of the countries $X_2$ and (3) other covariates and candidate IVs concerning geographical characteristics, energy, the environment and natural resources, and business activities\footnote{$\hat{T}$, $X_1$ and $X_2$ are instruments and covariates that have been widely recognized in the literature since FR99. To make better comparisons to the literature, we do not penalize them in the Lasso problems, following the suggestions of \cite{belloni2014}.}. The dependent variable, the endogenous variable, the original FR99 covariates, and a subset of the baseline instruments used in \cite{fan2018nonparametric}, together with three additional and possibly invalid IVs, are summarized in Table 
\ref{tab:data} of the supplement's Section \ref{app:empt1}. We perform overidentification tests using this (sub)set of IVs.
\begin{table}[h]
\centering
\caption{P-values of different tests.}
\label{tab:empirical}
\begin{tabular}{cccc}
\hline\hline 
Instrument Sets&  MCD & M & PM   \\
\hline 
$\{Z_1,Z_2,\cdots,Z_{16}\}$ & 0.062 & 0.029 & 0.000   \\
$\{Z_1,Z_2,\cdots,Z_{13}\}$ & 0.317 & 0.275 & 0.275  \\
\hline \hline 
\end{tabular} 
\end{table}

We standardize the data so that all variables have zero sample mean and unit standard deviation, under which the weighting matrix $A$ is the identity matrix.  Table \ref{tab:empirical} shows the p-values of different tests performed on the real data. We first test the correct specifications of all 16 instruments in Table \ref{tab:data} and expect the null hypothesis to be rejected since at least some of the instruments, namely $Z_{14}$ (air pollution), $Z_{15}$ (access to electricity) and $Z_{16}$ (business environment), are likely to have a direct effect on economic growth. The variable dimensions in this case are $p_x=42$ and $p_z=16$. We can see that the M test and PM test reject the null hypothesis at the 5\% and 1\% levels, respectively, while MCD fails to reject the validity of IVs at the 5\% level. 
\par Next, we test whether a previously studied subset of IVs is valid. This application shows that empirical researchers can also use our method to test whether a subset of IVs is valid. Here, we select the subset of IVs used in \citet{fan2018nonparametric}, including $Z_1,Z_2,\cdots,Z_{13}$, as displayed in Table \ref{tab:data}, and treat the other three instruments as covariates. Therefore, the variable dimensions are now $p_x=45$ and $p_z=13$. All the considered tests do not reject the null hypothesis, meaning there is no evidence that this subset of instruments is invalid.  
  \par The takeaway from this empirical exercise is that practitioners should be cautious in the interpretation of a failure to reject the null hypothesis by existing overidentification tests when many covariates and/or instruments are present. Using tests with low power would result in further difficulty in the estimation and inference of the endogenous treatment effect. Our proposed test improves the power in the high-dimensional IV model with potentially invalid instruments; hence, it is recommended in a data-rich environment to detect invalid instruments.

\section{Conclusion}\label{sec:conc}
In this paper, we develop a new test on overidentifying restriction for linear IV models with high-dimensional covariates and/or IVs. This test {allows for $p>n$ and $p_x\to\infty$}, and is robust to heteroskedasticity. We show that, by utilizing a sparse model structure, our PM test has better power than the $\chi^2$-type tests even when {$p<n$ and $p\to\infty$, under which all tests under discussion are feasible.} As high-dimensional data become more common in observational studies, the PM test should have many applications in detecting instrument misspecifications. From a technical perspective, this paper extends the inference of maximum and $L_2$ norms to heteroskedastic errors, and shows its applicability to triangular systems such as the linear IV regression model.  

\bibliographystyle{apalike}

\bibliography{ValidityTest}

\clearpage
\setcounter{section}{0}
\renewcommand\thesection{\Alph{section}}
\setcounter{page}{1}
\setcounter{footnote}{0}
\begin{center}
{\large\bf SUPPLEMENTARY MATERIAL OF ``A HETEROSKEDASTICITY-ROBUST OVERIDENTIFYING RESTRICTION TEST WITH HIGH-DIMENSIONAL COVARIATES''}
\end{center}
The Appendices include the following parts: Section \ref{app:a} provides additional discussions complementary to the theory and the empirical example in the main text. Section \ref{sec:simu} collects simulation results. Section \ref{app:proof} contains all technical proofs. 
\vspace{-1em}
\newcounter{counter}[section]
\setcounter{table}{0}
\renewcommand{\thetable}{\thesection\arabic{table}}
\setcounter{equation}{0}
\renewcommand\theequation{\thesection\arabic{equation}}
\setcounter{figure}{0}
\renewcommand\thefigure{\thesection\arabic{figure}}
\setcounter{theorem}{0}
\setcounter{assumption}{0}
\setcounter{lemma}{0}
\setcounter{remark}{0}
\setcounter{corollary}{0}
\setcounter{proposition}{0}
\renewcommand\thetheorem{\thesection\arabic{theorem}}
\renewcommand\thelemma{\thesection\arabic{lemma}}
\renewcommand\theremark{\thesection\arabic{remark}}
\renewcommand\thecorollary{\thesection\arabic{corollary}}
\renewcommand\theproposition{\thesection\arabic{proposition}}
\renewcommand\theassumption{\thesection\arabic{assumption}}

\setcounter{subsection}{0}
\setcounter{table}{0}
\renewcommand{\thetable}{\thesection\arabic{table}}
\setcounter{equation}{0}
\renewcommand\theequation{\thesection\arabic{equation}}
\setcounter{figure}{0}
\renewcommand\thefigure{\thesection\arabic{figure}}
\setcounter{theorem}{0}
\setcounter{assumption}{0}
\setcounter{lemma}{0}
\setcounter{remark}{0}
\setcounter{corollary}{0}
\setcounter{proposition}{0}

\section{Additional Discussions and Details}\label{app:a}


\subsection{Relation between $\pi$ and $\pi_{A}$}
\label{app: pi piA}
For simplicity of discussion, we do not distinguish $A$ from $A^*$ in this section, and we use $A^*$ to be consistent with the alternative set (\ref{eq: HAstar}).
\par As discussed in Section \ref{sec:overid} of the paper, the true $\pi$ is of our interest while we work with the data scale-invariant version of $\pi_{A^*}$. It is thus helpful to look into the relation between $\pi$ and the identified $\pi_{A^*}$ for a clearer picture of the alternative set $\mathcal{H}_{M}(t)$ defined as (\ref{eq: HAstar}). Below are two  illustrative examples. Example \ref{eg:Astar0} shows that perfectly 
parallel $A^{*1/2}\pi$ and $A^{*1/2}\gamma$ cause a zero $\pi_{A^*}$ even if $\pi\neq 0$, and hence the M test has no power to detect invalid IVs. Other overidentifying restriction tests also have no power under similar conditions. Example \ref{eg: define Rstar} shows that when $A^{*1/2}\pi$ and $A^{*1/2}\gamma$ are far away from perfectly parallel, the alternative set defined by $\|A^{*1/2}\pi_{A^*}\|_\infty$ is similar to that defined by $\|{A}^{*1/2}\pi\|_\infty$ up to a square root term of sparsity indices.  

\begin{example}\label{eg:Astar0}
Recall that the discussions from (\ref{eq:Q}) to (\ref{eq: weighted pi A}) illustrate the absence of power when $A^{*1/2}\pi$ and $A^{*1/2}\gamma$ are perfectly parallel. A trivial example is $p_z=1$, under which the model is not overidentified. Another example with $p_z=2$ is given as follows. For simplicity, let $A^*=I_2$, $\pi = \rho_\pi(1,1)^\top$ and $\gamma = (1,1)^\top$. Here $|\rho_\pi|$ measures the strength of IV invalidity. Then it is easy to compute the $\pi_{A^*} = \pi - \gamma (\pi^\top\gamma/\gamma^\top \gamma) = 0$ even if $\rho_\pi\neq 0$.
\end{example}

\begin{example}\label{eg: define Rstar}Recall that ${\rm R}_{A^*}(\gamma,\pi)$ is defined as (\ref{eq: RAstar}). Following the arguments from (\ref{eq:Q}) to (\ref{eq: weighted pi A}), when $|{\rm R}_{A^*}(\gamma,\pi)|$ is  strictly bounded away from one, we have $\QAst(\pi) \asymp \QAst(\pi_{A^*})$. Hence, when $A^*$ is diagonal, 
\[\|A^{*1/2}\pi\|_\infty \lesssim \sqrt{\QAst(\pi)} \asymp \sqrt{\QAst(\pi_{A^*})} \lesssim \sqrt{s_\pi + s_\gamma}\|A^{*1/2}\pi_{A^*}\|_\infty,\]
where\footnote{The last inequality applies $\pi_{A^*} = \pi - (\beta_{A^*}-\beta)\gamma $, which implies $\|A^{*1/2}\pi_{A^*}\|_0=\|\pi_{A^*}\|_0\leq s_\pi+s_\gamma$ when $A^*$ is diagonal.} $s_\pi=\|\pi\|_0$ is the number of invalid IVs and $s_\gamma=\|\gamma\|_0$ is the number of relevant IVs. Consequently, $\pi\in\mathcal{H}_{M}(t)$ for some sufficiently large absolute constant $t$ whenever $\|A^{*1/2}\pi\|_\infty \geq t^\prime \sqrt{(s_\pi + s_\gamma) \log p_z/n}$ for some absolute constant $t^\prime$. Following symmetric arguments, we deduce that 
\[\|A^{*1/2}\pi_{A^*}\|_\infty\lesssim\sqrt{s_\pi+s_\gamma}\|A^{*1/2}\pi\|_\infty,\]
and hence any $\pi\in\mathcal{H}_{M}(t)$ satisfies $\|A^{*1/2}\pi\|_\infty \geq t^{\prime\prime} \sqrt{\log p_z/(n(s_\pi+s_\gamma))}$ for some $t^{\prime\prime}$. Thus, when $A^{*1/2}\pi$ and $A^{*1/2}\gamma$ are not perfectly parallel, the alternative set induced by $\|A^*\pi_{A^*}\|_\infty$ as (\ref{eq: HAstar}) is similar to that induced by $\|A^*\pi\|_\infty$ up to a square root term of sparsity indices. 
\end{example}

\par In summary, the alternative set induced by the data scale-invariant version of $\pi_{A^*}$ is appropriate for power analysis of the M test.

\subsection{Descriptive Statistics of the Empirical Example in Section \ref{sec:emp}}\label{app:empt1}

\begin{table}[H]
\centering 
\caption{Descriptive Statistics of the Raw Data.}
\label{tab:data}
\begin{tabular}{ccccccc}
\hline\hline
Notation & Variable Name                & Min    & Median   & Max        & Mean      & Std. Dev.  \\
\hline 
$Y$        & Log GDP                      & 7.463  & 10.422   & 12.026     & 10.184    & 1.102      \\
$D$        & Trade                        & 0.098  & 0.758    & 4.129      & 0.869     & 0.520      \\
$X_1$       & Log Population               & -3.037 & 1.472    & 6.674      & 1.355     & 1.830      \\
$X_2$       & Log Area                     & 5.193  & 11.958   & 16.611     & 11.685    & 2.312      \\
$Z_1$       & $\hat T$                       & 0.015  & 0.079    & 0.297      & 0.092     & 0.052      \\
$Z_2$       & Languages                    & 1.000  & 1.000    & 16.000     & 1.887     & 2.129      \\
$Z_3$       & Water Area                   & 0.000  & 2340.000 & 891163.000 & 25218.771 & 100518.984 \\
$Z_4$       & Land Boundaries              & 0.000  & 1881.000 & 22147.000  & 2819.987  & 3404.441   \\
$Z_5$       & \% Forest                    & 0.000  & 30.319   & 98.258     & 29.713    & 22.416     \\
$Z_6$       & Arable Land                  & 0.558  & 42.035   & 82.560     & 40.760    & 21.611     \\
$Z_7$       & Coast                        & 0.000  & 515.000  & 202080.000 & 4242.147  & 17399.583  \\
$Z_8$        & $Z_1\cdot Z_2$                    & 0.017  & 0.113    & 1.480      & 0.170     & 0.199      \\
$Z_9$       & $Z_1\cdot Z_3$                        & 0.000  & 201.263  & 87556.265  & 1872.710  & 8160.430   \\
$Z_{10}$       & $Z_1\cdot Z_4$                       & 0.000  & 184.863  & 2231.550   & 242.217   & 287.270    \\
$Z_{11}$      & $Z_1\cdot Z_5$                       & 0.000  & 1.946    & 20.573     & 2.686     & 3.025      \\
$Z_{12}$       & $Z_1\cdot Z_6$                        & 0.033  & 3.099    & 19.408     & 3.802     & 3.112      \\
$Z_{13}$       & $Z_1\cdot Z_7$                        & 0.000  & 39.891   & 19854.247  & 352.687   & 1675.864   \\
$Z_{14}$       & PM2.5                        & 5.861  & 22.252   & 99.734     & 27.868    & 19.436     \\
$Z_{15}$       & Access to Electricity        & 9.300  & 99.800   & 100.000    & 84.434    & 26.245     \\
$Z_{16}$       & Ease of Doing Business Index & 1.000  & 85.000   & 188.000    & 88.356    & 54.022 \\ 
\hline\hline
\end{tabular}
\flushleft{\footnotesize Data sources: the World Bank, CIA World Factbook, R package \texttt{naivereg}. }
\end{table}

\input{JBES_Simulation}

\clearpage

\input{proofs}

\clearpage

\end{document}

%% file: JBES_Simulation.tex
\section{Simulations}
\label{sec:simu}
\subsection{Setup}
\par The simulation DGP follows Model (\ref{eq:model}) in the main text. We focus on high-dimensional covariates where \((n,p_x)\in\{(150,50),(150,100),(300,150),(300,250),(500,350),(500,450)\}\). For each pair $(n,p_x)$, we set $p_z\in\{10,100\}$ to consider both low- and high-dimensional instrumental variables. The exogenous variables $W_{i\cdot}$ are independently generated by a multivariate Gaussian distribution with mean zero and covariance matrix $\Sigma=(|0.5|^{|i-j|})_{i,j\in[p]}$. We construct the error terms as follows:
\[\begin{aligned}
  e_{i} &= a_0 \cdot e_{i}^1 + \sqrt{1-a_0^2} \cdot e_{i}^0, \\
  \varepsilon_{i,D} &= 0.5\cdot e_{i}  + \sqrt{1-0.5^2}\cdot \varepsilon_{i,D}^0,
\end{aligned}\]
where $e_i^1|W_{i\cdot}\sim N(0,Z_{i1}^2)$, $\varepsilon_{i,D}^0$ and $e_{i}^0$ are i.i.d. $N(0,1)$ variables. We set $a_0 = 0$ for homoskedasticity and $a_0 = 2^{-1/4}$ for heteroskedasticity so that the R-square\footnote{According to Footnote 11 of \citet{bekker2015jackknife}, $R^2(e^2|Z) = \dfrac{\var[\expect(e^2|Z)]}{\var[\expect(e^2|Z)] + \expect[\var(e^2|Z)]}$. } for the regression of $e_{i}^2$ on the IVs equals 0.2.
\par We fix $\beta=1$. For each combination $(n,p_x,p_z)$, we set $\varphi=(1,0.5,\cdots,0.5^{s_\varphi-1}, 0^\top_{p_x-s_\varphi})^\top$ and $\psi=(1,0.6,\cdots,0.6^{s_\psi-1}, 0^\top_{p_x-s_\psi})^\top$. We consider two sparse settings of $\gamma:$
\begin{itemize}
\item The relevant IVs are all strong: $\gamma^{(1)} = (1_{s_\gamma}^\top,0^\top_{p_z-s_\gamma})^\top$; 
\item There is a mixture of strong and weak IVs: $\gamma^{(2)} = (1,0.8,0.8^2,\cdots,0.8^{s_\gamma-1},0^\top_{p_z-s_\gamma})^\top$. 
\end{itemize}
Throughout the simulation study, we set $s_\varphi=s_\psi=10$ and $s_\gamma = 7$. For IV validity, we first consider \[\pi=\pi^{(1)}=(\rho_\pi,0_{p_z-1}^\top)^\top,\] where only the first IV is invalid. To demonstrate the necessity of power enhancement, we also consider another setting of $\pi$, given as 
\[\pi^{(2)}:=\begin{cases}0.5\rho_\pi\cdot(1,-1,1,-1,0_6^\top)^\top,\ \ p_z=10,\\0.1\rho_\pi\cdot(1_{30}^\top,0_{70}^\top)^\top,\ \ p_z=100. 
\end{cases}\] 
When $p_z=100$, the vector $\pi^{(2)}$ induces a much larger number of invalid instruments with a smaller maximum norm compared to $\pi^{(1)}$. In this case, the Q test applying the $L_2$ norm is expected to be more powerful than the M test. We will see the benefit of power enhancement in the simulation results. We vary $\rho_\pi$ from -1 to 1.
\par The Lasso problems are solved by the \texttt{glmnet} R package. The tuning parameter is selected by cross-validation with the one-standard-error rule that is also favored in the current literature \citep{windmeijer2019use}. We use the \texttt{fastclime} \citep{pang2014fastclime}  package with the built-in parameters to obtain the CLIME estimator (\ref{eq:clime1}) and (\ref{eq:clime}). The package efficiently solves the problem using the parametric simplex method. In addition to 
the M test and PM test, we report the simulation results of the MCD test proposed by \citet{kolesar2018minimum} as a representative of $\chi^2$-type tests, which allows many covariates with the restriction $(p_x+p_z)/n \to c_p\in(0,1)$ as $n\to\infty$. 
\subsection{Summary of Simulation Results} 
 Tables and fugures of the empirical size and power from the simulation studies are available in Section \ref{app:powerCurves}.
Table \ref{tab:size} shows the empirical type I errors of different tests under $\rho_\pi=0$. The MCD test controls the type I error below or close to the nominal size. However, it is infeasible when $p_x+p_z>n$. In comparison, our M test and PM test are robust to high-dimensional covariates and instruments even when $p>n$. The most severe over-rejection occurs under $(n,p_x,p_z)=(300,250,10)$, which is no more than 0.03 off from the target rejection rate 5\%. In most cases, the rejection rate is close to the nominal size. The slight bias in Type I error is offset by substantial power gains compared to the MCD test, as in the figures  shown below. In addition, the empirical type I errors are similar between the M test and PM test, indicating that the power enhancement for the M test has almost no effect on the empirical size. In Section \ref{app:betaSimul}, we also show the simulation results of our proposed IQ estimator (\ref{eq:hatbetaA}) under the null hypothesis $\pi=0$. The IQ estimator has satisfactory performance in estimation and inference for $\beta$.
\par We then discuss the power. To fix ideas, we focus on the power curves from $(n,p_x)\in\{(150,50),(500,450)\}$ shown in Figures \ref{fig:powern150px50pz10}-\ref{fig:powern500px450pz100} in the discussions. Other power curves are also available in Section \ref{app:powerCurves}. Figures \ref{fig:powern150px50pz10} and \ref{fig:powern500px450pz10} show the results when $p_z=10$. With this small number of IVs, the M test and PM test have almost the same power. In addition, both tests are more powerful than the MCD test. The power improvement is more evident when $n=500$ and $p_x=450$, where $p$ is very close to the sample size $n$. 
\par  Figures \ref{fig:powern150px50pz100} and \ref{fig:powern500px450pz100} show the results when $p_z=100$. Given $p\geq n$, the $\chi^2$-type MCD test becomes infeasible; hence, the results of the MCD test are unavailable in these two figures. Again, the power curves of the M test and PM test are close when there is only one invalid IV ($\pi=\pi^{(1)}$), as shown in the first and third rows of the two figures. However, with 30 locally invalid instruments ($\pi=\pi^{(2)}$, the second and fourth rows), the M test is outperformed by the PM test. This result shows that our power enhancement procedure makes the test more powerful in some extreme cases with many locally invalid instruments without significant impacts on type I errors. Finally, the results are robust to the settings of $\gamma$ and heteroskedastic errors.  
\label{app:OmittedSimul} 
\subsection{Tables and Figures of Size and Power}
\input{tables/size.tex}
\input{figs/PowerCurves.tex}
\label{app:powerCurves}
\input{figs/OmittedSimul}

\clearpage
\subsection{Simulation Results of Estimation and Inference for $\beta$}
Tables \ref{tab:beta_homo} and \ref{tab:beta_hetero} show the simulation results of our proposed IQ estimator (\ref{eq:hatbetaA}) and its confidence interval under the null hypothesis $\pi=0$. The $100(1-\alpha)\%$ confidence interval is given by 
\[\left[\hat\beta_A-z_{\alpha/2}\sqrt{\dfrac{\hat{\rm V}_\beta}{n}},\hat\beta_A+z_{\alpha/2}\sqrt{\dfrac{\hat{\rm V}_\beta}{n}}\right]\]
where $z_{\alpha/2}$ is the $1-\alpha/2$ quantile of the standard normal distribution, and $\hat{\rm V}_\beta$ is defined in Theorem \ref{thm:betaAsympNorm}. The IQ estimator has satisfactory performance in estimation and inference for $\beta$ when all IVs are valid.
\label{app:betaSimul}
\input{tables/beta_homo.tex}
\input{tables/beta_hetero.tex}

%% file: tables/size.tex
\begin{table}[]
\caption{Type I Errors of the Overidentifying Restriction Tests under 5\% Level}
\label{tab:size}
\vspace{-1em}
\begin{center}
\begin{tabular}{ccc|rrr|rrr}
\hline \hline 
\multirow{2}{*}{$n$}   & \multirow{2}{*}{$p_x$}  & \multirow{2}{*}{$p_z$} & \multicolumn{3}{c|}{Homoskedasticity} & \multicolumn{3}{c}{Heteroskedasticity} \\ \cline{4-9}
                     &                      &                     & MCD        & M        & PM       & MCD         & M          & PM       \\
                  \hline  
\multicolumn{9}{c}{ $\gamma = \gamma^{(1)}$ }                                                            \\
\hline 
\multirow{4}{*}{150} & \multirow{2}{*}{50}  & 10                  & 0.022      & 0.073      & 0.073      & 0.023       & 0.042       & 0.042      \\
                     &                      & 100                 & NA         & 0.044      & 0.068      & NA          & 0.035       & 0.044      \\
                     & \multirow{2}{*}{100} & 10                  & 0.023      & 0.057      & 0.057      & 0.021       & 0.056       & 0.056      \\
                     &                      & 100                 & NA         & 0.038      & 0.061      & NA          & 0.023       & 0.028      \\
                     \hline 
\multirow{4}{*}{300} & \multirow{2}{*}{150} & 10                  & 0.025      & 0.056      & 0.056      & 0.032       & 0.044       & 0.044      \\
                     &                      & 100                 & 0.056      & 0.047      & 0.047      & 0.044       & 0.030        & 0.030       \\
                     & \multirow{2}{*}{250} & 10                  & 0.033      & 0.058      & 0.058      & 0.038       & 0.079       & 0.079      \\
                     &                      & 100                 & NA         & 0.039      & 0.039      & NA          & 0.038       & 0.038      \\
                     \hline 
\multirow{4}{*}{500} & \multirow{2}{*}{350} & 10                  & 0.035      & 0.052      & 0.052      & 0.028       & 0.052       & 0.052      \\
                     &                      & 100                 & 0.057      & 0.041      & 0.041      & 0.050        & 0.051       & 0.051      \\
                     & \multirow{2}{*}{450} & 10                  & 0.041      & 0.048      & 0.048      & 0.042       & 0.054       & 0.054      \\
                     &                      & 100                 & NA         & 0.038      & 0.038      & NA          & 0.037       & 0.037      \\
                     \hline  
\multicolumn{9}{c}{ $\gamma = \gamma^{(2)}$ }                                                           \\
\hline 
\multirow{4}{*}{150} & \multirow{2}{*}{50}  & 10                  & 0.023      & 0.068      & 0.069      & 0.020       & 0.045       & 0.045      \\
                     &                      & 100                 & NA         & 0.044      & 0.067      & NA          & 0.030       & 0.041      \\
                     & \multirow{2}{*}{100} & 10                  & 0.022      & 0.05       & 0.05       & 0.023       & 0.061       & 0.061      \\
                     &                      & 100                 & NA         & 0.039      & 0.06       & NA          & 0.023       & 0.028      \\
                     \hline 
\multirow{4}{*}{300} & \multirow{2}{*}{150} & 10                  & 0.029      & 0.057      & 0.057      & 0.028       & 0.039       & 0.039      \\
                     &                      & 100                 & 0.057      & 0.044      & 0.044      & 0.041       & 0.030       & 0.030      \\
                     & \multirow{2}{*}{250} & 10                  & 0.031      & 0.056      & 0.056      & 0.035       & 0.056       & 0.056      \\
                     &                      & 100                 & NA         & 0.040      & 0.040      & NA          & 0.039       & 0.039      \\
                     \hline 
\multirow{4}{*}{500} & \multirow{2}{*}{350} & 10                  & 0.036      & 0.053      & 0.053      & 0.026       & 0.044       & 0.044      \\
                     &                      & 100                 & 0.055      & 0.041      & 0.041      & 0.047       & 0.054       & 0.054      \\
                     & \multirow{2}{*}{450} & 10                  & 0.041      & 0.047      & 0.047      & 0.041       & 0.051       & 0.051      \\
                     &                      & 100                 & NA         & 0.039      & 0.039      & NA          & 0.039       & 0.039      \\
                     \hline \hline 
\end{tabular}
\end{center}
{\footnotesize Note: This table reports the type I errors over 1000 simulations. ``MCD'', ``M'', ``PM'' are the abbreviations of the modified Cragg–Donald test, the maximum test and the power-enhanced maximum test, respectively. ``NA'' means ``not available''. }
\end{table}

%% file: figs/PowerCurves.tex
\begin{figure}[ht]
\centering
\includegraphics[scale=0.5]{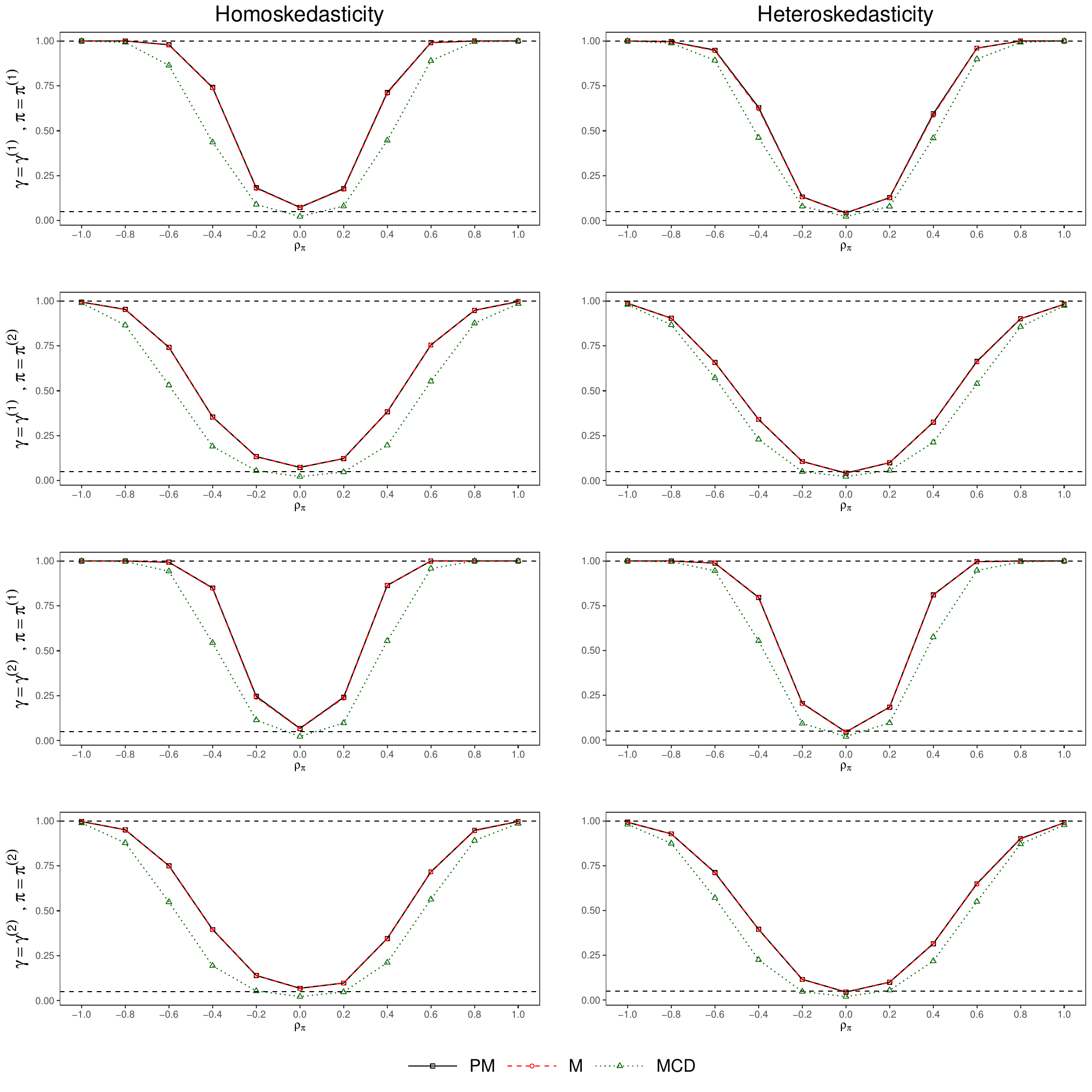}
\setlength{\abovecaptionskip}{0pt}
\caption{Power of tests with $(n,p_x,p_z)=(150,50,10)$ under 5\% level over 1000 simulations. ``MCD" represents the modified Cragg–Donald test by \citet{kolesar2018minimum}.  The nominal size 0.05 and power 1 are shown by the horizontal dashed lines. }
\label{fig:powern150px50pz10}
\end{figure}

\begin{figure}[!ht]
\centering
\includegraphics[scale=0.5]{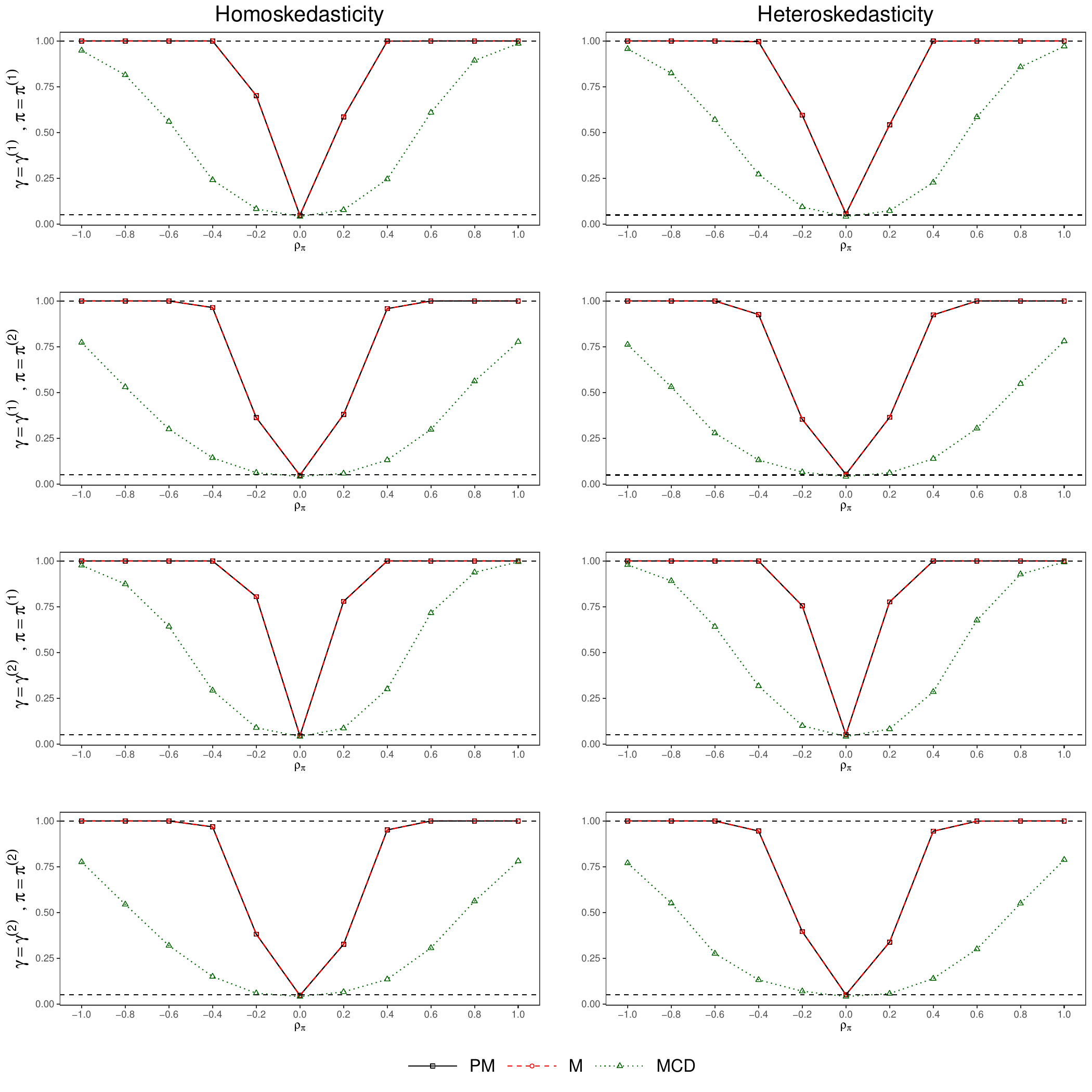}
\setlength{\abovecaptionskip}{0pt}
\caption{Power of tests with $(n,p_x,p_z)=(500,450,10)$ under 5\% level over 1000 simulations. ``MCD" represents the modified Cragg–Donald test by \citet{kolesar2018minimum}.  The nominal size 0.05 and power 1 are shown by the horizontal dashed lines. }
\label{fig:powern500px450pz10}
\end{figure}

\begin{figure}[!ht]
\centering
\includegraphics[scale=0.5]{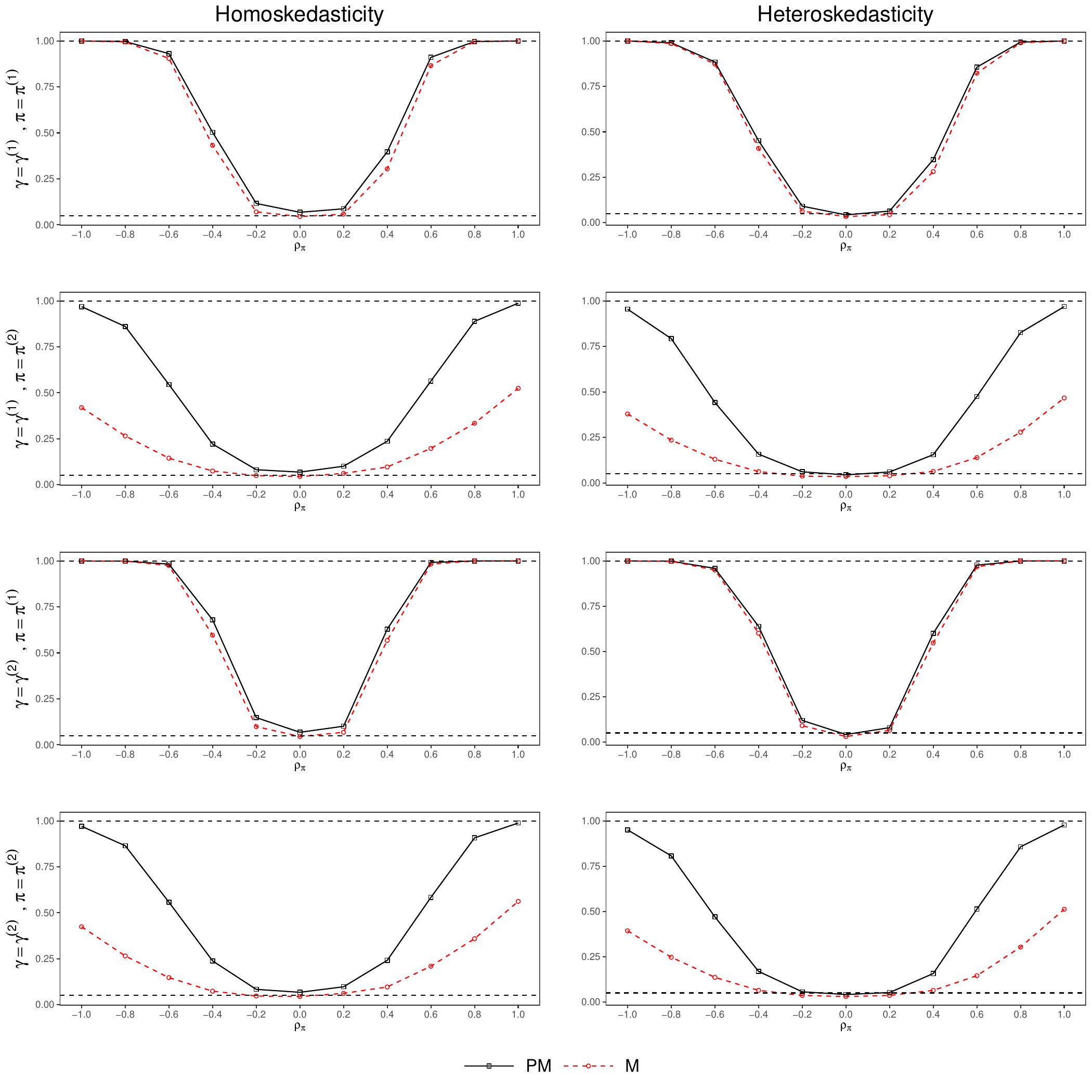}
\setlength{\abovecaptionskip}{0pt}
\caption{Power of tests with $(n,p_x,p_z)=(150,50,100)$ under 5\% level over 1000 simulations. The nominal size 0.05 and power 1 are shown by the horizontal dashed lines. }
\label{fig:powern150px50pz100}
\end{figure}
 
\begin{figure}[!ht]
\centering
\includegraphics[scale=0.5]{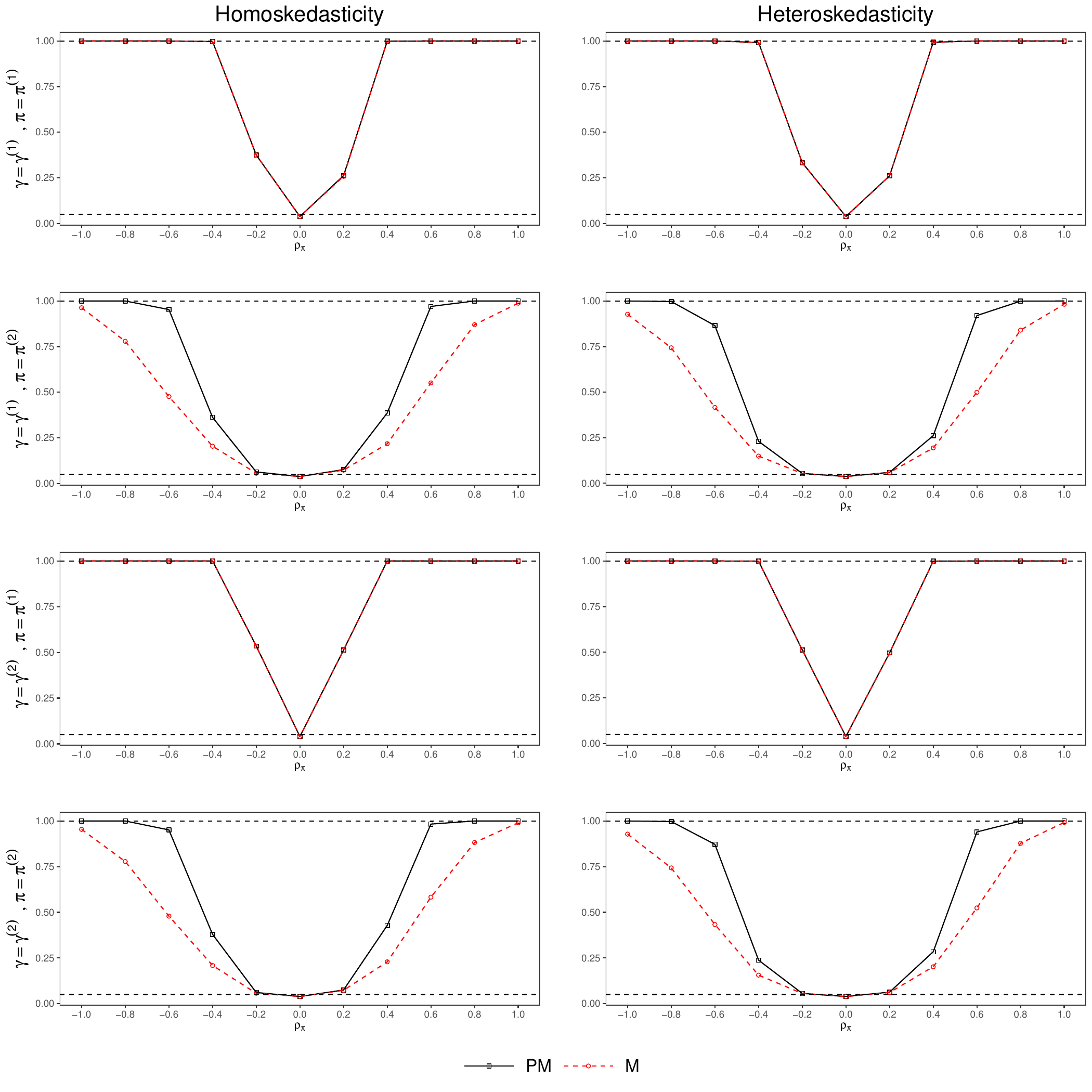}
\setlength{\abovecaptionskip}{0pt}
\caption{Power of tests with $(n,p_x,p_z)=(500,450,100)$ under 5\% level over 1000 simulations. The nominal size 0.05 and power 1 are shown by the horizontal dashed lines. }
\label{fig:powern500px450pz100}
\end{figure}
 

%% file: figs/OmittedSimul.tex
\begin{figure}[!ht]
\centering
\includegraphics[scale=0.5]{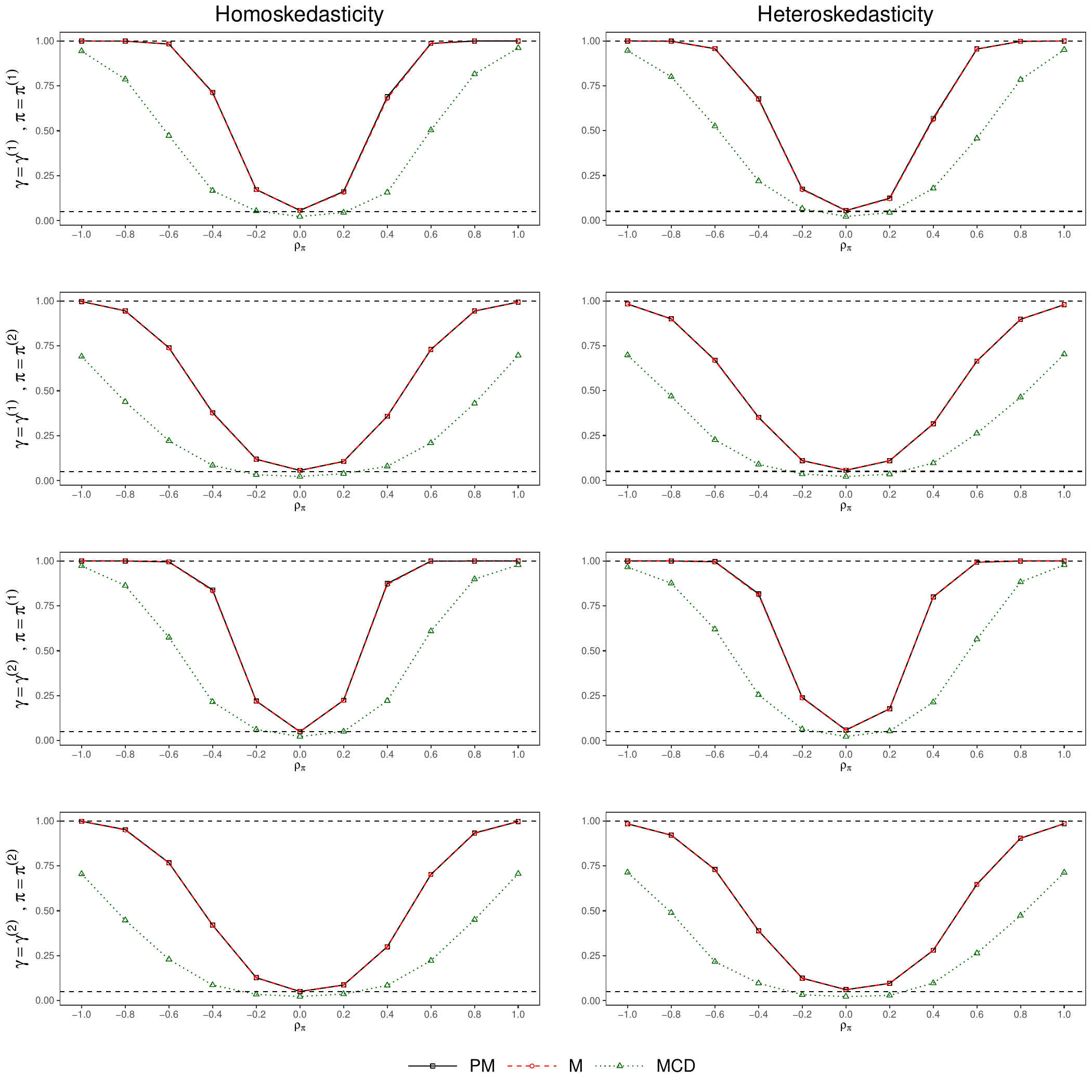}
\setlength{\abovecaptionskip}{0pt}
\caption{Power of tests with $(n,p_x,p_z)=(150,100,10)$ under 5\% level over 1000 simulations. ``MCD" represents the modified Cragg–Donald test by \citet{kolesar2018minimum}.  The nominal size 0.05 and power 1 are shown by the horizontal dashed lines. }
\label{fig:powern150px100pz10}
\end{figure}

\begin{figure}[!ht]
\centering
\includegraphics[scale=0.5]{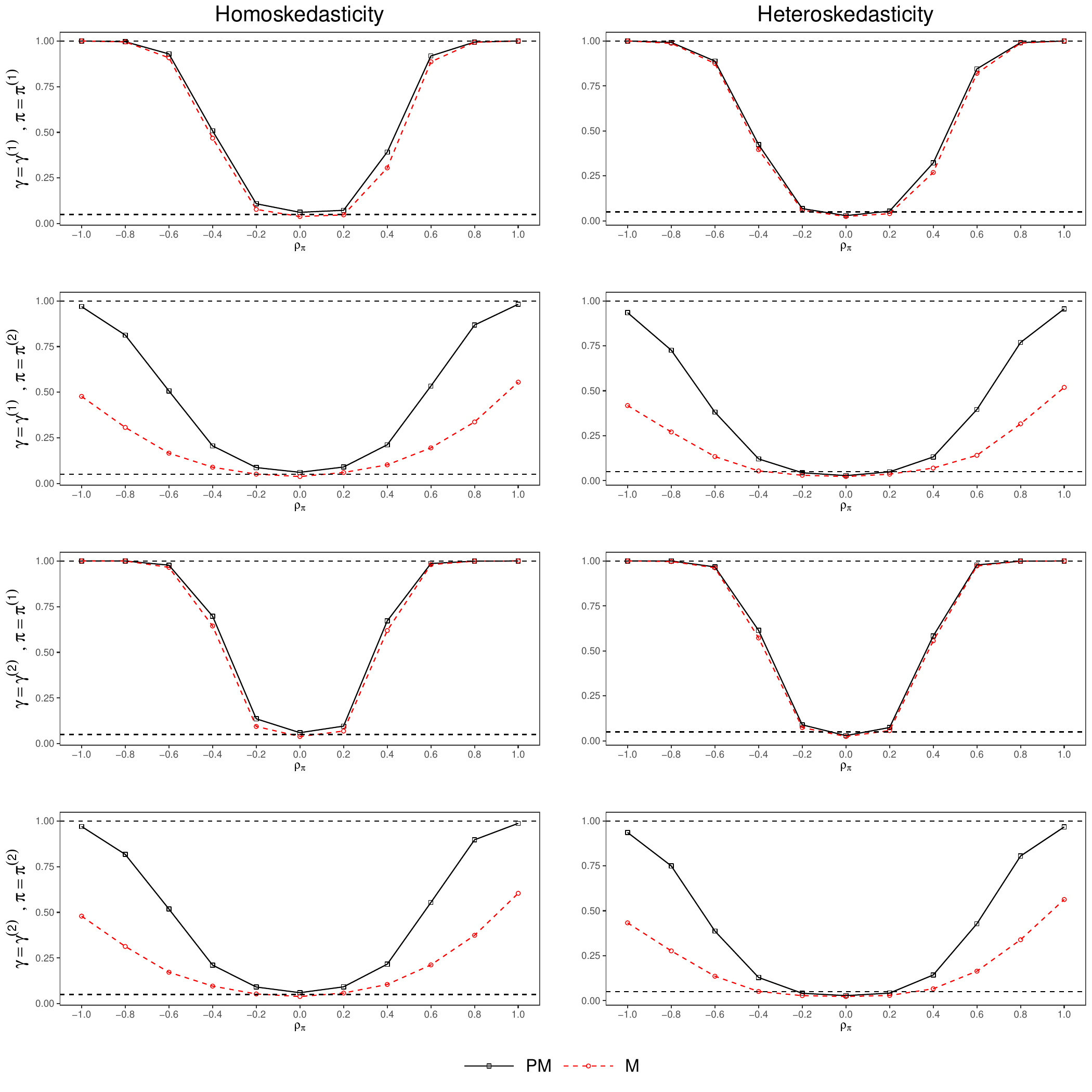}
\setlength{\abovecaptionskip}{0pt}
\caption{Power of tests with $(n,p_x,p_z)=(150,100,100)$ under 5\% level over 1000 simulations.  The nominal size 0.05 and power 1 are shown by the horizontal dashed lines. }
\label{fig:powern150px100pz100}
\end{figure}

\begin{figure}[!ht]
\centering
\includegraphics[scale=0.5]{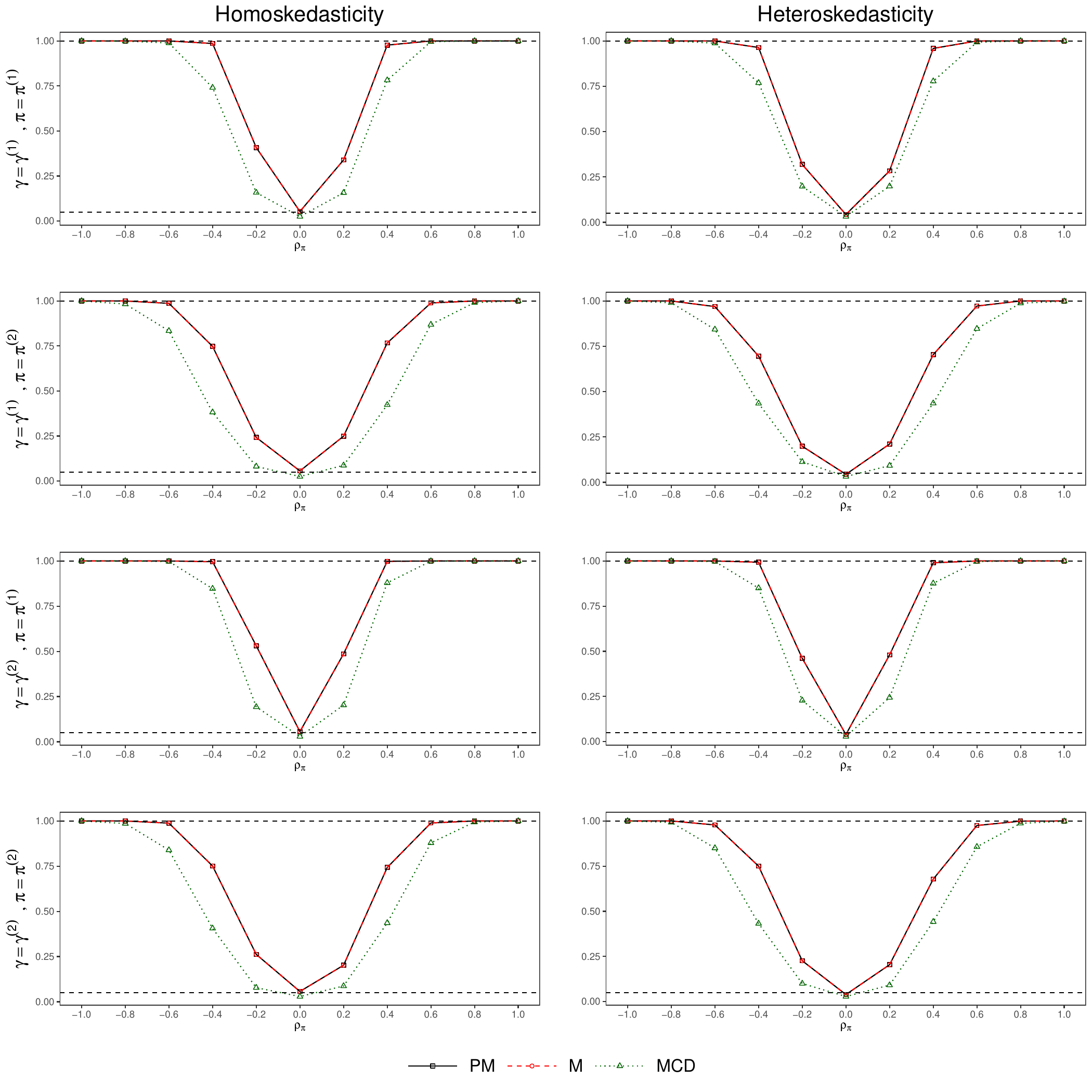}
\setlength{\abovecaptionskip}{0pt}
\caption{Power of tests with $(n,p_x,p_z)=(300,150,10)$ under 5\% level over 1000 simulations. ``MCD" represents the modified Cragg–Donald test by \citet{kolesar2018minimum}.  The nominal size 0.05 and power 1 are shown by the horizontal dashed lines. }
\label{fig:powern300px150pz10}
\end{figure}

\begin{figure}[!ht]
\centering
\includegraphics[scale=0.5]{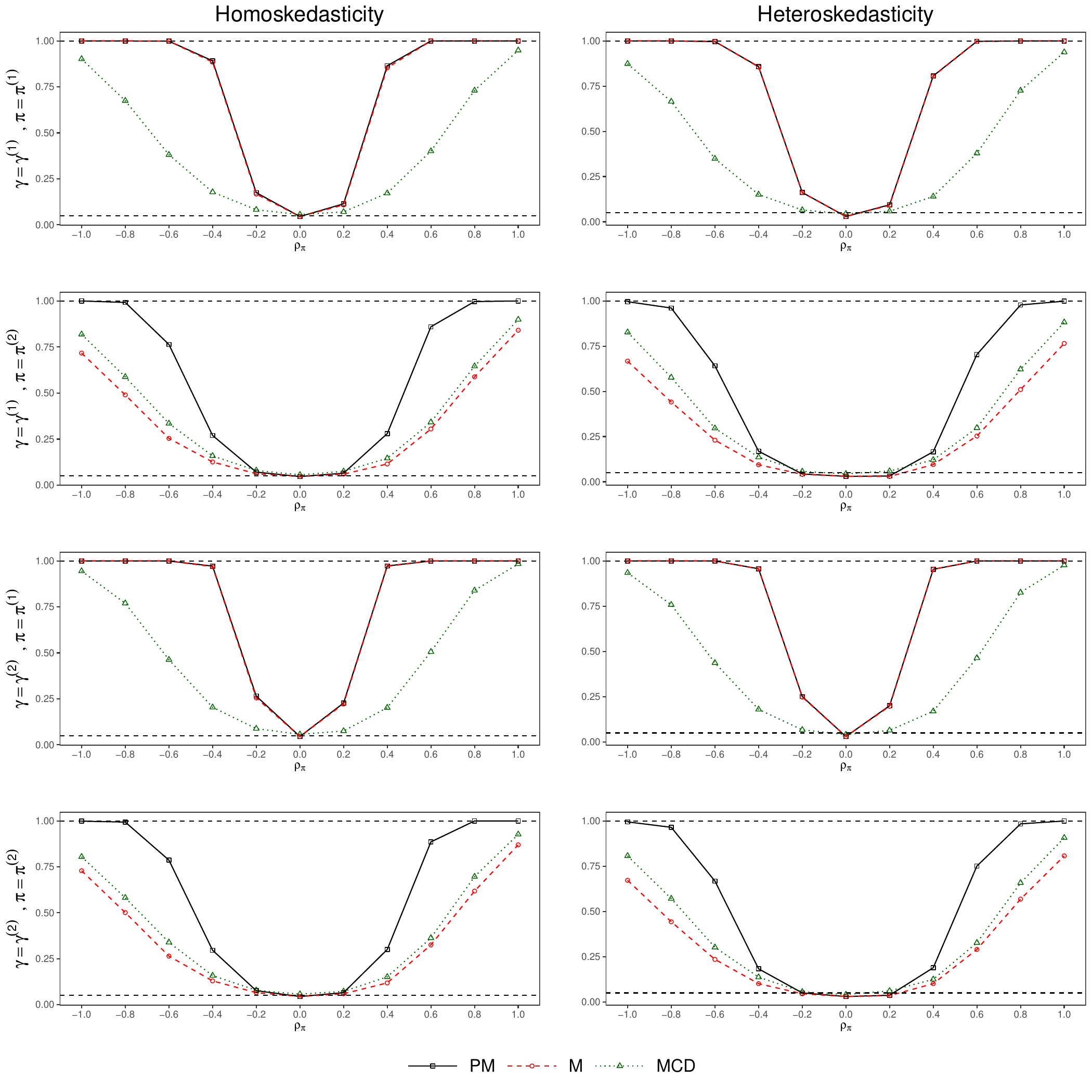}
\setlength{\abovecaptionskip}{0pt}
\caption{Power of tests with $(n,p_x,p_z)=(300,150,100)$ under 5\% level over 1000 simulations. ``MCD" represents the modified Cragg–Donald test by \citet{kolesar2018minimum}.  The nominal size 0.05 and power 1 are shown by the horizontal dashed lines. }
\label{fig:powern300px150pz100}
\end{figure}

\begin{figure}[!ht]
\centering
\includegraphics[scale=0.5]{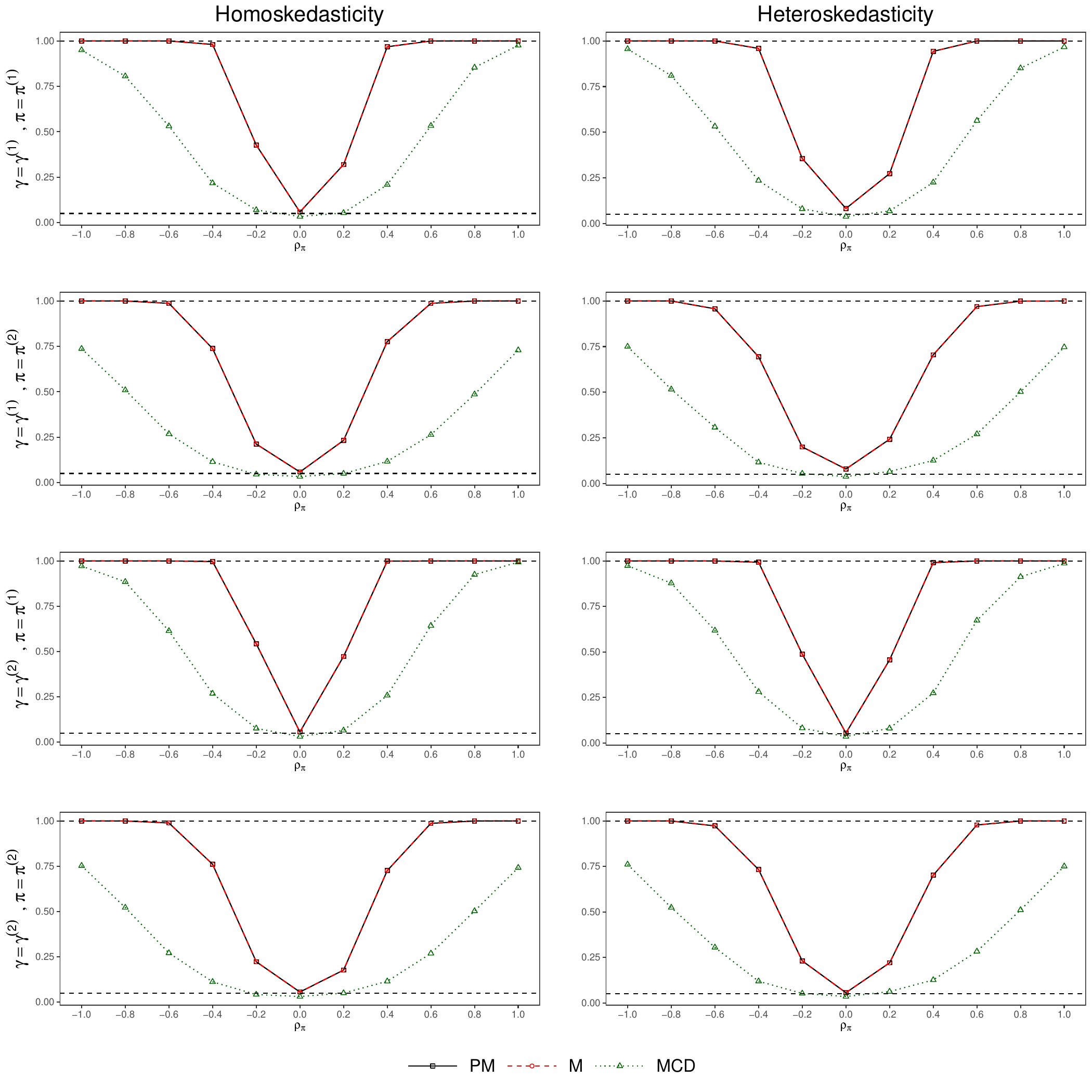}
\setlength{\abovecaptionskip}{0pt}
\caption{Power of tests with $(n,p_x,p_z)=(300,250,10)$ under 5\% level over 1000 simulations. ``MCD" represents the modified Cragg–Donald test by \citet{kolesar2018minimum}.  The nominal size 0.05 and power 1 are shown by the horizontal dashed lines. }
\label{fig:powern300px250pz10}
\end{figure}

\begin{figure}[!ht]
\centering
\includegraphics[scale=0.5]{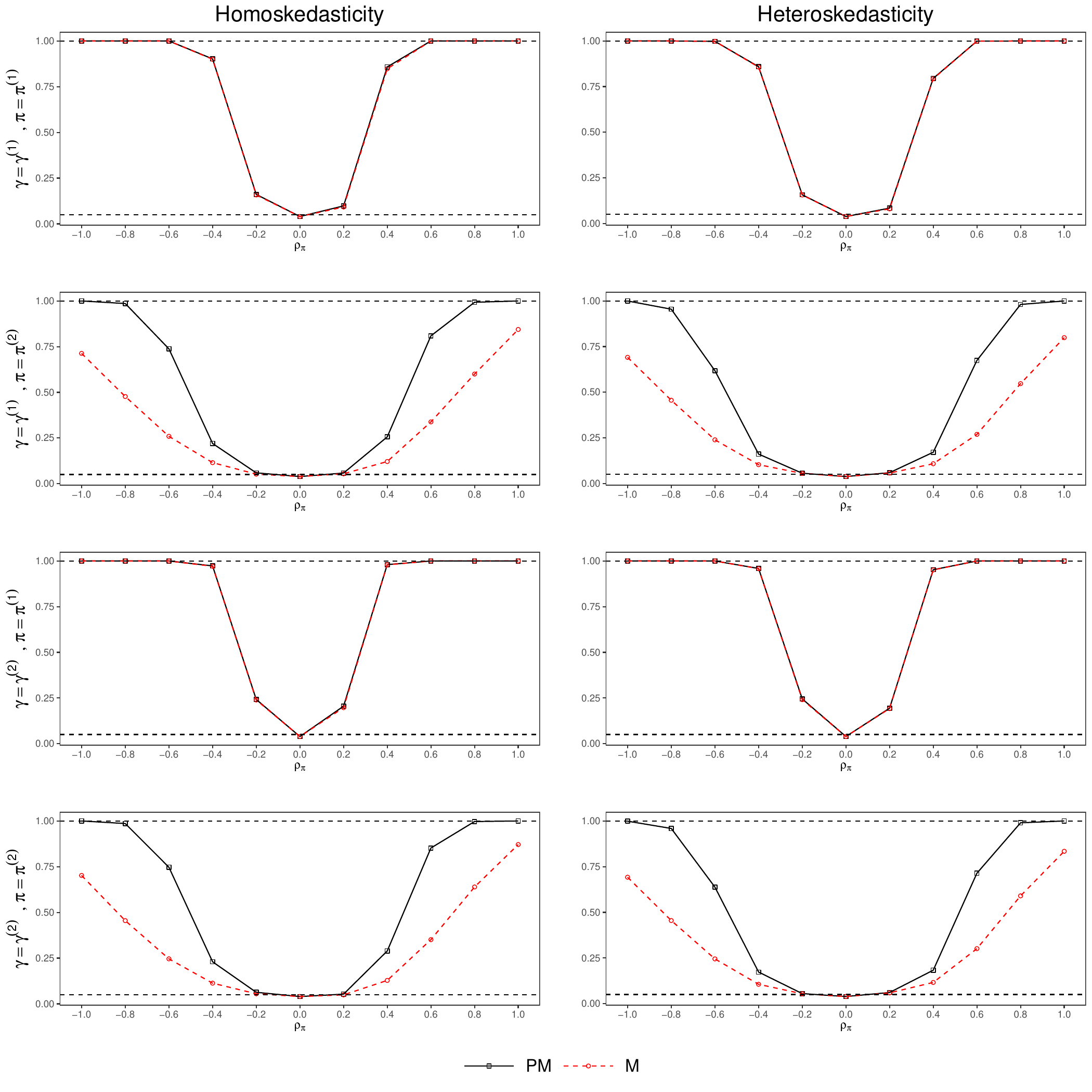}
\setlength{\abovecaptionskip}{0pt}
\caption{Power of tests with $(n,p_x,p_z)=(300,250,100)$ under 5\% level over 1000 simulations. ``MCD" represents the modified Cragg–Donald test by \citet{kolesar2018minimum}.  The nominal size 0.05 and power 1 are shown by the horizontal dashed lines. }
\label{fig:powern300px250pz100}
\end{figure} 

\begin{figure}[!ht]
\centering
\includegraphics[scale=0.5]{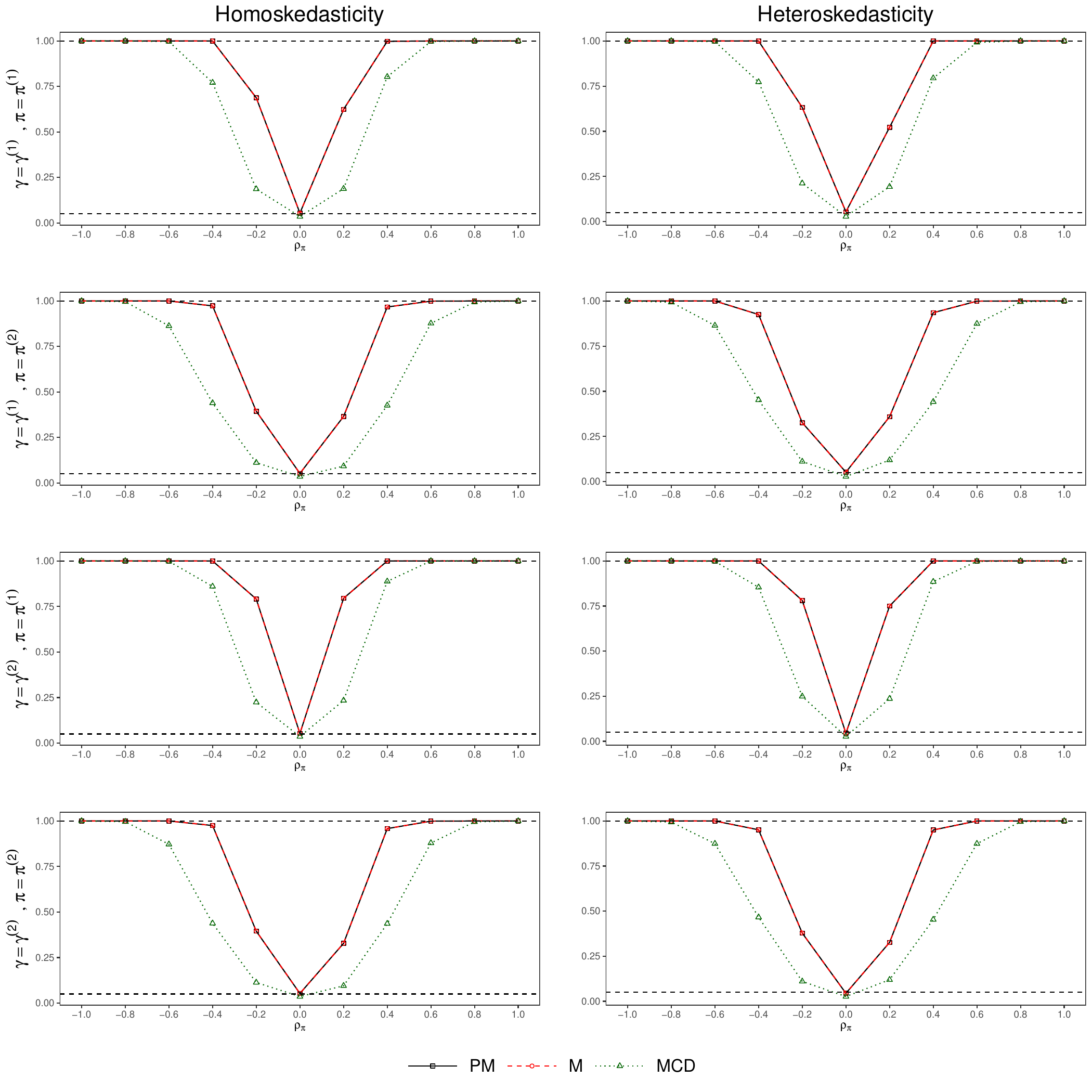}
\setlength{\abovecaptionskip}{0pt}
\caption{Power of tests with $(n,p_x,p_z)=(500,350,10)$ under 5\% level over 1000 simulations. ``MCD" represents the modified Cragg–Donald test by \citet{kolesar2018minimum}.  The nominal size 0.05 and power 1 are shown by the horizontal dashed lines. }
\label{fig:powern500px350pz10}
\end{figure}

\begin{figure}[!ht]
\centering
\includegraphics[scale=0.5]{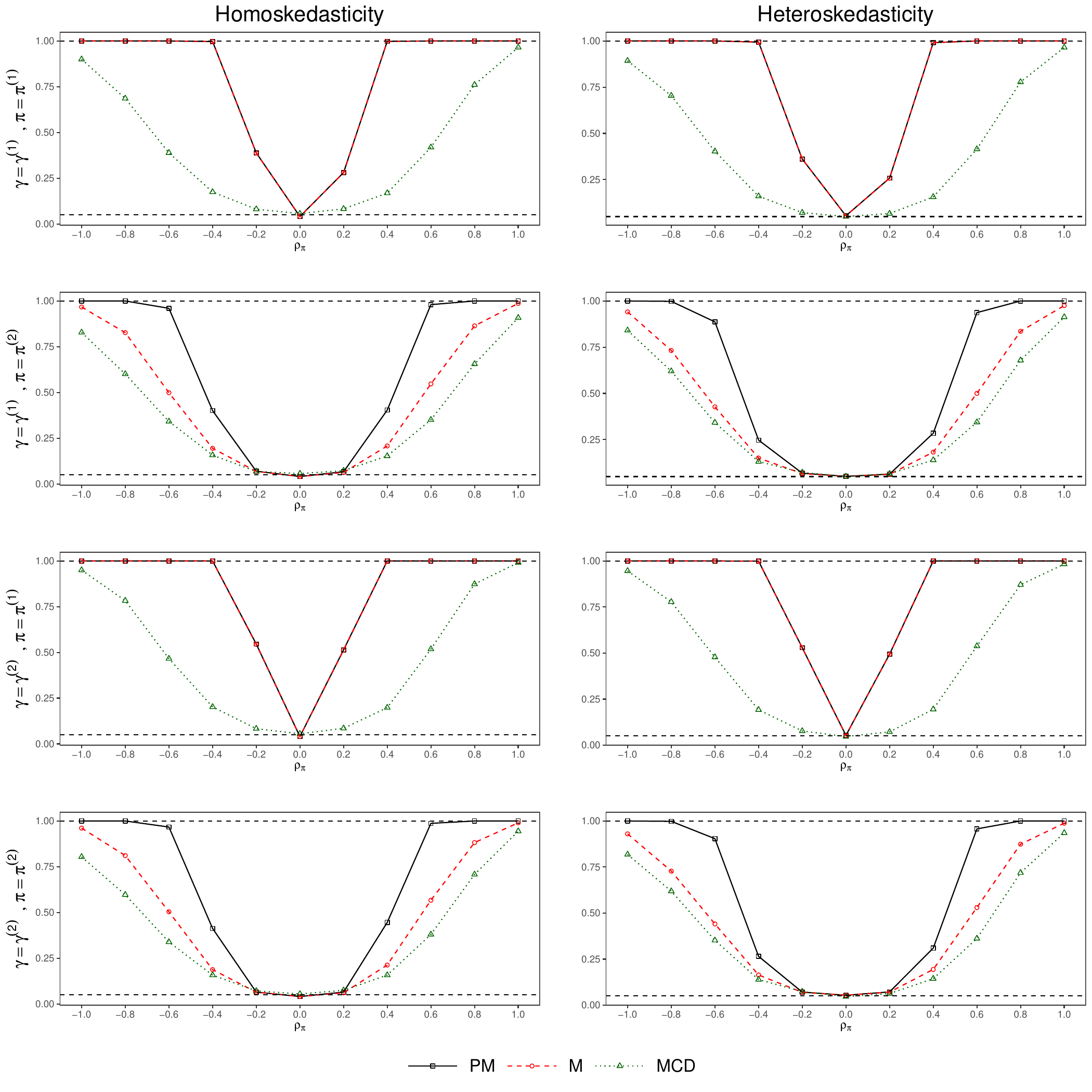}
\setlength{\abovecaptionskip}{0pt}
\caption{Power of tests with $(n,p_x,p_z)=(500,350,100)$ under 5\% level over 1000 simulations. ``MCD" represents the modified Cragg–Donald test by \citet{kolesar2018minimum}. The nominal size 0.05 and power 1 are shown by the horizontal dashed lines. }
\label{fig:powern500px350pz100}
\end{figure}

%% file: tables/beta_homo.tex
\begin{table}[!ht]
\caption{Estimation and Inference of Endogenous Effects under Homoskedasticity}
\label{tab:beta_homo}
\vspace{-1em}
\begin{center}
\begin{tabular}{ccc|rrr|rrr|rrr}
\hline\hline 
\multirow{2}{*}{$n$}   & \multirow{2}{*}{$p_x$}  & \multirow{2}{*}{$p_z$} & \multicolumn{3}{c|}{MAE} & \multicolumn{3}{c|}{Coverage} & \multicolumn{3}{c}{Length} \\ \cline{4-12}
                     &                      &                     & \multicolumn{1}{c}{IQ}     & LIML  & mbtsls &  \multicolumn{1}{c}{IQ}      & LIML    & mbtsls   &  \multicolumn{1}{c}{IQ}     & LIML   & mbtsls  \\
                     \hline \\
\multicolumn{12}{c}{ $\gamma = \gamma^{(1)}$ }                                                              \\
 \hline 
\multirow{4}{*}{150} & \multirow{2}{*}{50}  & 10                  & 0.022  & 0.021 & 0.021  & 0.942   & 0.937   & 0.940    & 0.101   & 0.099  & 0.099   \\
                     &                      & 100                 & 0.026  & NA    & NA     & 0.915   & NA      & NA       & 0.107   & NA     & NA      \\
                     & \multirow{2}{*}{100} & 10                  & 0.023  & 0.029 & 0.029  & 0.935   & 0.945   & 0.947    & 0.106   & 0.139  & 0.142   \\
                     &                      & 100                 & 0.027  & NA    & NA     & 0.895   & NA      & NA       & 0.108   & NA     & NA      \\
                      \hline 
\multirow{4}{*}{300} & \multirow{2}{*}{150} & 10                  & 0.015  & 0.016 & 0.016  & 0.935   & 0.949   & 0.952    & 0.070   & 0.080  & 0.080   \\
                     &                      & 100                 & 0.016  & 0.016 & 0.017  & 0.931   & 0.953   & 0.962    & 0.072   & 0.084  & 0.087   \\
                     & \multirow{2}{*}{250} & 10                  & 0.015  & 0.030 & 0.030  & 0.945   & 0.941   & 0.941    & 0.072   & 0.141  & 0.143   \\
                     &                      & 100                 & 0.016  & NA    & NA     & 0.929   & NA      & NA       & 0.073   & NA     & NA      \\
                      \hline 
\multirow{4}{*}{500} & \multirow{2}{*}{350} & 10                  & 0.011  & 0.017 & 0.017  & 0.937   & 0.950   & 0.952    & 0.053   & 0.080  & 0.081   \\
                     &                      & 100                 & 0.012  & 0.018 & 0.018  & 0.935   & 0.941   & 0.944    & 0.053   & 0.084  & 0.087   \\
                     & \multirow{2}{*}{450} & 10                  & 0.012  & 0.029 & 0.029  & 0.948   & 0.943   & 0.947    & 0.054   & 0.140  & 0.142   \\
                     &                      & 100                 & 0.011  & NA    & NA     & 0.953   & NA      & NA       & 0.054   & NA     & NA      \\
                      \hline  \\
\multicolumn{12}{c}{ $\gamma = \gamma^{(2)}$ }                                                                  \\
 \hline 
\multirow{4}{*}{150} & \multirow{2}{*}{50}  & 10                  & 0.036  & 0.036 & 0.036  & 0.924   & 0.943   & 0.943    & 0.163   & 0.174  & 0.175   \\
                     &                      & 100                 & 0.040  & NA    & NA     & 0.892   & NA      & NA       & 0.167   & NA     & NA      \\
                     & \multirow{2}{*}{100} & 10                  & 0.036  & 0.052 & 0.053  & 0.931   & 0.943   & 0.947    & 0.168   & 0.248  & 0.255   \\
                     &                      & 100                 & 0.040  & NA    & NA     & 0.900   & NA      & NA       & 0.166   & NA     & NA      \\
                      \hline 
\multirow{4}{*}{300} & \multirow{2}{*}{150} & 10                  & 0.024  & 0.028 & 0.028  & 0.936   & 0.942   & 0.946    & 0.113   & 0.141  & 0.142   \\
                     &                      & 100                 & 0.025  & 0.032 & 0.034  & 0.932   & 0.955   & 0.961    & 0.114   & 0.160  & 0.173   \\
                     & \multirow{2}{*}{250} & 10                  & 0.024  & 0.053 & 0.054  & 0.950   & 0.937   & 0.934    & 0.115   & 0.250  & 0.257   \\
                     &                      & 100                 & 0.025  & NA    & NA     & 0.933   & NA      & NA       & 0.115   & NA     & NA      \\
                      \hline 
\multirow{4}{*}{500} & \multirow{2}{*}{350} & 10                  & 0.018  & 0.029 & 0.029  & 0.942   & 0.951   & 0.949    & 0.086   & 0.141  & 0.142   \\
                     &                      & 100                 & 0.018  & 0.034 & 0.037  & 0.937   & 0.933   & 0.945    & 0.086   & 0.160  & 0.174   \\
                     & \multirow{2}{*}{450} & 10                  & 0.018  & 0.050 & 0.051  & 0.940   & 0.949   & 0.950    & 0.086   & 0.249  & 0.255   \\
                     &                      & 100                 & 0.018  & NA    & NA     & 0.957   & NA      & NA       & 0.086   & NA     & NA      \\ 
                      \hline  \hline 
\end{tabular}
\end{center}
{\footnotesize Note: The results come from the average of 1000 simulations. ``MAE'' denotes the mean absolute error. ``Coverage'' and ``Length'' are the empirical coverage rate and the average length of the 95\% confidence intervals, respectively. ``IQ'' represents the IQ estimator defined in \eqref{eq:hatbetaA}. ``LIML'' and ``mbtsls'' represent the LIML estimator and modified bias-corrected two stage least square estimator \citep{kolesar2015identification}, respectively. The standard errors of the latter two estimators are constructed by the minimum distance approach \citep{kolesar2018minimum}. ``NA'' means ``not available''. }
\end{table}

%% file: tables/beta_hetero.tex
\begin{table}[!ht]
\caption{Estimation and Inference of Endogenous Effects under Heteroskedasticity}
\label{tab:beta_hetero}
\vspace{-1em}
\begin{center}
\begin{tabular}{ccc|rrr|rrr|rrr}
\hline\hline 
\multirow{2}{*}{$n$}   & \multirow{2}{*}{$p_x$}  & \multirow{2}{*}{$p_z$}  & \multicolumn{3}{c|}{MAE} & \multicolumn{3}{c|}{Coverage} & \multicolumn{3}{c}{Length} \\ \cline{4-12}
                     &                      &                     &  \multicolumn{1}{c}{IQ}     & LIML  & mbtsls &  \multicolumn{1}{c}{IQ}      & LIML    & mbtsls   &  \multicolumn{1}{c}{IQ}       & LIML   & mbtsls \\
                     \hline \\
\multicolumn{12}{c}{ $\gamma = \gamma^{(1)}$ }                                                              \\
 \hline 
\multirow{4}{*}{150} & \multirow{2}{*}{50}  & 10                  & 0.025  & 0.021 & 0.021  & 0.942   & 0.921   & 0.927    & 0.117   & 0.099  & 0.100   \\
                     &                      & 100                 & 0.029  & NA    & NA     & 0.918   & NA      & NA       & 0.119   & NA     & NA      \\
                     & \multirow{2}{*}{100} & 10                  & 0.026  & 0.031 & 0.031  & 0.927   & 0.933   & 0.930    & 0.120   & 0.142  & 0.143   \\
                     &                      & 100                 & 0.029  & NA    & NA     & 0.918   & NA      & NA       & 0.122   & NA     & NA      \\
                     \hline
\multirow{4}{*}{300} & \multirow{2}{*}{150} & 10                  & 0.017  & 0.017 & 0.017  & 0.935   & 0.929   & 0.928    & 0.081   & 0.080  & 0.080   \\
                     &                      & 100                 & 0.019  & 0.018 & 0.018  & 0.921   & 0.933   & 0.932    & 0.082   & 0.084  & 0.087   \\
                     & \multirow{2}{*}{250} & 10                  & 0.019  & 0.029 & 0.030  & 0.932   & 0.943   & 0.943    & 0.083   & 0.141  & 0.143   \\
                     &                      & 100                 & 0.019  & NA    & NA     & 0.931   & NA      & NA       & 0.083   & NA     & NA      \\
                     \hline
\multirow{4}{*}{500} & \multirow{2}{*}{350} & 10                  & 0.013  & 0.017 & 0.017  & 0.950   & 0.937   & 0.939    & 0.062   & 0.080  & 0.080   \\
                     &                      & 100                 & 0.014  & 0.017 & 0.018  & 0.929   & 0.941   & 0.942    & 0.062   & 0.084  & 0.087   \\
                     & \multirow{2}{*}{450} & 10                  & 0.013  & 0.029 & 0.029  & 0.929   & 0.931   & 0.931    & 0.062   & 0.142  & 0.144   \\
                     &                      & 100                 & 0.014  & NA    & NA     & 0.933   & NA      & NA       & 0.062   & NA     & NA      \\
                     \hline \\ 
\multicolumn{12}{c}{ $\gamma = \gamma^{(2)}$ }                                                                 \\
\hline
\multirow{4}{*}{150} & \multirow{2}{*}{50}  & 10                  & 0.045  & 0.039 & 0.039  & 0.935   & 0.914   & 0.921    & 0.208   & 0.174  & 0.176   \\
                     &                      & 100                 & 0.047  & NA    & NA     & 0.902   & NA      & NA       & 0.205   & NA     & NA      \\
                     & \multirow{2}{*}{100} & 10                  & 0.046  & 0.057 & 0.058  & 0.924   & 0.921   & 0.919    & 0.211   & 0.253  & 0.259   \\
                     &                      & 100                 & 0.046  & NA    & NA     & 0.925   & NA      & NA       & 0.208   & NA     & NA      \\
                     \hline
\multirow{4}{*}{300} & \multirow{2}{*}{150} & 10                  & 0.031  & 0.032 & 0.032  & 0.933   & 0.912   & 0.912    & 0.146   & 0.140  & 0.140   \\
                     &                      & 100                 & 0.033  & 0.035 & 0.038  & 0.904   & 0.925   & 0.924    & 0.147   & 0.159  & 0.174   \\
                     & \multirow{2}{*}{250} & 10                  & 0.033  & 0.052 & 0.053  & 0.931   & 0.948   & 0.946    & 0.149   & 0.250  & 0.256   \\
                     &                      & 100                 & 0.032  & NA    & NA     & 0.920   & NA      & NA       & 0.148   & NA     & NA      \\
                     \hline
\multirow{4}{*}{500} & \multirow{2}{*}{350} & 10                  & 0.023  & 0.030 & 0.030  & 0.944   & 0.923   & 0.929    & 0.113   & 0.140  & 0.141   \\
                     &                      & 100                 & 0.025  & 0.034 & 0.037  & 0.927   & 0.932   & 0.936    & 0.113   & 0.160  & 0.174   \\
                     & \multirow{2}{*}{450} & 10                  & 0.024  & 0.053 & 0.054  & 0.938   & 0.936   & 0.938    & 0.113   & 0.252  & 0.258   \\
                     &                      & 100                 & 0.024  & NA    & NA     & 0.937   & NA      & NA       & 0.114   & NA     & NA      \\ 
                      \hline  \hline 
\end{tabular}
\end{center}
{\footnotesize Note: The results come from the average of 1000 simulations. ``MAE'' denotes the mean absolute error. ``Coverage'' and ``Length'' are the empirical coverage rate and the average length of the 95\% confidence intervals, respectively. ``IQ'' represents the IQ estimator defined in \eqref{eq:hatbetaA}. ``LIML'' and ``mbtsls'' represent the LIML estimator and modified bias-corrected two stage least square estimator \citep{kolesar2015identification}, respectively. The standard errors of the latter two estimators are constructed by the minimum distance approach \citep{kolesar2018minimum}. ``NA'' means ``not available''. }
\end{table}

%% file: proofs.tex
\newcommand{\hatpiZ}{\widehat{\pi}_Z}
\newcommand{\piZ}{\pi_Z}
\newcommand{\widecheckpiZ}{\widecheck{\pi}_Z}
\newcommand{\tildepiZ}{\widetilde{\pi}_Z}
\newcommand{\EV}{\expect[V]}
\newcommand{\Wi}{W_{i\cdot}}
\newcommand{\uhat}{\widehat{u}_3} 
\section{Proofs}\label{app:proof}
\par Throughout the proof, we use $C$ and $c$ to denote generic absolute constants that may vary from place to place. We first present some useful preliminary lemmas in Section \ref{app:lem}.  Section \ref{app:prooftheo} includes the proofs of the theoretical results in Section \ref{sec:model} of the main text. Firstly, some essential propositions about the initial Lasso estimators and test statistics are summarized in 
Section \ref{app:b1}. Secondly, we give the proof of Theorem \ref{thm:betaAsympNorm} in Section \ref{app:b22}. Section \ref{app:b3} includes the proofs of the main theoretical results of the proposed tests in Section \ref{sec:overid} of the main text. Firstly, some essential propositions are given in 
Section \ref{app:b31}. Secondly, we give the proofs of Theorems \ref{thm:M test} and  \ref{thm:Q} in Sections \ref{app:b32} and  \ref{app:b33}, respectively.

\input{lemmas.tex}

\subsection{Proofs of the Initial Estimator in Section \ref{sec:model}} \label{app:prooftheo}
\input{propositions.tex}
\subsubsection{Proof of Theorem \ref{thm:betaAsympNorm}}\label{app:b22}
Define $\hat{\rm U}_\beta := n^{-1}\sum_{i=1}^n(W_{i\cdot}^\top\hugamma)^2(\hat{\varepsilon}_{i,Y}-\hbeta_A\hat{\varepsilon}_{i,D})^2$ where $\hugamma$ is defined in Theorem \ref{thm:betaAsympNorm}. Thus, $\hat{\rm V}_\beta^{-1/2}=\hQA(\gamma)^{-1}\hat{\rm U}_\beta^{1/2}$. We will show a stronger result that is useful for power analysis: when $\pi\in\mathcal{H}_{A^*}(t)$ for any absolute constant $t$, the following asymptotic normality holds
\begin{equation} \label{eq: asymp norm beta A}
    \hat{\rm U}_\beta^{-1/2}\hQA(\gamma)\sqrt{n}(\hbeta_A-\beta_A) \convd N(0,1).  
\end{equation} 
under the conditions in Theorem \ref{thm:betaAsympNorm}. Define $u_\gamma=\Omega(0_{p_x}^\top,(A\gamma)^\top)^\top$.
The estimation error of $\hQA(\gamma)$ can be decomposed as 
\begin{equation*}
    \begin{aligned}
        \hQA(\gamma) - \QA(\gamma) &= \dfrac{2}{n}\hugamma^\top W^\top \varepsilon_D -2(\hu_\gamma\hSigma - (0_{p_x}^\top,\hgamma^\top A))\left(\begin{array}{cc}
\widehat{\psi} - \psi  \\
  \widehat{\gamma} - \gamma 
\end{array}\right) - \QA(\hgamma-\gamma) \\
&= \dfrac{2}{n} u_\gamma^\top W^\top \varepsilon_D + \dfrac{2}{n}(\hugamma - u_\gamma)^\top  W^\top \varepsilon_D  - 2(0_{p_x}^\top,\hgamma^\top A)(\hOmega\hSigma-I_{p})\left(\begin{array}{cc}
\widehat{\psi} - \psi  \\
  \widehat{\gamma} - \gamma 
\end{array}\right) - \QA(\hgamma-\gamma).
    \end{aligned}
\end{equation*}
By Proposition \ref{prop:LassoTheta}, Proposition \ref{prop:A and A*}.
\[\QA(\hgamma-\gamma) \lep \|\hgamma-\gamma\|_2^2 \lep \dfrac{s\log p}{n}.\]
Additionally, by Lemma \ref{lem:CLIME}, 
\[ \begin{aligned}
    \left|(\widehat{u}_\gamma^\top\widehat{\Sigma}-({\bf 0},\widehat{\gamma}^\top A))\left(\begin{array}{cc}
\widehat{\psi} - \psi  \\
  \widehat{\gamma} - \gamma 
\end{array}\right)\right| &\lep \|\hgamma\|_1 \|A\|_1 \|\hSigma\hOmega - I\|_\infty \left(\|\hpsi-\psi\|_1 + \|\hgamma-\gamma\|_1\right) \\
&\lep \left(\|\hgamma-\gamma\|_1 + \|\gamma\|_1\right)\dfrac{s_\omega m_\omega^{2-2q}\cdot s (\log p)^{(1-q)/2}}{n^{(1-q)/2}}\cdot\sqrt{\dfrac{s^2\log p}{n}} \\
&\lep  \dfrac{s_\omega m_\omega^{2-2q}\cdot s^{2} (\log p)^{(3-q)/2}}{n^{(3-q)/2}} + \dfrac{s_\omega m_\omega^{2-2q}\cdot s^{3/2} (\log p)^{1-q/2}}{n^{1-q/2}}\|\gamma\|_2\\
&\lep  \dfrac{m_\omega s\log p}{n} + \dfrac{s_\omega m_\omega^{2-2q}s^{1/2}(\log p)^{(1-q)/2}}{n^{(1-q)/2}}\|\gamma\|_2,
\end{aligned}\]
where the last inequality applies $s_\omega m_{\omega}^{1-2q} s\log p^{(1-q/2)} = o(n^{(1-q)/2})$ and $s\sqrt{\log p} = o(\sqrt{n})$ implied by Assumption \ref{as:asym} and Lemma \ref{lem:asymp}. 
 Recall that $\hugamma = \hOmega(0_{p_x}^\top,\hgamma^\top A)^\top$ and $\ugamma = \Omega(0_{p_x}^\top,\gamma^\top A)^\top$. Thus, 
\begin{equation}\label{eq:hu-ugamma}
\begin{aligned}
    \|\hugamma - u_\gamma\|_1 &\leq \|A\|_1\left(\|\hgamma - \gamma\|_1 \|\hOmega\|_1 + \|\hOmega-\Omega\|_1\cdot\|\gamma\|_1\right) \\
    &\lep  m_\omega\sqrt{\dfrac{ s^2 \log p}{n}} + \dfrac{s_\omega m_{\omega}^{2-2q}\cdot s^{1/2}(\log p)^{(1-q)/2}}{n^{(1-q)/2}}\|\gamma\|_2,    
\end{aligned}
\end{equation}
where the last inequality applies Proposition \ref{prop:LassoTheta} and Lemma \ref{lem:CLIME}. Then by (\ref{eq:hu-ugamma}) and (\ref{eq:WeBound})
\[ \begin{aligned}
    \left|\dfrac{2}{n}(\hugamma - u_\gamma)^\top  W^\top \varepsilon_D\right| &\lep \left\|\hugamma - u_\gamma\right\|_1 \cdot \left\|\dfrac{W^\top \varepsilon_D}{n}\right\|_\infty \\
    &\lep \dfrac{m_\omega s\log p}{n} + \dfrac{s_\omega m_{\omega}^{2-2q}\cdot s^{1/2}(\log p)^{1-q/2}}{n^{1-q/2}}\|\gamma\|_2.
\end{aligned}\]
Thus, 
\begin{equation}\label{eq:QAdecom_rate}
    \begin{aligned}
       | \hQA(\gamma) - \QA(\gamma)| \lep \left|\dfrac{2}{n}\ugamma^\top W^\top \varepsilon_D\right| + \dfrac{m_\omega s\log p}{n} + \dfrac{s_\omega m_{\omega}^{2-2q}\cdot s^{1/2}(\log p)^{1-q/2}}{n^{1-q/2}}\|\gamma\|_2.
    \end{aligned}
\end{equation}
The probability bound of the first term implied by (\ref{eq:WeBound}) is given as
\[\begin{aligned}\left|\dfrac{u_\gamma^\top W^\top \varepsilon_D}{n}\right| &\leq \|u_\gamma\|_1 \left\|\dfrac{W^\top \varepsilon_D}{n}\right\|_\infty    \\
&\leq \|\gamma\|_1\|A\|_1\|\Omega\|_1\sqrt{\dfrac{\log p}{n}} \lep \|\gamma\|_2\sqrt{\dfrac{m_\omega^2s\log p}{n}},
\end{aligned}\]
which, together with (\ref{eq:QAdecom_rate}), implies that 
\begin{equation}\label{eq:consistent_Qgamma}
\begin{aligned}
   \left| \dfrac{\hQA(\gamma)}{\QA(\gamma)} - 1\right|  = \dfrac{m_\omega s\log p}{n\QA(\gamma)} + \dfrac{1}{\|\gamma\|_2}\left(\dfrac{s_\omega m_\omega^{2-2q}\cdot s^{1/2}(\log p)^{1-q/2}}{n^{1-q/2}}+\sqrt{\dfrac{m_\omega^2s\log p}{n}}\right).
\end{aligned}
\end{equation} 
Thus, $\hQA(\gamma)/\QA(\gamma)\convp1$ by (\ref{eq:consistent_Qgamma}), Assumption \ref{as:asym}, and the fact that $\sqrt{\QAst(\gamma)}\asymp \|\gamma\|_2$. 
\par Besides, define $\hat u_\Gamma = \hat\Omega(0_{p_x}^\top, (A\hat\Gamma)^\top )^\top$, $u_\Gamma = \Omega(0_{p_x}^\top, (A\Gamma)^\top )^\top$ and $u_{\pi_A} = \Omega(0_{p_x}^\top, (A\pi_A)^\top)^\top$ where $\pi_A=\pi-\gamma(\beta_A-\beta)=\Gamma-\gamma\beta_A$. The estimation error of $\hIA(\gamma,\Gamma)$ can be decomposed as
\begin{equation*}
\begin{aligned}
\hIA(\gamma,\Gamma) - \IA(\gamma,\Gamma) &= \widehat{u}_\Gamma^\top\dfrac{1}{n}W^\top\varepsilon_D + \widehat{u}_\gamma^\top\dfrac{1}{n}W^\top\varepsilon_Y - (\widehat{u}_\gamma^\top\widehat{\Sigma}-({\bf 0},\widehat{\gamma}^\top A))\left(\begin{array}{cc}
\widehat{\Psi} - \Psi  \\
  \widehat{\Gamma} - \Gamma 
\end{array}\right)\\
&\ \ \  - (\widehat{u}_\Gamma^\top\widehat{\Sigma}-({\bf 0},\widehat{\Gamma}^\top A))\left(\begin{array}{cc}
\widehat{\psi} - \psi  \\
  \widehat{\gamma} - \gamma 
\end{array}\right) - (\widehat{\Gamma}-{\Gamma})^\top A(\widehat{\gamma}-{\gamma}),
\end{aligned}
\end{equation*}
and following the same procedures to derive (\ref{eq:QAdecom_rate}), we deduce that  
\begin{equation} \label{eq:IAdecom_rate}
    \begin{aligned}
    \hIA(\gamma,\Gamma) - \IA(\gamma,\Gamma) &= \uGamma^\top \dfrac{1}{n}W^\top\varepsilon_D + \ugamma^\top \dfrac{1}{n}W^\top\varepsilon_Y +  \\
    &\ \ \ \  O_p\left(\dfrac{m_\omega s\log p}{n} + \dfrac{s_\omega m_{\omega}^{2-2q}\cdot s^{1/2}(\log p)^{1-q/2}}{n^{1-q/2}}(\|\Gamma\|_2 + \|\gamma\|_2) \right) + \\ 
    &\ \ \ \  O_p\left(\dfrac{m_\omega s\log p}{n} + \dfrac{s_\omega m_{\omega}^{2-2q}\cdot s^{1/2}(\log p)^{1-q/2}}{n^{1-q/2}}(\|\pi\|_2 + \|\gamma\|_2) \right), 
\end{aligned}
\end{equation} 
where the last step applies $\|\Gamma\|_2\lesssim \|\pi\|_2 + |\beta|\cdot\|\gamma\|_2\lesssim\|\pi\|_2 +  \|\gamma\|_2$. 
Then by (\ref{eq:IQdecom}), (\ref{eq:QAdecom_rate}) and (\ref{eq:IAdecom_rate}) we deduce that 
\begin{equation}\label{eq: beta hat decom with pi}
\begin{aligned} 
&\ \hQA(\gamma)\sqrt{n}(\hbeta_A - \beta_A)\\ =& \sqrt{n}\cdot \left[\hIA(\gamma,\Gamma)-\IA(\gamma,\Gamma) - \beta_A(\hQA(\gamma)-\QA(\gamma)) \right] \\
=& \dfrac{u_{\pi_A}^\top W^\top \varepsilon_D + u_\gamma^\top W^\top e_A}{\sqrt{n}} +  O_p\left(\dfrac{m_\omega s\log p}{\sqrt{n}} + \dfrac{s_\omega m_{\omega}^{2-2q}\cdot s^{1/2}(\log p)^{1-q/2}}{n^{(1-q)/2}}(\|\gamma\|_2+\|\pi\|_2)\right),
\end{aligned}
\end{equation}
and by (\ref{eq:WeBound}),
\begin{equation}\label{eq:bound u pi A W e2}
    \left|u_{\pi_A}^\top \dfrac{W^\top \varepsilon_D}{\sqrt{n}}\right| \leq \|u_{\pi_A}\|_1 \cdot \left\|\dfrac{W^\top \varepsilon_D}{\sqrt{n}}\right\|_\infty \lep \|\pi\|_2\cdot\|A\|_1\|\Omega\|_1\cdot\sqrt{s\log p} \lep  m_\omega \sqrt{s\log p}\|\pi\|_2.
\end{equation}
Note that when $\pi\in\mathcal{H}_{A^*}(t)$, $\|\pi\|_2 \lesssim \sqrt{s\log p/n}$, which implies $\left|u_{\pi_A}^\top \dfrac{W^\top \varepsilon_D}{\sqrt{n}}\right|\lep \dfrac{m_\omega s\log p}{\sqrt{n}}$ together with (\ref{eq:bound u pi A W e2}). Thus, 
\begin{equation}
\begin{aligned}\label{eq:root_n_betahat_decom}
\hQA(\gamma)\sqrt{n}(\hbeta_A - \beta_A) &= \sqrt{n}\cdot \left[\hIA(\gamma,\Gamma)-\IA(\gamma,\Gamma) - \beta_A(\hQA(\gamma)-\QA(\gamma)) \right] \\
&= \dfrac{u_\gamma^\top W^\top e_A}{\sqrt{n}} +  O_p\left(\dfrac{m_\omega s\log p}{\sqrt{n}} + \dfrac{s_\omega m_{\omega}^{2-2q}\cdot s^{1/2}(\log p)^{1-q/2}}{n^{(1-q)/2}}\|\gamma\|_2 \right). 
\end{aligned}
\end{equation}
Define the asymptotic variance of the first term on the RHS of (\ref{eq:root_n_betahat_decom})
\begin{equation}\label{eq: def U beta}
\begin{aligned}
    {\rm U}_\beta  :=& \dfrac{1}{n}\sum_{i=1}^n \mathbb{E}\left[( u_{\gamma}^\top W_{i\cdot}e_{iA} )^2 | W\right] = \dfrac{1}{n}\sum_{i=1}^n \left(u_{\gamma}^\top W_{i\cdot}\right)^2\sigma_{iA}^2. 
\end{aligned}
\end{equation}
The remaining of this proof will show that 
\begin{enumerate}
    \item The rate of the asymptotic variance
    \begin{equation}\label{eq:VbetaRate}
        \sqrt{{\rm U}_\beta} \asymp_p  \|\gamma\|_2.
    \end{equation} This result, together with Assumption \ref{as:asym} and Lemma \ref{lem:asymp}, implies 
    \begin{equation}\label{eq:sqrtVbLower}
        \dfrac{m_\omega s\log p}{\sqrt{n}} + \dfrac{s_\omega m_{\omega}^{2-2q}\cdot s^{1/2}(\log p)^{1-q/2}}{n^{(1-q)/2}}\cdot \|\gamma\|_2 = o_p(\sqrt{\rm V_\beta}).
    \end{equation}
 
    In other words, the $O_p$ term in (\ref{eq:root_n_betahat_decom}) is dominated by the square root of asymptotic variance $\sqrt{{\rm U}_\beta}$.      
    \item $\dfrac{u_\gamma^\top W^\top e_A}{\sqrt{n{\rm U}_\beta}} \convd N(0,1)$, which together with (\ref{eq:root_n_betahat_decom}) and (\ref{eq:sqrtVbLower}) implies the asymptotic normality 
    \begin{equation}\label{eq:betaAsympNormTrueV}
        {\rm U}_\beta^{-1/2}\hQA(\gamma)\sqrt{n}(\hbeta_A - \beta_A) \convd N(0,1).
    \end{equation}
    \item $\hat{\rm U}_\beta / {\rm U}_\beta \convp 1$. And then (\ref{eq: asymp norm beta A}) follows by (\ref{eq:betaAsympNormTrueV}) and the Slutsky's Theorem. 
\end{enumerate} 
\par \underline{\textbf{Step 1}.} Show that ${\rm U}_\beta \asymp_p \|\gamma\|_2^2$. Recall that $\sigma_{iA^*}^2 = \mathbb{E}(e_{iA^*}^2|W)$ where $e_{iA^*} = \varepsilon_{i,Y}-\beta_{A^*}\varepsilon_{i,D}$. By the upper and lower bounds of conditional variances and covariances in Assumption \ref{as:error}, we deduce that 
\[\begin{aligned}
  2(1+\beta_{A^*}^2)\sigma_{\max}^2 \geq  \sigma_{iA^*}^2 &= \sigma_{i,Y}^2 + \beta_A^2\sigma_{i,D}^2 - 2\beta_{A^*}\sigma_{i,YD}\\
    &\geq  \sigma_{i,Y}^2 - \beta_{A^*}^2\sigma_{i,D}^2 + 2|\beta_{A^*}|\rho_\sigma\sigma_{i,Y}\sigma_{i,D} \\
    &\geq (1-\rho_\sigma)(\sigma_{i,Y}^2 + \beta_{A^*}^2\sigma_{i,D}^2) \geq (1-\rho_\sigma)\sigma_{\min}^2.
\end{aligned}\]
where $\rho_\sigma$ is specified in Assumption \ref{as:error}. By (\ref{eq: beta Ast bound}),
\begin{equation}\label{eq:bound_sigma_iA}
    (1-\rho_\sigma)\sigma_{\min}^2 \leq \sigma_{iA^*}^2 \lesssim \sigma_{\max}^2, 
\end{equation}
and hence by the bound of the second term on the LHS of (\ref{eq:bound_hat_Wsigma_star}),  $\sigma_{iA}^2\asymp_p \sigma_{iA^*}^2 \asymp 1$ uniformly for all $i\in[n]$. 
In addition, 
\[ \left|\dfrac{u_\gamma^\top \hSigma u_\gamma}{u_\gamma^\top \Sigma u_\gamma} - 1 \right| \leq \dfrac{\|u_\gamma\|_1^2\cdot\|\hSigma - \Sigma\|_\infty}{u_\gamma^\top \Sigma u_\gamma} \lep \dfrac{\|\gamma\|_2^2\sqrt{\dfrac{m_\omega^2 s^2 \log p}{n}}}{u_\gamma^\top \Sigma u_\gamma} \lesssim  m_\omega\sqrt{\dfrac{ s^2 \log p}{n}} = o(1), \] 
under Assumption \ref{as:asym}, implying that 
\begin{equation}\label{eq:rate_ugamma}
   u_\gamma^\top \hSigma u_\gamma \asymp_p u_\gamma^\top \Sigma u_\gamma \asymp_p \|\gamma\|_2^2. 
\end{equation}
Consequently, by the definition of ${\rm U}_\beta$ in (\ref{eq: def U beta}), 
\begin{equation}\label{eq: U beta rate}
    \begin{aligned}
    {\rm U}_\beta = \dfrac{1}{n}\sum_{i=1}^n \left(u_{\gamma}^\top W_{i\cdot}\right)^2\sigma_{iA}^2 \asymp_p \|\gamma\|_2^2.
\end{aligned}
\end{equation}
\par \underline{\textbf{Step 2}.} Define $\chi_i = \dfrac{ u_\gamma^\top W_{i\cdot} e_{iA}}{\sqrt{n{\rm U}_\beta}}$ where $e_{iA}$ is the $i$-th element in the $n$-dimensional vector $e_A$. Thus we have $\mathbb{E}(\chi_i|W) = 0$ and $\sum_{i=1}^n \mathbb{E}(\chi_i^2|W) = 1$. By Corollary 3.1 of \citet{hall1980martingale}, it suffices to show the following conditional Lindeberg condition
\begin{equation}\label{eq:lindeberg_beta}
\sum_{i=1}^n\mathbb{E}\left[\chi_i^2\boldsymbol{1}(|\chi_i|>\chi)\Big|W\right]\convp{0}.
\end{equation} 
for any fixed $\chi > 0$. Following the same arguments in the proof of Lemma 24 in \citet{javanmard2014confidence}, each element of the matrix $\Omega W^\top$ is sub-Gaussian. Consequently, $\|\Omega W^\top\|_\infty \lep \sqrt{\log n + \log p} \leq n^{1/4}$   when $\log p= o(n^{1/3})$ as implied by Assumption \ref{as:asym}. Thus by Proposition \ref{prop:A and A*} and (\ref{eq:VbetaRate}), 
\[\begin{aligned}
|\chi_i|&\leq (n{\rm U}_\beta)^{-1/2}\|A\|_1\|\gamma\|_1\cdot \|\Omega W^\top \|_\infty \cdot |e_{iA}| \\
&\leq C_\chi n^{-1/2}\|\gamma\|_2^{-1} \cdot \sqrt{s}\|\gamma\|_2 \cdot n^{1/4} \cdot |e_{iA}|   = C_\chi s^{1/2}n^{-1/4}\cdot|e_{iA}| 
\end{aligned}
\]
w.p.a.1 for some absolute constant $C_\chi>0$. Also, by (\ref{eq: beta A bound}),
\begin{equation}\label{eq: 2+c0 moment}
    \mathbb{E}\left(|e_{iA}|^{2+c_0}\Big|W\right) \lesssim \mathbb{E}\left(|\varepsilon_{i,Y}|^{2+c_0}\Big|W\right) + |\beta_A|^{2+c_0} \mathbb{E}\left(|\varepsilon_{i,D}|^{2+c_0}\Big|W\right) \lep C_0,
\end{equation} 
where the absolute constants $c_0$ and $C_0$ are defined in Assumption \ref{as:error}. In addition, (\ref{eq:bound_sigma_iA}) and the definition of ${\rm U}_\beta$ in (\ref{eq: def U beta}) imply
\begin{equation}\label{eq: nUbeta LB}
    n{\rm U}_\beta \geq (1-\rho_\sigma)\sigma_{\min}^2\sum_{i=1}^n(u_\gamma^\top W_{i\cdot} )^2
\end{equation}
Therefore, for any $\chi>0$, w.p.a.1, 
\begin{equation}
\begin{aligned}
&\ \ \ \sum_{i=1}^n\mathbb{E}\left[\chi_i^2\boldsymbol{1}(|\chi_i|>\chi)\Big|W\right]\\
&\leq \sum_{i=1}^n\mathbb{E}\left[\chi_i^2 \boldsymbol{1}\left(|e_{iA}| \geq \dfrac{\chi}{C_\chi\cdot s^{1/2}\cdot n^{-1/4} }\right) \Big|W\right]         \\
&= \sum_{i=1}^n\dfrac{(u_\gamma^\top W_{i\cdot} )^2}{n{\rm U}_\beta }\mathbb{E}\left[e_{iA}^2 \boldsymbol{1}\left(|e_{iA}| \geq \dfrac{\chi}{C_\chi\cdot s^{1/2}\cdot n^{-1/4} }\right) \Big|W\right]\\
&\leq  \dfrac{1}{(1-\rho_\sigma)\sigma_{\min}^2} \sum_{i=1}^n\dfrac{(u_\gamma^\top W_{i\cdot} )^2}{ \sum_{i=1}^n(u_\gamma^\top W_{i\cdot} )^2 }\mathbb{E}\left[e_{iA}^2 \boldsymbol{1}\left(|e_{iA}| \geq \dfrac{\chi}{C_\chi\cdot s^{1/2}\cdot n^{-1/4} }\right) \Big|W\right]\\ 
&\lesssim  \left(\dfrac{C_\chi\cdot s^{1/2}}{\chi\cdot  n^{-1/4}}\right)^{c_0} \sum_{i=1}^n\dfrac{(u_\gamma^\top W_{i\cdot} )^2}{ \sum_{i=1}^n(u_\gamma^\top W_{i\cdot} )^2 }\mathbb{E}\left[|e_{iA}|^{2+c_0} \Big|W\right] \leq C_0\cdot \left(\dfrac{C_\chi\cdot s^{1/2}}{\chi\cdot  n^{1/4}}\right)^{c_0}, \\ 
\end{aligned}
\end{equation}
where the fourth row applies (\ref{eq: nUbeta LB}) and the fifth row applies (\ref{eq: 2+c0 moment}). 
The upper bound $\left(\dfrac{C_\chi\cdot s^{1/2}}{\chi \cdot  n^{1/4}}\right)^{c_0} \to 0$ as $n\to\infty$, as $s^{1/2} = o(n^{1/4})$ implied by Assumption \ref{as:asym}. Then the Lindeberg condition (\ref{eq:lindeberg_beta}) holds. We have completed Step 2. 
\par \underline{\textbf{Step 3.}} Show that $\hat{\rm U}_\beta / {\rm U}_\beta \convp 1$. We decompose the estimation error of the asymptotic variance $\hat{\rm U}_\beta - {\rm U}_\beta$ as 
\[\begin{aligned}
    \hat{\rm U}_\beta - {\rm U}_\beta &= \Delta_{1\beta} + \Delta_{2\beta},
\end{aligned}\]
where \[\Delta_{1\beta} =  \dfrac{1}{n}\sum_{i=1}^n\left(\hugamma^\top W_{i\cdot}\right)^2\left( (\hat\varepsilon_{i,Y} - \hbeta_A\hat\varepsilon_{i,D})^2 - \sigma_{iA}^2\right),\]
and 
\[\Delta_{2\beta} =  \dfrac{1}{n}\sum_{i=1}^n\left((\hugamma-\ugamma)^\top W_{i\cdot}\right)^2\sigma_{iA}^2 + (\hugamma-\ugamma)^\top \dfrac{2}{n}\sum_{i=1}^n W_{i\cdot}W_{i\cdot}^\top \sigma_{iA}^2 u_\gamma.  \] 
\par \underline{\textbf{We first bound $\Delta_{1\beta}.$}} Note that by (\ref{eq:VbetaRate}) and (\ref{eq: asymp norm beta A}),  
\begin{equation}\label{eq:betahat_rate}
\begin{aligned}
    \hat\beta_A - \beta_A &= O_p\left(\dfrac{1}{\sqrt{n}\|\gamma\|_2}\right) \lep \dfrac{\sqrt{\log p}}{\sqrt{n}\|\gamma\|_2}.
\end{aligned} 
\end{equation}
Then by Lemma \ref{lem:asymp} $\hat\beta_A - \beta_A = o_p(1)$ and hence by (\ref{eq: beta A bound}) $\hat\beta_A\lep 1$ and 
\begin{equation}\label{eq:betahatsq_rate}
    |\hat\beta_A^2-\beta_A^2| = |(\hat\beta_A-\beta_A)(\hat\beta_A+\beta_A)| \lep  \dfrac{\sqrt{\log p}}{\sqrt{n}\|\gamma\|_2}.
\end{equation}
Then by Proposition \ref{prop:bound_sd_hat}, (\ref{eq:WWsigma2Bound1}), (\ref{eq:WWsigma12Bound1}), (\ref{eq:betahat_rate}), (\ref{eq:betahatsq_rate}) and the fact that $\sigma_{iA}^2 = \sigma_{i,Y}^2 + \beta_A^2\sigma_{i,D}^2 - 2\beta_A\sigma_{i,YD}$, we deduce that 
\[\begin{aligned}
&\ \ \ \ \left\|\dfrac{1}{n}\sumn \Wi\Wi^\top [(\hat\varepsilon_{i,Y}-\hbeta_A\hat\varepsilon_{i,D})^2 - \sigma_{iA}^2]\right\|_\infty \\
&\leq \left\|\dfrac{1}{n}\sumn \Wi\Wi^\top (\hat\varepsilon_{i,Y}^2 - \sigma_{i,Y}^2)\right\|_\infty + \left\|\dfrac{1}{n}\sumn \Wi\Wi^\top (\hbeta_A^2\hat\varepsilon_{i,D}^2 - \beta_A^2\sigma_{i,D}^2)\right\|_\infty \\
&\ \ \ \ + 2\left\|\dfrac{1}{n}\sumn \Wi\Wi^\top (\hbeta_A\hat\varepsilon_{i,Y}\hat\varepsilon_{i,D} - \beta_A\sigma_{i,YD})\right\|_\infty  \\
&\leq \left\|\dfrac{1}{n}\sumn \Wi\Wi^\top (\hat\varepsilon_{i,Y}^2 - \sigma_{i,Y}^2)\right\|_\infty + \hbeta_A^2\left\|\dfrac{1}{n}\sumn \Wi\Wi^\top (\hat\varepsilon_{i,D}^2 - \sigma_{i,D}^2)\right\|_\infty + |\hbeta_A^2-\beta_A^2|\cdot \left\|\dfrac{1}{n}\sumn \Wi\Wi^\top\sigma_{i,D}^2\right\|_\infty \\
&\ \ \ \ \ +\hbeta_A\left\|\dfrac{1}{n}\sumn \Wi\Wi^\top (\hat\varepsilon_{i,Y}\hat\varepsilon_{i,D} - \sigma_{i,YD})\right\|_\infty + |\hbeta_A-\beta_A|\cdot \left\|\dfrac{1}{n}\sumn \Wi\Wi^\top\sigma_{i,YD}\right\|_\infty \\
&\lep \dfrac{s^2\log p}{n} + \left(1+\dfrac{1}{\|\gamma\|_2}\right)\sqrt{\dfrac{\log p}{n}}, 
\end{aligned}\]
where the last inequality applies (\ref{eq:betahat_rate}) and (\ref{eq:betahatsq_rate}). 
In addition, by Lemma \ref{lem:sub-Gaussian trans} we know the entries of \[\tilde{W}_{i\cdot} := \Omega W_{i\cdot}\] are sub-Gaussian with uniformly bounded sub-Gaussian norms. Then similar upper bounds as Proposition \ref{prop:bound_sd_hat} still hold with $W_{i\cdot}$ replaced by $\tilde{W}_{i\cdot}$, which implies 
\[\left\|\dfrac{1}{n}\sum_{i=1}^n \tilde{W}_{i\cdot}\tilde{W}_{i\cdot}^\top\left[(\hat\varepsilon_{i,Y} - \hbeta_A\hat\varepsilon_{i,D})^2-\sigma_{i
    A}^2\right]\right\|_\infty  \lep \dfrac{s^2\log p}{n} + \left(1+\dfrac{1}{\|\gamma\|_2}\right)\sqrt{\dfrac{\log p}{n}}, \]
and 
\[\left\|\dfrac{1}{n}\sum_{i=1}^n  {W}_{i\cdot}\tilde{W}_{i\cdot}^\top\left[(\hat\varepsilon_{i,Y} - \hbeta_A\hat\varepsilon_{i,D})^2-\sigma_{i
    A}^2\right]\right\|_\infty  \lep \dfrac{s^2\log p}{n} + \left(1+\dfrac{1}{\|\gamma\|_2}\right)\sqrt{\dfrac{\log p}{n}}.\]
Then, by Assumption \ref{as:asym}, Lemma \ref{lem:asymp} and (\ref{eq:hu-ugamma}), 
\[\begin{aligned}
     &\ \ \ |\Delta_{1\beta}|\\
     &\leq \left| \dfrac{1}{n}\sum_{i=1}^n \left(W_{i\cdot}^\top(\hugamma-\ugamma)\right)^2\left[(\hat\varepsilon_{i,Y} - \hbeta_A\hat\varepsilon_{i,D})^2-\sigma_{i
    A}^2\right]\right| +  \left| \dfrac{1}{n}\sum_{i=1}^n \left(W_{i\cdot}^\top\ugamma\right)^2\left[(\hat\varepsilon_{i,Y} - \hbeta_A\hat\varepsilon_{i,D})^2-\sigma_{i
    A}^2\right]\right|\\
    &\ \ \ \ \ + \left| \dfrac{2}{n}\sum_{i=1}^n \hugamma^\top W_{i\cdot}W_{i\cdot}^\top\ugamma\left[(\hat\varepsilon_{i,Y} - \hbeta_A\hat\varepsilon_{i,D})^2-\sigma_{i
    A}^2\right]\right| \\
    &\leq \|\hugamma-\ugamma\|_1^2\cdot\left\|\dfrac{1}{n}\sum_{i=1}^n W_{i\cdot}W_{i\cdot}^\top\left[(\hat\varepsilon_{i,Y} - \hbeta_A\hat\varepsilon_{i,D})^2-\sigma_{i
    A}^2\right]\right\|_\infty\\
    &\ \ \ \  + \|\gamma\|_1^2 \cdot\left\|\dfrac{1}{n}\sum_{i=1}^n \tilde{W}_{i\cdot}\tilde{W}_{i\cdot}^\top\left[(\hat\varepsilon_{i,Y} - \hbeta_A\hat\varepsilon_{i,D})^2-\sigma_{i
    A}^2\right]\right\|_\infty \\
     &\ \ \ \ \ + \|\hugamma-\ugamma\|_1 \cdot \|\gamma\|_1 \cdot\left\|\dfrac{1}{n}\sum_{i=1}^n  {W}_{i\cdot}\tilde{W}_{i\cdot}^\top\left[(\hat\varepsilon_{i,Y} - \hbeta_A\hat\varepsilon_{i,D})^2-\sigma_{i
    A}^2\right]\right\|_\infty \\
    &\lep \left( m_\omega\sqrt{\dfrac{ s^2 \log p}{n}} + \dfrac{s_\omega m_{\omega}^{2-2q}\cdot s^{1/2}(\log p)^{(1-q)/2}}{n^{(1-q)/2}}\|\gamma\|_2  \right)^2 \cdot \left(\dfrac{s^2\log p}{n} + \left(1+\dfrac{1}{\|\gamma\|_2}\right)\sqrt{\dfrac{\log p}{n}} \right) \\
    &\ \ \ \ \ + \|\gamma\|_2^2 \cdot \left(\dfrac{s^3\log p}{n}+\sqrt{\dfrac{s^2\log p}{n}}+ \dfrac{1}{\|\gamma\|_2}\sqrt{\dfrac{m_\omega^2 s\log p}{n}}\right)  \\
    &\ \ \ \ \ + \|\gamma\|_2\cdot \left( m_\omega\sqrt{\dfrac{ s^2 \log p}{n}} + \dfrac{s_\omega m_{\omega}^{2-2q}\cdot s^{1/2}(\log p)^{(1-q)/2}}{n^{(1-q)/2}}\|\gamma\|_2  \right)  \left(\dfrac{s^{5/2}\log p}{n}+\left(1+\dfrac{1}{\|\gamma\|_2}\right)\sqrt{\dfrac{s^{3/2}\log p}{n}} \right)  \\ 
    &=o_p(\|\gamma\|_2^2)\cdot o_p(1) + \|\gamma\|_2^2\cdot o_p(1) + \|\gamma\|_2\cdot o_p(\|\gamma\|_2)\cdot o_p(1)\\
    &= o_p(\|\gamma\|_2^2) = o_p({\rm U}_\beta),
\end{aligned}\]
where the last equality applies (\ref{eq:VbetaRate}). 
\par \underline{\textbf{We next bound $\Delta_{2\beta}$}}. By Assumption \ref{as:asym}, Lemma \ref{lem:asymp} and Proposition \ref{prop:bound_sd_hat}, 
\[\begin{aligned}
   |\Delta_{2\beta}| 
    &\lep \|\hugamma-\ugamma\|_1^2\cdot\left\|\dfrac{1}{n}\sum_{i=1}^n W_{i\cdot}W_{i\cdot}^\top\sigma_{iA}^2\right\|_\infty  + \|\hugamma-\ugamma\|_1^2\cdot\|\gamma\|_1\cdot\left\|\dfrac{1}{n}\sum_{i=1}^n \tilde{W}_{i\cdot}W_{i\cdot}^\top\sigma_{iA}^2\right\|_\infty \\
    &\lep \left( m_\omega\sqrt{\dfrac{ s^2 \log p}{n}} + \dfrac{s_\omega m_{\omega}^{2-2q}\cdot s^{1/2}(\log p)^{(1-q)/2}}{n^{(1-q)/2}}\|\gamma\|_2  \right)^2 + \\
    &\ \ \ \ \ \|\gamma\|_2 \cdot \left(m_\omega \sqrt{\dfrac{s^3 \log p}{n}} + \dfrac{s_\omega m_{\omega}^{2-2q}\cdot s(\log p)^{(1-q)/2}}{n^{(1-q)/2}}\|\gamma\|_2  \right) \\
    &= o_p(\|\gamma\|_2^2) = o_p({\rm U}_\beta).
\end{aligned}\]
The probability upper bounds of $\Delta_{1\beta}$ and $\Delta_{2\beta}$ imply that 
\[\dfrac{\hat{\rm U}_\beta}{{\rm U}_\beta} - 1 = \dfrac{\Delta_{1\beta} + \Delta_{2\beta}}{{\rm U}_\beta} \convp 0,\]
or equivalently, $\dfrac{\hat{\rm U}_\beta}{{\rm U}_\beta} \convp 1.$ This completes the proof of Theorem \ref{thm:betaAsympNorm}.

\subsection{Proofs of Theorems in Section \ref{sec:overid}}\label{app:b3}

\input{propositions2}

\subsubsection{Proof of Theorem \ref{thm:M test}}\label{app:b32}
This proof follows the procedure in the proof of Theorem 2.2 in \citet{zhang2017simultaneous}. Conditional on the observed data, the normal vector $\eta \in [p_z]$ is equal in distribution to 
\[\eta \eqd \dfrac{1}{\sqrt{n}}\sum_{i=1}^n \hat{A}_0 \hat{\Omega}_z W_{i\cdot}\hat{e}_{iA}\cdot w_{i}, \]
where $\{w_i\}_{i=1}^n$ are i.i.d. standard normal variables. 
 Define \[\mathcal{T} = \max_{j\in[p_z]}\sqrt{n}A_j^{1/2\top}(\tilde\pi_A-\pi_A),\ \mathcal{T}_0 = \max_{j\in[p_z]}\dfrac{1}{\sqrt{n}}\sumn \xi_{ij},\] 
where $A_j^{1/2\top}$ denotes the $j$-th row of the matrix $A^{1/2}$, and 
\[\mathcal{W} = \max_{j\in[p_z]} \eta_j. \]
Here $\mathcal{T}$ and $\mathcal{T}_0$ are analogs of ``$T_0$'' and ``$T$'' in (14) of \citet{chernozhukov2013gaussian}, and $\mathcal{W}$ is ``$W$'' and ``$W_0$'' in (15) of the same paper.  Proposition \ref{prop:debias} shows that $|\mathcal{T}-\mathcal{T}_0| = o_p\left(\dfrac{1}{\log p}\right)$ and hence, 
\begin{equation}\label{eq: T minus T0}
    \Pr\left(|\mathcal{T}-\mathcal{T}_0| > \zeta_1\right) <\zeta_2,
\end{equation} 
where $\zeta_1 = \dfrac{1}{\log p}$ and $\zeta_2 = o(1)$. Furthermore, define $\varpi(\vartheta):=C_\varpi \vartheta^{1/3}(1\vee \log (p_z/\vartheta))^{2/3}$ with $C_\varphi>0$ large enough and
\[\Delta_{\rm V} := \|\hat{\rm V}_A - {\rm V}_{A^*}\|_\infty. \]
Finally, define the critical value of $\mathcal{W}$
\[{\rm cv}_{\mathcal{W}}(\alpha) := \inf\{x\in\mathbb{R}:{\rm Pr}_\eta(\mathcal{W}\leq x) \geq 1-\alpha\}.\]
Following the same path to verify Theorem 3.2 of \citet{chernozhukov2013gaussian}, we can deduce that 
\[\sup_{\alpha^*\in(0,1)}\left|\Pr\left(\mathcal{T}_0 > {\rm cv}_{\mathcal{W}}(\alpha^*)\right)-\alpha^*\right| \lesssim \varpi(\vartheta) + \Pr(\Delta_{\rm V} > \vartheta) + \zeta_1\sqrt{1\vee \log (p/\zeta_1)} + \zeta_2. \]
where the $n^{-c}$ comes from Proposition \ref{prop:GaussianApprox}. By (\ref{eq: T minus T0}) 
\[\zeta_1\sqrt{1\vee \log (p/\zeta_1)} + \zeta_2 = o(1).\] Take $\vartheta = 1/(\log p)^3$. By (\ref{prop:VA}) and the definition of $\varpi(\vartheta)$ below (\ref{eq: T minus T0}), 
\[\varpi(\vartheta) + \Pr(\Delta_{\rm V} > \vartheta) = o(1).\]
Thus
\begin{equation}\label{eq: Prob Difference alpha}
    \sup_{\alpha^*\in(0,1)}\left|\Pr\left(\mathcal{T}_0 > {\rm cv}_{\mathcal{W}}(\alpha^*)\right)-\alpha^*\right| \to 0,
\end{equation}
as $n \to \infty$. 
\par \underline{\textbf{Prove (\ref{eq: M test size})}}. Recall that $\pi_A = 0$ when $\pi=0$. Then  (\ref{eq: M test size}) is a direct corollary of (\ref{eq: Prob Difference alpha}).  
\par  \underline{\textbf{Prove (\ref{eq: M test power})}}. Let $a_{ij}$ be the normal variable with covariance matrix ${\rm V}_{A^*}$ as defined in Proposition \ref{prop:GaussianApprox}. By Step 1 in the proof of the same lemma, we have $\min_{j\in[p_z]}({\rm V}_{A^*})_{jj})\leq C$ for some absolute constant $C$. By Lemma 6 of \citet{tony2014two}, for any $x\in\mathbb{R}$, 
\[\Pr\left(\max_{j\in[p_z]}\dfrac{(\sum_{i=1}^n a_{ij})^2}{n({\rm V}_{A^*})_{jj}} - 2\log p_z + \log \log p_z \leq x\right)\to F(x) := \exp\left[\dfrac{1}{\sqrt{\pi}}\exp\left(-\dfrac{x}{2}\right)\right],\]
as $p_z \to\infty$, which implies 
\[\Pr\left(\max_{j\in[p_z]}\dfrac{(\sum_{i=1}^n a_{ij})^2}{n({\rm V}_{A^*})_{jj}} 
< 2\log p_z - 0.5\log \log p_z \right)\to1.\]
By the bounds of $({\rm V}_{A^*})_{jj}$, we deduce for some absolute constant $C$,

\begin{equation*} 
    \Pr\left(\max_{j\in[p_z]}\dfrac{(\sum_{i=1}^n a_{ij})^2}{n} < 2C\log p_z  - 0.5C\log \log p_z \right)\to1.
\end{equation*}
 
The Gaussian approximation result from Proposition \ref{prop:GaussianApprox} implies that 
\[\begin{aligned}
&\ \ \ \ \Pr\left(\mathcal{T}_0^2 < 2C\log p_z - 0.5\log\log p_z \right) \\
&=  \Pr\left(\max_{j\in[p_z]}\dfrac{(\sum_{i=1}^n \xi_{ij})^2}{n} < 2C\log p_z - 0.5\log\log p_z \right) \\
&\geq \Pr\left(\max_{j\in[p_z]}\dfrac{(\sum_{i=1}^n a_{ij})^2}{n} < 2C\log p_z - 0.5\log\log p_z \right) - Cn^{-c} \to 1. 
\end{aligned}\]

Then (\ref{eq: T minus T0}) implies
\begin{equation} \label{eq:TstatBound}
    \begin{aligned}
&\ \ \ \ \Pr\left( \mathcal{T}^2 < 3C\log p_z - 0.5 C\log \log p_z \right) \\
&\geq  \Pr\left( |\mathcal{T} - \mathcal{T}_0| + |\mathcal{T}_0| < \sqrt{3C\log p_z - 0.5\log\log p_z} \right) \\
&\geq \Pr\left(|\mathcal{T}_0| < \sqrt{3C\log p_z - 0.5\log\log p_z} - \zeta_1 \right)  - \Pr\left( |\mathcal{T} - \mathcal{T}_0| > \zeta_1 \right) \\
&\geq \Pr\left(|\mathcal{T}_0| < \sqrt{2C\log p_z - 0.5\log\log p_z} \right) - \zeta_2 \to 1.
\end{aligned}
\end{equation}

Recall that conditional on the observed data, $\eta\sim N(0,{\rm V}_A)$. By Proposition \ref{prop:VA} and Lemma 3.1 of \citet{chernozhukov2013gaussian}, taking $t=2C\log p_z - 0.5C\log\log p_z$ for the same Lemma, we deduce that the distribution of $\sqrt{n}\|\eta\|_\infty$ can be well approximated by $\max_{j\in[p_z]}n^{-1/2}|\sum_{i=1}^n a_{ij}|$ so that 
\[\begin{aligned}
    &\ \ {\rm\Pr}_\eta \left( n\|\eta\|_\infty^2 < 2C\log p_z  - 0.5C\log \log p_z \right) \\
    &=\Pr\left(\max_{j\in[p_z]}\dfrac{(\sum_{i=1}^n a_{ij})^2}{n} < 2C\log p_z  - 0.5C\log \log p_z \right) + o_p(1) \to 1.
\end{aligned}\]
Consequently, w.p.a.1, 
\begin{equation}\label{eq: cv bound}
    [{\rm cv}_A(\alpha)]^2 \leq 2C\log p_z - 0.5C\log \log p_z.
\end{equation}
 
Furthermore, since $\|A^{1/2}-A^{*1/2}\|_1 \lep \sqrt{\log p/n}$ by Proposition \ref{prop:A and A*} and 
\[\begin{aligned}
    \|\pi_A-\pi_{A^*}\|_\infty &= \|\gamma\|_\infty \cdot |\beta_A - \beta_{A^*}| \\
    &\leq \|\gamma\|_2\left[\dfrac{| \pi_{A^*}^\top (A-A^*)\gamma|}{\QA(\gamma)} + |\pi_{A^*}^\top A^*\gamma|\cdot\left|\dfrac{1}{\QA(\gamma)} - \dfrac{1}{\QAst(\gamma)}\right| \right] \\
    &\lep \dfrac{\|\pi_{A^*}\|_2\|A-A^*\|_2\|\gamma\|_2 }{\|\gamma\|_2} + \dfrac{\|\gamma\|_2\|\pi_{A^*}\|_2\|A^*\|_2\|\gamma\|_2|\QAst(\gamma)-\QA(\gamma)|}{\QAst^2(\gamma)} \\
    &\lep  \|\pi_{A^*}\|_2\|A-A^*\|_2 +  \dfrac{\|\pi_{A^*}\|_2\QAst(\gamma)\|\gamma\|_2\|A-A^*\|_2\|\gamma\|_2}{\QA(\gamma)} \\
    &\leq  \|\pi_{A^*}\|_2\sqrt{\dfrac{s\log p}{n}} \\
    &\lep \sqrt{\dfrac{s^2\log p}{n}}\|\pi_{A^*}\|_\infty  = o_p(\|(A^*)^{-1/2}\|_1\cdot\|A^{*1/2}\pi_{A^*}\|_\infty) = o_p(\|A^{*1/2}\pi_{A^*}\|_\infty). \\
\end{aligned}\]
We then deduce that 
\begin{equation}\label{eq:ApiSupErr}
    \begin{aligned}
    &\ \ \ \ \|A^{1/2}\pi_A - A^{*1/2}\pi_{A^*}\|_\infty\\
    &\leq \|A^{1/2}-A^{*1/2}\|_1\cdot\|\pi_A-\pi_{A^*}\|_\infty + \|A^{*1/2}\|_1 \|\pi_A-\pi_{A^*}\|_\infty + \|A^{1/2}-A^{*1/2}\|_1\cdot \|\pi_{A^*}\|_\infty \\
    & = o_p(\|A^{*1/2}\pi_{A^*}\|_\infty). \\
\end{aligned} 
\end{equation}
 
Take $C_\pi=\sqrt{5C}$. Combining (\ref{eq:TstatBound}) and (\ref{eq:ApiSupErr}), whenever $\pi\in \mathcal{H}_{A^*}(C_\pi + \epsilon)$ for any $\epsilon>0$,
\[\begin{aligned}
   &\ \ \ \ \Pr\left(\| \sqrt{n}\|A^{1/2}\tilde\pi_A\|_\infty > {\rm cv}_A(\alpha)\right) \\ &\geq \Pr\left( \|\sqrt{n}A^{1/2}\tilde\pi_A\|_\infty^2 >  2C\log p_z - 0.5C\log \log p_z \right) + o(1) \\
    &\geq \Pr\left(\|\sqrt{n}A^{1/2}\pi_A\|_\infty^2 > \|\sqrt{n}A^{1/2}(\tilde\pi_A-\pi_A)\|_\infty^2 + 2C\log p_z - 0.5C\log \log p_z \right) + o(1) \\
       &\geq \Pr\left(\|\sqrt{n}A^{1/2}\pi_A\|_\infty > \sqrt{ \mathcal{T}^2 + 2C\log p_z - 0.5C\log \log p_z} \right) + o(1) \\
    &\geq \Pr\left(\|\sqrt{n}A^{1/2}\pi_A\|_\infty >  \sqrt{ 5C\log p_z - C\log \log p_z } \right) + o(1) \\
     &\geq \Pr\left(\|\sqrt{n}A^{*1/2}\pi_{A^*}\|_\infty > \sqrt{n}\|A^{1/2}\pi_A - A^{*1/2}\pi_{A^*}\|_\infty +  \sqrt{ 5C\log p_z - C\log \log p_z } \right) + o(1) \\
      &\geq \Pr\left(\|\sqrt{n}A^{*1/2}\pi_{A^*}\|_\infty > o_p(\|\sqrt{n}A^{*1/2}\pi_{A^*}\|_\infty) +  \sqrt{ 5C\log p_z - C\log \log p_z } \right) + o(1) \\
       &\geq \Pr\left(\dfrac{\sqrt{5C}}{\sqrt{5C}+\epsilon}\|\sqrt{n}A^{*1/2}\pi_{A^*}\|_\infty >  \sqrt{ 5C\log p_z - C\log \log p_z } \right) + o(1) \\
       &\geq \Pr\left(  \sqrt{ 5C\log p_z} >  \sqrt{ 5C\log p_z - C\log \log p_z } \right) + o(1) \to 1. \\
\end{aligned}\]

\subsubsection{Proof of Theorem \ref{thm:Q}}\label{app:b33} 
We have the following decomposition of $\hat Q_A$
\begin{equation}
\begin{aligned}
\hat Q_A-Q_A &=  \widehat{Q}_A - \widecheck{Q}_A +  \widecheck{Q}_A - Q_A \\
&= \dfrac{2}{n} \hu_{\pi_A}^\top W^\top\widecheck{e}_A + 2(\widehat{\Sigma}\hu_{\pi_A}-(0_{p_x}^\top , A\hat\pi_A^\top)^\top)^\top\left(\begin{array}{c}
    \widehat{\varphi}_A - \widecheck{\varphi}_A  \\
    \widehat{\pi}_A - \widecheck{\pi}_A  \\
\end{array}\right) - \QA(\widehat{\pi}_A-\widecheck{\pi}_A) + \QA(\gamma)(\hbeta_A-\beta_A)^2 \\
&=\dfrac{2}{n}u_{\pi_A}^\top W^\top e_A + \Delta_{1Q} + \Delta_{2Q} ,\\
\end{aligned}
\end{equation}
where 
\begin{equation} \label{eq:B1}
\Delta_{1Q}=2(\widehat{\Sigma}\hu_{\pi_A}-(0_{p_x}^\top , \hat\pi_A^\top A)^\top)^\top\left(\begin{array}{c}
    \widehat{\varphi}_A - \widecheck{\varphi}_A  \\
    \widehat{\pi}_A - \widecheck{\pi}_A  \\
\end{array}\right) -  \QA(\widehat{\pi}_A-\widecheck{\pi}_A) + \QA(\gamma)(\hbeta_A-\beta_A)^2.
\end{equation}
and 
\begin{equation} \label{eq:B2}
\Delta_{2Q} = \dfrac{2}{n}(W\hat{u}_{\pi_A})^\top(\widecheck{e}_A - e_A) + \dfrac{2}{n}(\hu_{\pi_A} - u_{\pi_A})^\top\Omega W^\top e_A. 
\end{equation}
Recall that ${u}_{\pi_A} = (0_{p_x}^\top, \pi_A^\top A)^\top$ and ${\hat u}_{\pi_A} = (0_{p_x}^\top, \hat\pi_A^\top \hat A)^\top$. 
Define $\epsilon_n:=n^{1/4}\|\pi_{A^*}\|_2$. Note that by Proposition \ref{prop:A and A*} we can deduce $Q_A \asymp_p n^{-1/2}\epsilon_n^2$. Suppose that $|\Delta_{1Q} + \Delta_{2Q}| = o_p\left(\dfrac{1+\epsilon_n^2}{\sqrt{n}\log p}\right)$. \begin{itemize}
    \item[(a)] When $\epsilon_n=0$, we have $\pi=0$, and thus $Q_A=0$, $u_{\pi_A}=0$. Then  
\[\sqrt{n}\log p \hat Q_A = \sqrt{n}\log p(\Delta_{1Q} + \Delta_{2Q}) \convp 0.\]
\item[(b)] When $\epsilon_n\gtrsim 1$, we have $\epsilon_n\lesssim \epsilon_n^2$. Besides,
\begin{equation}\label{eq: betaA bound for Q}
    |\beta_A|\lesssim |\beta| + \dfrac{|\IA(\pi,\gamma)|}{\QA(\gamma)} \lep 1 + \pigamma \lep 1+\dfrac{\epsilon_n}{n^{1/4}\|\gamma\|_2}.
\end{equation}
Thus, by Assumption \ref{as:SI}, Proposition \ref{prop:A and A*} and (\ref{eq:WeBound})
\[\begin{aligned}
    \left|\dfrac{u_{\pi_A}^\top W^\top e_A}{n}\right| &\lep \|u_{\pi_A}\|_1 \cdot \left\|n^{-1}W^\top(\varepsilon_Y-\varepsilon_D \beta_A)\right\|_\infty  \\
    &\lep \left(1+|\beta_A|\right)\|A\|_1\|\Omega\|_1\|\pi_A\|_2\sqrt{s\log p/n} \\
    &\lep \left(1+\pigamma\right)\cdot m_\omega \cdot\dfrac{\epsilon_n}{n^{1/4}}\cdot\sqrt{\dfrac{s\log p}{n}}   \\
       &\lep  m_\omega \cdot\dfrac{\epsilon_n}{n^{1/4}}\cdot\sqrt{\dfrac{s\log p}{n}} + m_\omega\cdot \dfrac{\epsilon_n^2}{\sqrt{n}\|\gamma\|_2}\cdot \sqrt{\dfrac{s\log p}{n}} \\
    &= o_p\left(\dfrac{\epsilon_n^2}{\sqrt{n}}\right),
\end{aligned}\]
\end{itemize}
where the last step applies Lemma \ref{lem:asymp}. 
Thus, w.p.a.1, \[\begin{aligned}
     \sqrt{n}\log p\hat Q_A  - c\sqrt{\log p} &=   \sqrt{n}\log p \left(Q_A + \dfrac{2}{n}u_{\pi_A}^\top W^\top e_A + \Delta_{1Q} + \Delta_{2Q} \right)- c\sqrt{\log p}\\
     &= \log p \cdot\epsilon_n^2 - c\sqrt{\log p} + o_p\left(\log p\cdot\epsilon_n^2\right)   \gtrsim  \log p - c\sqrt{\log p} \convp \infty, 
\end{aligned}\]
for any $c>0$. Consequently, it suffices to show that $|\Delta_Q| = |\Delta_{1Q} + \Delta_{2Q}| = o_p(n^{-1/2}(\log p)^{-1}(1+\epsilon_n^2)). $
\par \underline{\textbf{Show $|\Delta_{1Q}| = o_p(n^{-1/2}(\log p)^{-1})(1+\epsilon_n^2))$}}. By (\ref{eq:betahat_rate with pi}), the definition $\epsilon_n = n^{1/4}\|\pi_{A^*}\|_2$ and (\ref{eq: equi L2 norm}), 
\begin{equation}\label{eq:betahat with eps}
    |\hat\beta_A - \beta_A| = O_p\left(\dfrac{\epsilon_n\cdot m_\omega\sqrt{s\log p}}{n^{3/4}\QA(\gamma)} + \dfrac{m_\omega\sqrt{s\log p}}{\sqrt{n}\|\gamma\|_2}\right).
\end{equation} 
In addition, by (\ref{eq:betahat with eps}), Proposition \ref{prop:Xihat}, Lemma \ref{lem:CLIME}, Proposition \ref{prop:Xihat} and Proposition \ref{prop:A and A*},
\[\begin{aligned}
&\ \ \ \ \|\widehat{\Sigma}\hu_{\pi_A}-(0_{p_x}^\top , \hat\pi_A^\top A)^\top\|_\infty  \\
&\leq \|A\|_1 \cdot \|\hat\pi_A\|_1 \cdot \|\hSigma\hOmega - I\|_\infty \\
&\lep \dfrac{s_\omega m_\omega^{2-2q} (\log p)^{(1-q)/2}}{n^{(1-q)/2}}  \cdot \left(\|\hat\pi_A - \widecheck\pi_A\|_1 + \|\widecheck\pi_A - \pi_A\|_1 + \|\pi_A\|_1\right)  \\
&\lep \dfrac{s_\omega m_\omega^{2-2q} (\log p)^{(1-q)/2}}{n^{(1-q)/2}}  \cdot \left( \left(1+\pigamma\right)\sqrt{\dfrac{s^2\log p}{n}} + \sqrt{s}\|\gamma\|_2\cdot |\hbeta_A-\beta_A| + \sqrt{s}\|\pi_A\|_2\right)  \\
&\lep \dfrac{s_\omega m_\omega^{2-2q} (\log p)^{(1-q)/2}}{n^{(1-q)/2}}  \cdot \left( \left(1+\pigamma\right)\sqrt{\dfrac{m_\omega^2 s^2\log p}{n}} + \dfrac{\epsilon_nm_\omega s\sqrt{\log p}}{n^{3/4}\|\gamma\|_2} + \sqrt{s}\|\pi_A\|_2\right), 
\end{aligned}\]
and hence, 
\begin{equation}\label{eq: for Delta1Q temp}
    \begin{aligned}
    &\ \ \ \ \ \left|(\widehat{\Sigma}\hu_{\pi_A}-(0_{p_x}^\top , \hat\pi_A^\top A)^\top)^\top\left(\begin{array}{c}
    \widehat{\varphi}_A - \widecheck{\varphi}_A  \\
    \widehat{\pi}_A - \widecheck{\pi}_A  \\
\end{array}\right)\right|\\
&\leq \left[\|\hat\varphi_A-\widecheck\varphi_A\|_1 + \|\hat\pi_A-\widecheck\pi_A\|_1\right] \cdot \|\widehat{\Sigma}\hu_{\pi_A}-(0_{p_x}^\top , \hat\pi_A^\top A)^\top\|_\infty  \\
&\lep \dfrac{s_\omega m_\omega^{2-2q} (\log p)^{(1-q)/2}}{n^{(1-q)/2}}   \cdot \\
&\ \ \ \ \left( \left(1+\pigamma\right)^2 \dfrac{\epsilon_nm_\omega s^2\log p}{n}  + \dfrac{\left(1+\pigamma\right)m_\omega s^2 {\log p}}{n^{5/4}\|\gamma\|_2} + \left(1+\pigamma\right)\dfrac{s^{3/2}\log p}{\sqrt{n}}\|\pi_A\|_2\right) \\
&= o_p\left(\dfrac{1\wedge\|\gamma\|_2}{s\log p}\right)  \cdot \left( \left(1+\pigamma\right)^2 \dfrac{\epsilon_nm_\omega s^2\log p}{n}  + \dfrac{\left(1+\pigamma\right)m_\omega s^2 {\log p}}{n^{5/4}\|\gamma\|_2} + \left(1+\pigamma\right)\dfrac{s^{\frac{3}{2}}\log p\|\pi_A\|_2}{\sqrt{n}}\right)
\\
&= O_p\left[\left(1+\pigamma\right)^2 \dfrac{m_\omega s\log p}{n}\right]  + \\
&\ \ \ \ o_p\left(\dfrac{1\wedge\|\gamma\|_2}{\log p}\right)  \cdot \left(  \dfrac{\left(1+\pigamma\right)\epsilon_nm_\omega s \cdot \log p}{n^{5/4}\|\gamma\|_2} + \left(1+\pigamma\right)\dfrac{s^{1/2}\log p}{\sqrt{n}}\|\pi_A\|_2\right).
\end{aligned}
\end{equation}
With (\ref{eq:betahat with eps}), (\ref{eq: for Delta1Q temp}) and $\QA(\hat\pi_A-\widecheck\pi_A)\lep (1+\pisgamma)^2 s\log p/n$ from Proposition \ref{prop:Xihat}, we can deduce that 
\[\begin{aligned}
    |\Delta_{1Q}| &=o_p\left(\dfrac{ \left(1+\pigamma\right)(1\wedge \|\gamma\|_2)}{\log p}\right)\cdot\left(\dfrac{\epsilon_nm_\omega s\log p}{n^{5/4}\|\gamma\|_2} + \dfrac{s^{1/2}\log p\|\pi_A\|_2}{\sqrt{n}}\right)+\\
    &\ \ \ \ \   O_p\left[ \left(1+\pigamma\right)^2\dfrac{m_\omega^2 s\log p}{n}  \right]  + O_p\left(\dfrac{\epsilon_n^2m_\omega^2 s\log p}{n^{3/2}\QA(\gamma)} + \dfrac{m_\omega^2 s^2 \log p}{n}\right)\\
    &= o_p\left(\dfrac{1+\epsilon_n^2}{\sqrt{n}\log p}\right).
\end{aligned}\]
Below we show the last step to derive the $o_p\left(\dfrac{1+\epsilon_n^2}{\sqrt{n}\log p}\right)$ term by term. 
By Lemma \ref{lem:asymp} and (\ref{eq: betaA bound for Q}),
\[\begin{aligned}
    \left(1+\pigamma\right)^2\dfrac{m_\omega^2 s\log p}{n} &\lep \dfrac{m_\omega^2 s\log p}{n} + \dfrac{m_\omega^2 s(\log p)^2}{n\QA(\gamma)}\cdot\dfrac{\epsilon_n^2}{\sqrt{n}\log p} = o_p\left(\dfrac{1+\epsilon_n^2}{\sqrt{n}\log p}\right),
\end{aligned}\]
and 
\[\begin{aligned}
     &\ \ \ \ o_p\left(\dfrac{ \left(1+\pigamma\right)(1\wedge \|\gamma\|_2)}{\log p}\right)\cdot\left(\dfrac{\epsilon_nm_\omega s\log p}{n^{5/4}\|\gamma\|_2} + \dfrac{s^{1/2}\log p\|\pi_A\|_2}{\sqrt{n}\log p}\right) \\
     &= o_p\left(\dfrac{(1 + \|\pi\|_2/\|\gamma\|_2)\cdot (1\wedge\|\gamma\|_2)}{\sqrt{n} \log p} \right)\cdot\left(\dfrac{\epsilon_n m_\omega s \log p}{n^{3/4}\|\gamma\|_2 } + \dfrac{s^{1/2}\log p \cdot \epsilon_n}{n^{1/4}  } \right) \\
     &= o_p\left(\dfrac{(1 + \|\pi\|_2)\epsilon_n}{\sqrt{n}\log p }\right) =  o_p\left(\dfrac{\epsilon_n + n^{-1/4}\epsilon_n^2}{\sqrt{n}\log p }\right) = o_p\left(\dfrac{1+\epsilon_n^2}{\sqrt{n}\log p }\right).
\end{aligned}\] 
where the last equality applies $\epsilon_n\leq (1+\epsilon_n^2)/2$, and 
\[ \dfrac{\epsilon_n^2 m_\omega^2 s\log p}{n^{3/2}\QA(\gamma)} = \dfrac{m_\omega^2 s(\log p)^2}{n\QA(\gamma)}\cdot \dfrac{\epsilon_n^2}{\sqrt{n}\log p} = o_p\left(\dfrac{1+\epsilon_n^2}{\sqrt{n} \log p}\right).\]
This completes the proof of $\Delta_{1Q} = o_p(n^{-1/2}(\log p)^{-1}(1+\epsilon_n^2))$.
\par \underline{\textbf{Show $\Delta_{2Q} = o_p(n^{-1/2}(\log p)^{-1}(1+\epsilon_n^2))$.}} Note that by Propositions \ref{prop:Xihat} and \ref{prop:A and A*}, Equation \eqref{eq: betaA bound for Q} and Lemma \ref{lem:CLIME},
\[\begin{aligned}
    \left\|\hu_{\pi_A}-u_{\pi_A}\right\|_1 &\leq \|A\|_1 \|\hOmega-\Omega\|_1\|\pi_A\|_1+\|A\|_1\|\Omega\|_1\|\hat\pi_A - \pi_A\|_1 \\
    &\lep \dfrac{s_{\omega}m_{\omega}^{2-2q}(\log p)^{(1-q)/2}}{n^{(1-q)/2}} \cdot \sqrt{s}\|\pi_A\|_2 + m_\omega\left(\|\hat\pi_A - \widecheck\pi_A\|_1 +  \|\widecheck\pi_A - \pi_A\|_1\right) \\
    &\lep \dfrac{s_{\omega}m_{\omega}^{2-2q}(\log p)^{(1-q)/2}}{n^{(1-q)/2}} \cdot \sqrt{s}\|\pi_A\|_2 + \dfrac{\left(1+\pigamma\right)m_\omega s\sqrt{\log p}}{\sqrt{n}} + m_\omega \|\gamma\|_1\cdot |\hat\beta_A - \beta_A| \\
    &\lep \dfrac{s_{\omega}m_{\omega}^{2-2q}\sqrt{s}(\log p)^{(1-q)/2}}{n^{(1-q)/2}} \cdot \|\pi_A\|_2 + \\
    &\ \ \ \ \ \dfrac{\left(1+\pigamma\right)m_\omega s\sqrt{\log p}}{\sqrt{n}} + m_\omega\sqrt{s} \|\gamma\|_2\cdot \left(\dfrac{\epsilon_nm_\omega \sqrt{s\log p}}{n^{3/4}\QA(\gamma)} + \left(1+\pigamma\right)\dfrac{m_\omega \sqrt{s\log p}}{\sqrt{n}\|\gamma\|_2}\right)\\
    &\lep \dfrac{s_{\omega}m_{\omega}^{2-2q}\sqrt{s}(\log p)^{(1-q)/2}}{n^{(1-q)/2}} \cdot \|\pi_A\|_2 +  \dfrac{\left(1+\pigamma\right)m_\omega^2 s\sqrt{\log p}}{\sqrt{n}} +   \dfrac{\epsilon_nm_\omega^2 s \sqrt{\log p}}{n^{3/4} \|\gamma\|_2} \\
    &\lep o_p\left(\dfrac{\|\pi\|_2}{\sqrt{\log p}}\right) + \left(1+\dfrac{\|\pi\|_2}{\|\gamma\|_2}\right)\dfrac{m_\omega^2 s\sqrt{\log p}}{\sqrt{n}} +  \dfrac{\epsilon_nm_\omega^2 s\sqrt{\log p}}{n^{3/4}\|\gamma\|_2} \\ 
    &= o_p\left(\dfrac{\epsilon_n}{n^{1/4}\sqrt{\log p}}\right) + O_p\left(\dfrac{m_\omega^2 s\sqrt{\log p}}{\sqrt{n}}\right) +  \dfrac{2\epsilon_nm_\omega^2 s\sqrt{\log p}}{n^{3/4}\|\gamma\|_2} \\
    &= o_p\left(\dfrac{\epsilon_n}{n^{1/4}\sqrt{\log p}}\right) + O_p\left(\dfrac{m_\omega^2 s\sqrt{\log p}}{\sqrt{n}}\right),
\end{aligned}\]
where the last equality applies that $\dfrac{m_\omega^2 s\log p}{n^{1/2}\|\gamma\|_2} = o_p(1)$. In addition, 
\[\|u_{\pi_A}\|_1 \leq \|\Omega\|_1\|A\|_1\|\pi_A\|_1\lep m_\omega\sqrt{s}\|\pi\|_2 \lesssim \dfrac{m_\omega \sqrt{s} \epsilon_n}{n^{1/4}}. \]
Then applying the upper bounds of $\left\|\hu_{\pi_A}-u_{\pi_A}\right\|_1 $ and $\|u_{\pi_A}\|_1$ derived above, together with Proposition \ref{prop:A and A*}, (\ref{eq:WeBound}), (\ref{eq:betahat with eps}) and (\ref{eq: betaA bound for Q}), the first term of $\Delta_{2Q}$ is bounded by 
\[\begin{aligned}
    &\ \ \ \ \left|\dfrac{2}{n}\hu_{\pi_A}^\top W^\top (\widecheck e_A - e_A)\right|\\
    &\leq 2\|\hat u_{\pi_A}\|_1\cdot \|n^{-1}W^\top \varepsilon_D\|_\infty\cdot |\hat\beta_A - \beta_A|  \\
    &\lep (\|\hat u_{\pi_A} - u_{\pi_A}\|_1 + \|u_{\pi_A}\|_1) \cdot \sqrt{\dfrac{\log p}{n}} \cdot \left(\dfrac{\epsilon_nm_\omega \sqrt{s\log p}}{n^{3/4}\QA(\gamma)} + \left(1+\pigamma\right)\dfrac{m_\omega \sqrt{s\log p}}{\sqrt{n}\|\gamma\|_2}\right) \\
    &\lep O_p\left(\dfrac{\epsilon_n}{n^{1/4}\sqrt{\log p}} + \dfrac{m_\omega\sqrt{s}\epsilon_n}{n^{1/4}}\right)  \sqrt{\dfrac{\log p}{n}}\cdot  \left(\dfrac{\epsilon_nm_\omega \sqrt{s\log p}}{n^{3/4}\QA(\gamma)} + \left(1+\pigamma\right)\dfrac{m_\omega \sqrt{s\log p}}{\sqrt{n}\|\gamma\|_2}\right) \\
    &\ \ \ \ +  O_p\left(\dfrac{m_\omega^2 s\sqrt{\log p}}{\sqrt{n}} \right) \sqrt{\dfrac{\log p}{n}}\cdot  \left(\dfrac{\epsilon_nm_\omega \sqrt{s\log p}}{n^{3/4}\QA(\gamma)} + \left(1+\pigamma\right)\dfrac{m_\omega \sqrt{s\log p}}{\sqrt{n}\|\gamma\|_2}\right) \\
    &\lep O_p\left(\dfrac{\epsilon_n}{\sqrt{n}\log p}\right) \cdot  \left(\dfrac{\epsilon_nm_\omega^2 s(\log p)^{3/2}}{n\QA(\gamma)} + \dfrac{m_\omega^2 s(\log p)^{3/2} }{n^{3/4}\|\gamma\|_2} + \dfrac{\|\pi\|_2 m_\omega^2 s(\log p)^{3/2} }{n^{3/4}\QA(\gamma)}\right) \\
    &\ \ \ \ +  O_p\left(\dfrac{1}{\sqrt{n}\log p}\right)   \left(\dfrac{\epsilon_nm_\omega^3 s^{3/2} (\log p)^{5/2}}{n^{5/4}\QA(\gamma)} +  \dfrac{m_\omega^3 s^{3/2} (\log p)^{5/2}}{n\|\gamma\|_2} + \dfrac{\|\pi\|_2m_\omega^3 s^{3/2} (\log p)^{5/2}}{n\QA(\gamma)}\right) \\
    &\lep o_p\left(\dfrac{1+\epsilon_n+\epsilon_n^2}{\sqrt{n}\log p}\right) \cdot  \left(o_p(1) + o_p(1) + \dfrac{m_\omega s(\log p)^{3/2}}{n\QA(\gamma)}\right) \\
    &\ \ \ \ +  O_p\left(\dfrac{1+\epsilon_n+\epsilon_n^2}{\sqrt{n}\log p}\right) \cdot  \left(\dfrac{m_\omega^2 s(\log p)^2 }{n\QA(\gamma)} +  \dfrac{m_\omega \sqrt{s}(\log p)^{3/2}}{\sqrt{n}} + \dfrac{m_\omega^2 s (\log p)^2}{n\QA(\gamma)}\right) \\ 
    &= o_p\left(\dfrac{1+\epsilon_n^2}{\sqrt{n}\log p}\right),
\end{aligned}\]
where the last two steps apply Lemma \ref{lem:asymp}. Besides, using the same set of probability upper bounds, the second term of $\Delta_{2Q}$ is bounded by 
\[\begin{aligned}
    \left|\dfrac{2}{n}(\hu_{\pi_A} - u_{\pi_A})^\top\Omega W^\top e_A\right| &\leq \|\hu_{\pi_A} - u_{\pi_A}\|_1 \cdot \left\|\dfrac{2W^\top e_A}{n}\right\|_\infty \\
    &= \left[o_p\left(\dfrac{\epsilon_n}{n^{1/4}\sqrt{\log p}} \right) + O_p\left(\dfrac{m_\omega^2 s\sqrt{\log p}}{\sqrt{n}}\right)\right]\cdot\left\|\dfrac{2W^\top(\varepsilon_Y-\beta_A\varepsilon_D)}{n}\right\|_\infty \\
    &= \left[o_p\left(\dfrac{\epsilon_n}{n^{1/4}\sqrt{\log p}} \right) + O_p\left(\dfrac{m_\omega^2 s\sqrt{\log p}}{\sqrt{n}}\right)\right]\cdot\left(1+|\beta_A|\right)\sqrt{\dfrac{\log p}{n} }\\ 
     &= O_p\left(\dfrac{1+\epsilon_n}{\sqrt{n}\log p}\right)\cdot\left[\dfrac{\log p}{n^{1/4}}  + \dfrac{m_\omega^2 s(\log p)^2}{\sqrt{n}} + \dfrac{\|\pi\|_2}{\|\gamma\|_2}\dfrac{\log p}{n^{1/4}} + \dfrac{\|\pi\|_2}{\|\gamma\|_2}\dfrac{m_\omega^2 s (\log p)^2}{\sqrt{n}}\right]  \\
     &= O_p\left(\dfrac{1+\epsilon_n+\epsilon_n^2}{\sqrt{n}\log p}\right)\cdot\left[o_p(1)  + o_p(1) +  \dfrac{\log p}{n^{1/2}\|\gamma\|_2} +  \dfrac{m_\omega s(\log p)^2}{\sqrt{n}\|\gamma\|_2}\right]\\
     &= o_p\left(\dfrac{1+\epsilon_n^2}{\sqrt{n}\log p}\right). 
\end{aligned}\]
This completes the proof of Theorem \ref{thm:Q}.

%% file: lemmas.tex
\subsection{Preliminary Lemmas}\label{app:lem}
This subsection provides useful lemmas implied by (or directly from) other literature. 
\par Define the \emph{restricted eigenvalue} of the empirical Gram matrix $\hSigma = W^\top W /n $, given as 
\begin{equation}
    \kappa(\hSigma,s) = \inf_{\theta\in\mathcal{R}(s)} \dfrac{\theta^\top\hSigma\theta}{\|\theta\|_2^2},
\end{equation}
where the restricted set $\mathcal{R}(s):=\{\theta\in\mathbb{R}^p:\|\theta_{\mathcal{M}^c}\|_1\leq 3\|\theta_\mathcal{M}\|_1\text{ for all }\mathcal{M}\subset\mathbb{R}^p \text{ and } |\mathcal{M}|\leq s\}$. 
Lemma \ref{lem:Lasso} provides the Lasso convergence rate. This is a direct result of Lemma 1 in \citet{mei2022lasso} and Theorem 6.1 of \citet{buhlmann2011statistics}. 
\begin{lemma} \label{lem:Lasso}
Suppose that $4\|n^{-1}W^\top\varepsilon_j\|_\infty \leq \lambda_{jn}$ for $j=1,2$. Then 
\begin{equation}
    \begin{aligned}
\max\{\|\widehat{\Gamma}-\Gamma\|_2,\|\widehat{\gamma}-\gamma\|_2,\|\widehat{\Psi}-\Psi\|_2,\|\widehat{\psi}-\psi\|_2\} &\lesssim \dfrac{\sqrt{s}\lambda_n}{\kappa(\hSigma,s)}, \\
\max\{\|\widehat{\Gamma}-\Gamma\|_1,\|\widehat{\gamma}-\gamma\|_1,\|\widehat{\Psi}-\Psi\|_1,\|\widehat{\psi}-\psi\|_1\} &\lesssim \dfrac{s\lambda_n}{\kappa(\hSigma,s)}.  
\end{aligned}
\end{equation}
with $\lambda_n=\max(\lambda_{1n},\lambda_{2n})$. In addition, if $4\|n^{-1}W^\top \widecheck{e}_A \|_\infty \leq \lambda_{3n}$, 
\begin{equation}
    \begin{aligned}
\max\{\|\widehat{\pi}_A-\widecheck\pi_A\|_2,\|\widehat{\varphi}_A-\widecheck\varphi_A\|_2\} &\lesssim \dfrac{\sqrt{s}\lambda_{3n}}{\kappa(\hSigma,s)}, \\
\max\{\|\widehat{\pi}_A-\widecheck\pi_A\|_1,\|\widehat{\varphi}_A-\widecheck\varphi_A\|_1\} &\lesssim \dfrac{s\lambda_{3n}}{\kappa(\hSigma,s)}.  
\end{aligned}
\end{equation}
\end{lemma}
Lemma \ref{lem:re_db} shows the probability bounds for the maximum norm of some sub-Gaussian and sub-exponential variables, and a lower bound of the restricted eigenvalue useful in the proofs. 
\begin{lemma} \label{lem:re_db}Under Assumptions \ref{as:DecisionMatrix}-\ref{as:error},
\begin{align}
   \max_{i\in[n],j\in[p]}|W_{ij}| &\lep \sqrt{\log p + \log n}.\label{eq:Wbound}  
\end{align}
When $\log p=o(n)$,
\begin{align}    
    \|\hat\Sigma - \Sigma\|_\infty &\lep \sqrt{\dfrac{\log p}{n}},\label{eq:Sigma_bound} \\
    \|n^{-1}W^\top\varepsilon_j\|_\infty &\lep \sqrt{\dfrac{\log p}{n}}, \text{ for }j=1,2.\label{eq:WeBound}
\end{align}
Besides, when $s = o(\sqrt{n/\log p})$, w.p.a.1 
\begin{equation}\label{eq:RE_lower}
    \kappa(\hSigma,s) \geq 0.5c_\Sigma. 
\end{equation}
\end{lemma}
\begin{proof}[Proof of Lemma \ref{lem:re_db}]
    By Assumption \ref{as:DecisionMatrix}, we can deduce (\ref{eq:Wbound}) by the sub-Gaussianity of $W_{i\cdot}$, which implies 
    \[\Pr\left(\max_{i\in[n],j\in[p]}|W_{ij}| > \sqrt{2c^{-1}\cdot\log (np)} \right) \leq np\cdot C{\rm e}^{-2\log (np) } = C(np)^{-1}\to0.\]
    In terms of (\ref{eq:Sigma_bound}) and (\ref{eq:WeBound}), note that the products of two sub-Gaussian variables are sub-exponential. The LHS of the inequalities is the maximum norm of sub-exponential vectors with mean zero. By Corollary 5.17 in \citet{vershynin2010introduction}, when $\log p = o(n)$ there exists some $c>0$ such that  
    \[\Pr\left(\|\hSigma-\Sigma\|_\infty > \sqrt{2\log p/(cn)} \right) \leq 2p\cdot \exp(-2\log p)\to 0,\]
    and similar probability bound holds for $n^{-1}W^\top \varepsilon_j$. 
    As for (\ref{eq:RE_lower}), for any $\theta\in\mathcal{R}$
    \[\begin{aligned}
        \theta^\top\hSigma\theta &\geq \theta^\top\Sigma\theta - \left|\theta^\top(\hSigma-\Sigma)\theta\right| \\
        &\geq c_\Sigma \theta^\top\theta - \|\theta\|_1^2\|\hSigma-\Sigma\|_\infty \\
        &\geq c_\Sigma \|\theta\|_2^2 - (\|\theta_{\mathcal{M}}\|_1+\|\theta_{\mathcal{M}^c}\|_1)^2\cdot c\sqrt{\dfrac{\log p}{n}} \\
        &\geq c_\Sigma \|\theta\|_2^2 - (4\|\theta_{\mathcal{M}}\|_`)^2 c\sqrt{\dfrac{\log p}{n}}\\
        &\geq c_\Sigma \|\theta\|_2^2 - 16c\cdot s\sqrt{\dfrac{\log p}{n}}\cdot\|\theta\|_2^2 
        \geq 0.5c_\Sigma \|\theta\|_2^2,
    \end{aligned}\]
    for some absolute constant $c>0$, where the last inequality applies $s=o(\sqrt{n/\log p})$.
\end{proof}

Lemma \ref{lem:sub-Gaussian trans} shows that under certain conditions, linear transformations of sub-Gaussian vectors are still sub-Gaussian. 
\begin{lemma} \label{lem:sub-Gaussian trans} Suppose that all entries in the vector $x=(x_1,x_2,\cdots,x_p)^\top\in\mathbb{R}^p$ is a centered sub-Gaussian vector such that $\mathbb{E}(x)=0$ and $\|x\|_{\psi_2}\leq C_x$ for some absolute constant $C_x$. Then for any matrix $B\in\mathbb{R}^{p\times p}$ such that $\|B\|_2 \leq C_B$, then the $p\times 1$ vector $Bx$ is also sub-Gaussian such that $\|Bx\|_{\psi_2} \leq C_B\cdot C_x$. 
\end{lemma} 
\begin{proof}[Proof of Lemma \ref{lem:sub-Gaussian trans}] The result follows by 
\[\begin{aligned}
    \|Bx\|_{\psi_2} &= \sup_{\|b\|_2=1}\sup_{q\geq 1} \dfrac{1}{\sqrt{q}}\left(\mathbb{E}|b^\top Bx|^q\right)^{1/q} \\
    &= \sup_{\|b\|_2=1}\sup_{q\geq 1} \dfrac{\|B^\top b\|_2}{\sqrt{q}}\left(\mathbb{E}\left|\dfrac{b^\top B}{\|B^\top b\|_2}x\right|^q\right)^{1/q} \\
    &\leq \sup_{\|b\|_2=1}\sup_{q\geq 1} \dfrac{\|B\|_2}{\sqrt{q}}\left(\mathbb{E}\left|\dfrac{b^\top B}{\|B^\top b\|_2}x\right|^q\right)^{1/q} \\
    &\leq \|B\|_2 \cdot \sup_{\|\delta\|_2=1}\sup_{q\geq 1} \dfrac{1}{\sqrt{q}}\left(\mathbb{E}|\delta^\top x|^q\right)^{1/q} \leq C_B\cdot C_x\\
\end{aligned}\]
where the first and the last step applies the definition of sub-Gaussian norm in Definition \ref{def: subG}. 
\end{proof}

Lemma \ref{lem:CLIME} shows the asymptotic properties of the inverse covariance estimator CLIME \eqref{eq:clime}.
\begin{lemma} \label{lem:CLIME}Under Assumptions \ref{as:DecisionMatrix}-\ref{as:SI} and \ref{as:tuning}, 
\begin{equation}\label{eq:CLIMEestL1Bound}
    \|\hOmega\|_1 \leq m_\omega,
\end{equation}
w.p.a.1. Besides, 
\begin{equation}\label{eq:CLIMEL1}
    \|\hOmega-\Omega\|_1 \lep s_\omega\cdot m_\omega^{2-2q}\left(\dfrac{\log p}{n}\right)^{(1-q)/2},
\end{equation}
\begin{equation}\label{eq:CLIMEsup}
    \|\hSigma\hOmega-I\|_\infty \lep s_\omega\cdot m_\omega^{2-2q}\left(\dfrac{\log p}{n}\right)^{(1-q)/2}.
\end{equation}
\end{lemma}
\begin{proof}[Proof of Lemma \ref{lem:CLIME}]By Lemma \ref{lem:sub-Gaussian trans}, each element of $X\Omega^{1/2}$ is sub-Gaussian with uniformly bounded sub-Gaussian norm. By Lemma 23 in \citet{javanmard2014confidence}, $\Omega$ is a feasible solution w.p.a.1. in (\ref{eq:clime}) when $\mu_\omega = C\sqrt{\log p/n}$ with some sufficiently large absolute constant $C$, i.e. $\|\hSigma\Omega - I_p\|_\infty \leq \mu_\omega$ w.p.a.1. By the definition of $\hOmega$ in (\ref{eq:clime})
\[\|\hat\Omega\|_1\leq \|\hat\Omega^{(1)}\|_1 \leq \|\Omega\|_1 \leq m_\omega\]
w.p.a.1, which verifies (\ref{eq:CLIMEestL1Bound}). Besides, 
\[ \begin{aligned}
    \|\hOmega^{(1)} - \Omega\|_\infty &\leq \|\Omega\|_1 \|\Sigma\hOmega^{(1)} - I_p\|_\infty \\
     &\leq m_\omega \left(\|(\hSigma-\Sigma)(\hOmega^{(1)} - \Omega)\|_\infty + \|\hSigma(\hOmega^{(1)} - \Omega)\|_\infty    \right)\\
     &\leq m_\omega \left((\|\hOmega^{(1)}\|_1 + \|\Omega\|_1)\cdot\|\hSigma-\Sigma\|_\infty + \|\hSigma\hOmega^{(1)}-I_p\|_\infty + \|\hSigma\Omega-I_p\|_\infty \right)\\
     &\lep m_\omega^2 \sqrt{\dfrac{\log p}{n}}.
\end{aligned} \] 
Also, by definition of (\ref{eq:clime}), any entry of $\hat\Omega$ also appears in $\hat\Omega^{(1)}$. Thus, 
\[\|\hOmega - \Omega\|_\infty \leq \|\hOmega^{(1)} - \Omega\|_\infty \lep m_\omega^2 \sqrt{\dfrac{\log p}{n}}.\]
Following the proof of (14) in Theorem 6 of \citet{cai2011constrained} we can deduce \[\|\hOmega-\Omega\|_1 \lep s_\omega\cdot( \|\hOmega - \Omega\|_\infty)^{1-q}\lep s_\omega m_\omega^{2-2q}\left(\dfrac{\log p}{n}\right)^{(1-q)/2},\]
which is (\ref{eq:CLIMEL1}). For (\ref{eq:CLIMEsup}), 
\[\begin{aligned}
    \|\hSigma\hOmega-I\|_\infty 
    &\leq \|\hSigma\Omega-I\|_\infty + \|\hSigma(\hOmega-\Omega)\|_\infty \\
    &\lep \sqrt{\dfrac{\log p}{n}} + \|\hSigma\|_\infty \|\hOmega-\Omega|\|_1 \\
    &\lep \sqrt{\dfrac{\log p}{n}} + (\|\hSigma-\Sigma\|_\infty + \|\Sigma\|_\infty)\cdot s_\omega\cdot m_\omega^{2-2q}\left(\dfrac{\log p}{n}\right)^{(1-q)/2}\\
    &\lep s_\omega\cdot m_\omega^{2-2q}\left(\dfrac{\log p}{n}\right)^{(1-q)/2}.
\end{aligned}\]
This completes the proof of Lemma \ref{lem:CLIME}.
\end{proof}
Lemma \ref{lem:asymp} shows a more convenient asymptotic regime used in the proofs. 
\begin{lemma}\label{lem:asymp}Under Assumption \ref{as:asym}
\begin{equation}\label{eq:lem:asymp}
    \dfrac{m_\omega^3 s^{3/2}(\log p)^{(7+\nu)/2}}{\sqrt{n}} = o(1\wedge\|\gamma\|_2). 
\end{equation}
\end{lemma}
\begin{proof}[Proof of Lemma \ref{lem:asymp}]By Assumption \ref{as:DecisionMatrix}, $\sqrt{\QAst(\gamma)} \asymp \|\gamma\|_2$ By Assumption \ref{as:asym}, we have 
\[\begin{aligned}
\left(\dfrac{m_\omega^3 s^{3/2}(\log p)^{(7+\nu)/2}}{n^{1/2}}\right)^{1-q} &= \dfrac{m_\omega^{3-3q}s^{(3-3q)/2}(\log p)^{[(7+\nu)(1-q)]/2}}{n^{(1-q)/2}} \\
&\leq \dfrac{m_\omega^{3-2q}s^{(3-q)/2}(\log p)^{(7+\nu-q)/2}}{n^{(1-q)/2}} = o(1\wedge \|\gamma\|_2).
\end{aligned}\]
By $0\leq q<1$ and $\left(\dfrac{m_\omega^3 s^{3/2}(\log p)^{(7+\nu)/2}}{n^{1/2}}\right)^{1-q} < 1$ with $n$ large enough, we have 
\[\dfrac{m_\omega^3 s^{3/2}(\log p)^{(7+\nu)/2}}{n^{1/2}} < \left(\dfrac{m_\omega^3 s^{3/2}(\log p)^{(7+\nu)/2}}{n^{1/2}}\right)^{1-q} = o(1\wedge \|\gamma\|_2),\]
as $n\to\infty$.
\end{proof}


Lemma \ref{lem:bound_sd} shows the probability bounds for the maximum norms that are useful to bound the estimation errors of asymptotic variance. 
\begin{lemma}\label{lem:bound_sd}Under Assumptions \ref{as:DecisionMatrix}, \ref{as:error} and \ref{as:asym},  

\begin{equation}\label{eq:W4}
\max_{j,k,\ell,m\in[p]}\left|\dfrac{1}{n}\sum_{i=1}^nW_{ij}W_{ik}W_{i\ell}W_{im} - \dfrac{1}{n}\sum_{i=1}^n \mathbb{E}\left(W_{ij}W_{ik}W_{i\ell}W_{im}\right)\right| \lep\sqrt{\dfrac{\log p}{n}},
\end{equation}
\begin{equation}\label{eq:W3epsilon}
\max_{j,k,h\in[p]}\left|\dfrac{1}{n}\sum_{i=1}^nW_{ij}W_{ik}W_{ih}\varepsilon_{im}\right| \lep\sqrt{\dfrac{\log p}{n}},
\end{equation}
and 
\begin{equation}\label{eq:W2epsilon2}
\max_{j,k\in[p]}\left|\dfrac{1}{n}\sum_{i=1}^nW_{ij}W_{ik}\left(\varepsilon_{i\ell}\varepsilon_{im}- \expect[\varepsilon_{i\ell}\varepsilon_{im}|W]\right)\right|\lep \sqrt{\dfrac{\log p}{n}},
\end{equation}
for $\ell,m=1,2$.
\end{lemma}
\begin{proof}[Proof of Lemma \ref{lem:bound_sd}]We only show (\ref{eq:W4}). The other two inequalities can be verified following the same procedures. By Assumption \ref{as:DecisionMatrix}, for any $j,k,\ell,m\in[p]$, we know that 
\[\Pr( |W_{ij}W_{ik}W_{i\ell}W_{im}| > \mu)\leq C\exp(-c\mu^{0.5}), \]
 for some absolute constants $C$ and $c$. By Theorem 1 of \citet{merlevede2011bernstein}, we know that for any $\mu>0$
 \[\begin{aligned}
 &\ \ \Pr\left(\left|\sum_{i=1}^n \left(W_{ij}W_{ik}W_{i\ell}W_{im} - \mathbb{E}(W_{ij}W_{ik}W_{i\ell}W_{im})\right) \right| > \mu\right) \\
 &\leq n\exp\left(-\dfrac{\mu^r}{C_1}\right) + \exp\left(-\dfrac{\mu^2}{C_2(1+nV)}\right) + \exp\left(-\dfrac{\mu^2}{C_3n}\exp\left(\dfrac{\mu^{r(1-r)}}{C_4(\log \mu)^r}\right)\right),
 \end{aligned}\]
 where $r = \left(\dfrac{1}{0.5} + \dfrac{1}{r_2}\right)^{-1} < 1$ as defined in (2.8) of the same paper. Here $1/r_2$ measures the mixing coefficient of a time series, which can be arbitrarily small for independent data. Taking $\mu = \sqrt{C_xn\log p}$ with $C_x = (2C_1)^{2/r}\vee (5C_2V)$. Then 
  \[\begin{aligned}
 &\ \ \Pr\left(\max_{j,k,\ell,m\in[p]}\left|\dfrac{1}{n}\sum_{i=1}^n \left(W_{ij}W_{ik}W_{i\ell}W_{im} - \mathbb{E}(W_{ij}W_{ik}W_{i\ell}W_{im})\right) \right| > \sqrt{\dfrac{C_x\log p}{n}}\right) \\
 &\leq np^4\exp\left(-\dfrac{(C_xn\log p)^{r/2}}{C_1}\right) + p^4\exp\left(-\dfrac{C_xn\log p}{C_2(1+nV)}\right) + \\
 &\ \ \ \ p^4\exp\left(-\dfrac{C_xn\log p}{C_3n}\exp\left(\dfrac{(C_xn\log p)^{r(1-r)/2}}{C_4(0.5 \log (C_xn\log p) )^r}\right)\right) \\
 &\leq np^4\exp\left(-2(n\log p)^{r/2}\right) + p^4\exp(-5\log p) + o(1) \\
  &\leq\exp\left(-(2(n\log p)^{r/2} - \log n - \log p) \right) + o(1), \\
 \end{aligned}\]
 where the second inequality applies that 
 \[\dfrac{(C_xn\log p)^{r(1-r)/2}}{C_4(0.5 \log (C_xn\log p) )^r}\to\infty.\]
 Obviously, $(n\log p)^{r/2} - \log n \to\infty$. Take $r_2 = 0.5$ and hence $r=0.25$ and $2/r-1 = 7$. We thus also have $(n\log p)^{r/2} - \log p\to\infty$ as $(\log p)^{2/r-1} = (\log p)^7 = o(n)$ by Lemma \ref{lem:asymp}. Hence,
 \[\begin{aligned}
 &\ \ \Pr\left(\max_{j,k,\ell,m\in[p]}\left|\dfrac{1}{n}\sum_{i=1}^n \left(W_{ij}W_{ik}W_{i\ell}W_{im} - \mathbb{E}(W_{ij}W_{ik}W_{i\ell}W_{im})\right) \right| > \sqrt{\dfrac{C_x\log p}{n}}\right)= o(1),
 \end{aligned}\]
 and (\ref{eq:W4}) follows.
\end{proof}

%% file: propositions.tex
\subsubsection{Essential Propositions} \label{app:b1}
Proposition \ref{prop:LassoTheta} provides probability upper bounds of the Lasso estimators of the reduced form estimators. 
\begin{proposition}\label{prop:LassoTheta}Suppose that Assumptions \ref{as:DecisionMatrix}, \ref{as:error} and \ref{as:tuning} (i) hold. If $s = o(\sqrt{n/\log p})$, we have 
\begin{equation}
\begin{aligned}
\max\{\|\widehat{\Gamma}-\Gamma\|_2,\|\widehat{\gamma}-\gamma\|_2,\|\widehat{\Psi}-\Psi\|_2,\|\widehat{\psi}-\psi\|_2\} &\lep \sqrt{\dfrac{s\log p}{n}}, \\
\max\{\|\widehat{\Gamma}-\Gamma\|_1,\|\widehat{\gamma}-\gamma\|_1,\|\widehat{\Psi}-\Psi\|_1,\|\widehat{\psi}-\psi\|_1\} &\lep \sqrt{\dfrac{s^2 \log p}{n}}.  
\end{aligned}
\end{equation}
\end{proposition}
\begin{proof}[Proof of Proposition \ref{prop:LassoTheta}]
    The results are directly implied by Lemma \ref{lem:Lasso},  (\ref{eq:WeBound}) and (\ref{eq:RE_lower}).
\end{proof}

  Proposition \ref{prop:A and A*} provides probability upper bounds of the weighting matrix $A$. 
 \begin{proposition}\label{prop:A and A*}Suppose that Assumption \ref{as:DecisionMatrix} holds. Then 
 \begin{equation}\label{eq:A - A*}
     \|A-A^*\|_2 + \|A^{1/2}-A^{*1/2}\|_2 + \|A-A^*\|_1 + \|A^{1/2}-A^{*1/2}\|_1 \lep \sqrt{\dfrac{\log p}{n}}.
 \end{equation}
 Furthermore, when $\log p = o(n)$, 
  \begin{equation}\label{eq: A lep 1}
     \|A\|_2 + \|A^{1/2}\|_2 + \|A\|_1 + \|A^{1/2}\|_1 \lep 1.
 \end{equation}
and 
\begin{equation}\label{eq: A min lower}
     \lambda_{\min}(A) \gep 1.
 \end{equation}
\end{proposition} 
\begin{proof}[Proof of Proposition \ref{prop:A and A*}]By definitions of $A$ and $A^*$ in (\ref{eq: def A}) and (\ref{eq:def Astar}),  
\begin{equation}\label{eq: A temp}
    \|A-A^*\|_2 + \|A-A^*\|_1 \leq 2\|\hSigma-\Sigma\|_\infty \lep \sqrt{\dfrac{\log p}{n}}.
\end{equation}
  
Hence, 
\[\lambda_{\min}(A) \geq  \lambda_{\min}(A^*) - \|A-A^*\|_2 \gep 1,\]
which verifies (\ref{eq: A min lower}). Besides, 
\[\begin{aligned}
 \|A^{1/2}-A^{*1/2}\|_2 =  \|A^{1/2}-A^{*1/2}\|_1 &= \max_{j\in[p_z]}\left|\sqrt{n^{-1}\sumn Z_{ij}^2} - \sqrt{\mathbb{E}(Z_{ij}^2)} \right| \\
 &\leq \max_{j\in[p_z]}\dfrac{\left|n^{-1}\sumn Z_{ij}^2-\mathbb{E}(Z_{ij}^2)\right|}{\sqrt{n^{-1}\sumn Z_{ij}^2}+\sqrt{\mathbb{E}(Z_{ij}^2)}} \\
 &\leq \dfrac{\|\hSigma-\Sigma\|_\infty}{\sqrt{\lambda_{\min}(A)}+\sqrt{\lambda_{\min}(A^*)}} \lep \sqrt{\dfrac{\log p}{n}},
\end{aligned}\]
which, together with (\ref{eq: A temp}), induces (\ref{eq:A - A*}). Then (\ref{eq: A lep 1}) directly follows (\ref{eq:A - A*}) and the result that 
\[\|A^*\|_2+\|A^{*1/2}\|+\|A^*\|_1+\|A^{*1/2}\|_1\lesssim 1.\]
\end{proof}

Proposition \ref{prop:bound_sd_hat} provides some error bounds that are useful in deriving estimation error of the asymptotic variance. Define \[\sigma_{iA}^2 = \sigma_{i,Y}^2 - 2\beta_A \sigma_{i,YD} + \beta_A^2\sigma_{i,D}^2.\] Similarly, define \[\sigma_{iA^*}^2 = \sigma_{i,Y}^2 - 2\beta_{A^*} \sigma_{i,YD} + \beta_{A^*}^2\sigma_{i,D}^2\] where  $\beta_{A^*}=\dfrac{\IAst(\gamma,\Gamma)}{\QAst(\gamma)}$ is defined below (\ref{eq: HAstar}).
\begin{proposition} \label{prop:bound_sd_hat}Under Assumptions \ref{as:DecisionMatrix}-\ref{as:tuning}, if $\pi\in\mathcal{H}_M(t)$ for any absolute constant $t$, 
\begin{align}
&\left\|\dfrac{1}{n}\sumn \Wi\Wi^\top (\sigma^2_{iA} - \sigma^2_{iA^*})\right\|_\infty + \max_{i\in[n]}|\sigma^2_{iA} - \sigma^2_{iA^*}| \lep \sqrt{\dfrac{\log p}{n}},\label{eq:bound_hat_Wsigma_star}\\
&\left\|\dfrac{1}{n}\sumn \Wi\Wi^\top \sigma^2_{iA}\right\|_\infty + \max_{\ell,m\in\{1,2\}}\left\|\dfrac{1}{n}\sumn \Wi\Wi^\top \mathbb{E}(\varepsilon_{i\ell}\varepsilon_{im}|W_{i\cdot})\right\|_\infty \lep 1,\label{eq:bound_hat_Wsigma}\\
& \max_{\ell,m\in\{1,2\}}\left\|\dfrac{1}{n}\sumn \Wi\Wi^\top \varepsilon_{i\ell}\varepsilon_{im}\right\|_\infty \lep 1,\label{eq:bound_hat_Weps}\\
&\left\|\dfrac{1}{n}\sumn \Wi\Wi^\top(\widehat{\varepsilon}_{i\ell}\widehat{\varepsilon}_{im}-\mathbb{E}(\varepsilon_{i\ell}\varepsilon_{im}|W))\right\|_\infty  \lep \dfrac{s^2\log p}{n} + \sqrt{\dfrac{\log p}{n}}, \text{ for }\ell,m=1,2, \label{eq:bound_hat_1} 
\end{align} 
\end{proposition}
\begin{proof}[Proof of Proposition \ref{prop:bound_sd_hat}] \underline{Proof of (\ref{eq:bound_hat_Wsigma_star})}. We first need a bound for $\beta_{A}-\beta_{A^*}$. Note that when $\pi\in\mathcal{H}_M(t)$, 
\begin{equation}\label{eq: equi L2 norm}
    \|\pi\|_2 = \dfrac{\|\pi_{A^*}\|_2}{\sqrt{1-{\rm R}_{A^*}^2(\pi,\gamma)}} \asymp \|\pi_{A^*}\|_2,
\end{equation}
and hence
\begin{equation}\label{eq: pi 2 norm bound HM}
   \|\pi\|_2\lesssim \sqrt{s}\|A^{*1/2}\pi_{A^*}\|_\infty\lesssim \sqrt{\dfrac{s\log p}{n}}. 
\end{equation}
Thus, by Lemma \ref{lem:asymp} $\|\pi\|_2 \lesssim \|\gamma\|_2$. This implies  
\begin{equation}\label{eq: beta Ast bound}
    |\beta_{A^*}|\lesssim\dfrac{\|\Gamma\|_2\|\gamma\|_2}{\|\gamma\|_2^2}\lesssim \dfrac{\|\pi\|_2+|\beta|\cdot\|\gamma\|_2}{\|\gamma\|_2}\lesssim 1+\dfrac{\|\pi\|_2}{\|\gamma\|_2} \lesssim 1,
\end{equation}
and by Proposition \ref{prop:A and A*}
\[\begin{aligned}
|\IA(\gamma,\Gamma)-\IAst(\gamma,\Gamma)|&\leq \|\Gamma\|_2\|A-A^*\|_2\|\gamma\|_2 \\
&\lep \|\pi+\gamma\beta\|_2\|\gamma\|_2 \sqrt{\dfrac{\log p}{n}}\\
&\lesssim \left(\|\pi\|_2\|\gamma\|_2 + \|\gamma\|_2^2\right)\sqrt{\dfrac{\log p}{n}}\lesssim \QAst(\gamma)\sqrt{\dfrac{\log p}{n}},
\end{aligned}\]
and 
\[\begin{aligned}
|\QA(\gamma)-\QAst(\gamma)|&\leq \|A-A^*\|_2\|\gamma\|_2^2 \lep \QAst(\gamma)\sqrt{\dfrac{\log p}{n}}. 
\end{aligned}\]
which implies $\QA(\gamma)/\QAst(\gamma)\convp1$. 
We then deduce that 
\begin{equation}\label{eq: betaA betaAst}
    |\beta_A - \beta_{A^*}| = \left|\dfrac{\IA(\gamma,\Gamma)-\IAst(\gamma,\Gamma) - \beta_{A^*}(\QA(\gamma)-\QAst(\gamma))}{\QA(\gamma)}\right| \lep \sqrt{\dfrac{\log p}{n}},
\end{equation} 
which together with (\ref{eq: beta Ast bound}) also implies 
\begin{equation}\label{eq: beta A bound}
    |\beta_{A}| \lesssim_p 1. 
\end{equation}
In addition, we have 
\begin{equation}\label{eq: betaA2 betaAst2}
    |\beta_A^2 - \beta_{A^*}^2| =  |\beta_A - \beta_{A^*}| \cdot |\beta_A + \beta_{A^*} | \lep \sqrt{\dfrac{\log p}{n}}. 
\end{equation} 
Finally, by the  of $\sigma_{i,D}$ specified in Assumption \ref{as:error}, each entry of $\Wi\sigma_{i,D}$ is also sub-Gaussian. Thus $\|\mathbb{E}(\Wi \Wi^\top  \sigma_{i,D}^2)\|_\infty$ is uniformly bounded. Following the proof of (\ref{eq:W4}) we deduce that 
\begin{equation}\label{eq: DB WWsigma}
\left\|\dfrac{1}{n}\sumn\left[\Wi\Wi^\top \sigma^2_{i,D}-\mathbb{E}(\Wi \Wi^\top \sigma_{i,D}^2)\right]\right\|_\infty \lep \sqrt{\dfrac{\log p}{n}},
\end{equation}
and hence 
\begin{equation}\label{eq:WWsigma2Bound1}
\begin{aligned}
\left\|\dfrac{1}{n}\sumn \Wi\Wi^\top \sigma^2_{i,D}\right\|_\infty 
&\leq   \left\|\dfrac{1}{n}\sumn\left[\Wi\Wi^\top \sigma^2_{i,D}-\mathbb{E}(\Wi \Wi^\top \sigma_{i,D}^2)\right]\right\|_\infty  + \left\|\dfrac{1}{n}\sumn \mathbb{E}(\Wi \Wi^\top \sigma_{i,D}^2) \right\|_\infty  \\
&\lep \sqrt{\dfrac{\log p}{n}} + 1 \lesssim 1.
\end{aligned}
\end{equation}
Similarly, 
\begin{equation}\label{eq:WWsigma12Bound1}
    \left\|\dfrac{1}{n}\sumn \Wi\Wi^\top \sigma_{i,YD}\right\|_\infty \lep 1. 
\end{equation}

Then by (\ref{eq: betaA betaAst}), (\ref{eq: betaA2 betaAst2}) and (\ref{eq:WWsigma2Bound1}), 
\[\begin{aligned}
\left\|\dfrac{1}{n}\sumn \Wi\Wi^\top (\sigma^2_{iA} - \sigma^2_{iA^*})\right\|_\infty &\lesssim |\beta_A^2 - \beta_{A^*}^2| \cdot \left\|\dfrac{1}{n}\sumn \Wi\Wi^\top \sigma^2_{i,D} \right\|_\infty + |\beta_A - \beta_{A^*}|\cdot \left\|\dfrac{1}{n}\sumn \Wi\Wi^\top \sigma_{i,YD}  \right\|_\infty \\&\lep \sqrt{\dfrac{\log p}{n}}
\end{aligned},\]
and 
\[\begin{aligned}
\max_{i\in[n]}\left|\sigma^2_{iA} - \sigma^2_{iA^*}\right|_\infty &\lesssim |\beta_A^2 - \beta_{A^*}^2| \cdot \max_{i\in[n]}\sigma_{i,D}^2 + |\beta_A - \beta_{A^*}|\cdot|\max_{i\in[n]} \sigma_{i,YD}| \lep \sqrt{\dfrac{\log p}{n}}.
\end{aligned}\]

\par \underline{Proof of (\ref{eq:bound_hat_Wsigma})}. We only show the upper bound of the first term on the LHS since the second term goes through similarly. By the boundness of $\sigma_{iA^*}$, each entry of $\Wi\sigma_{iA^*}$ is also sub-Gaussian. Thus  $\|\mathbb{E}(\Wi \Wi^\top  \sigma_{iA^*}^2)\|_\infty$ is uniformly bounded. Following the proof of  (\ref{eq:W4}) we deduce that 
\begin{equation}\label{eq: DB WWsigma A}
\left\|\dfrac{1}{n}\sumn\left[\Wi\Wi^\top \sigma^2_{iA^*}-\mathbb{E}(\Wi \Wi^\top \sigma_{iA^*}^2)\right]\right\|_\infty \lep \sqrt{\dfrac{\log p}{n}},
\end{equation}
and hence 
\[\begin{aligned}
&\ \ \ \ \left\|\dfrac{1}{n}\sumn \Wi\Wi^\top \sigma^2_{iA}\right\|_\infty \\
&\leq  \left\|\dfrac{1}{n}\sumn \Wi\Wi^\top (\sigma^2_{iA} - \sigma^2_{iA^*})\right\|_\infty +  \left\|\dfrac{1}{n}\sumn\left[\Wi\Wi^\top \sigma^2_{iA^*}-\mathbb{E}(\Wi \Wi^\top \sigma_{iA^*}^2)\right]\right\|_\infty  + \left\|\dfrac{1}{n}\sumn \mathbb{E}(\Wi \Wi^\top \sigma_{iA^*}^2) \right\|_\infty  \\
&\lep \sqrt{\dfrac{\log p}{n}} + 1 \lesssim 1.
\end{aligned}\]
\par \underline{Proof of (\ref{eq:bound_hat_Weps})}. It immediately follows by (\ref{eq:W2epsilon2}) and (\ref{eq:bound_hat_Wsigma}) that 
\[\begin{aligned}
&\ \ \ \ \left\|\dfrac{1}{n}\sumn \Wi\Wi^\top \varepsilon_{i\ell}\varepsilon_{im}\right\|_\infty\\
&\lesssim \left\|\dfrac{1}{n}\sumn \Wi\Wi^\top \left(\varepsilon_{i\ell}\varepsilon_{im}-\mathbb{E}[\varepsilon_{i\ell}\varepsilon_{im}|W]\right)\right\|_\infty + \left\|\dfrac{1}{n}\sumn \Wi\Wi^\top \mathbb{E}[\varepsilon_{i\ell}\varepsilon_{im}|W]\right\|_\infty \\
&\lep \sqrt{\dfrac{\log p}{n}} + 1 \lep 1.
\end{aligned}\]
\par \underline{Proof of (\ref{eq:bound_hat_1})}. We only prove the case with $\ell=m=1$. Other cases can be verified in the same manner. Recall that $\sigma_{i,Y}^2 = \expect(\varepsilon_{i,Y}^2|W_{i\cdot})$. Note that
\begin{equation}\label{eq:WWhate temp}
    \begin{aligned}
\hat\varepsilon_{i,Y}^2 - \sigma_{i,Y}^2 &=  \hat\varepsilon_{i,Y}^2 - \varepsilon_{i,Y}^2 +  \varepsilon_{i,Y}^2 - \sigma_{i,Y}^2 \\
&= ( \hat\varepsilon_{i,Y} - \varepsilon_{i,Y} )^2 + 2\varepsilon_{i,Y}(\hat\varepsilon_{i,Y} - \varepsilon_{i,Y}) +  \varepsilon_{i,Y}^2 - \sigma_{i,Y}^2 \\
&= \left( X_{i\cdot}^\top(\hat\Psi-\Psi) + Z_{i\cdot}^\top(\hat\Gamma-\Gamma) \right)^2 + 2\varepsilon_{i,Y}(X_{i\cdot}^\top(\hat\Psi-\Psi) + Z_{i\cdot}^\top(\hat\Gamma-\Gamma)) +  \varepsilon_{i,Y}^2 - \sigma_{i,Y}^2 \\
&= \left( W_{i\cdot}^\top \left( \begin{array}{c}
     \hat\Psi-\Psi  \\
     \hat\Gamma-\Gamma 
\end{array} \right) \right)^2 + 2\varepsilon_{i,Y}W_{i\cdot}^\top \left( \begin{array}{c}
     \hat\Psi-\Psi  \\
     \hat\Gamma-\Gamma 
\end{array} \right) +  \varepsilon_{i,Y}^2 - \sigma_{i,Y}^2,
\end{aligned}
\end{equation}
and hence 
\[\begin{aligned}
&\ \ \ \left\|\dfrac{1}{n}\sumn \Wi\Wi^\top(\widehat{\varepsilon}_{i,Y}^2-\sigma_{i,Y}^2)\right\|_\infty \\
&\lesssim \left\|\dfrac{1}{n}\sumn \Wi\Wi^\top \left( W_{i\cdot}^\top \left( \begin{array}{c}
     \hat\Psi-\Psi  \\
     \hat\Gamma-\Gamma 
\end{array} \right) \right)^2\right\|_\infty + \max_{j,k,h\in[p]}\left|\dfrac{1}{n}\sum_{i=1}^nW_{ij}W_{ik}W_{ih}\varepsilon_{im}\right| \cdot \left(\|\hPsi-\Psi\|_1+\|\hGamma-\Gamma\|_1\right)\\
&\ \ \ + \left\|\dfrac{1}{n}\sumn W_{i\cdot}W_{i\cdot}(\varepsilon_{i,Y}^2 - \sigma_{i,Y}^2)\right\|_\infty. \\ 
\end{aligned}\]
Note that the first term on the RHS of (\ref{eq:WWhate temp}) can be written as 
\[\begin{aligned}
\left\|\dfrac{1}{n}\sumn \Wi\Wi^\top \left( W_{i\cdot}^\top \left( \begin{array}{c}
     \hat\Psi-\Psi  \\
     \hat\Gamma-\Gamma 
\end{array} \right) \right)^2\right\|_\infty  &=
\left\|\dfrac{1}{n}\sumn {\rm vec}\left[\Wi W_{i\cdot}^\top \left( \begin{array}{c}
     \hat\Psi-\Psi  \\
     \hat\Gamma-\Gamma 
\end{array} \right)\left( \begin{array}{c}
     \hat\Psi-\Psi  \\
     \hat\Gamma-\Gamma 
\end{array} \right)^\top \Wi^\top \Wi\right]\right\|_\infty \\
&= \left\|\dfrac{1}{n}\sumn \Wi W_{i\cdot}^\top\otimes \Wi W_{i\cdot}^\top {\rm vec}\left[ \left( \begin{array}{c}
     \hat\Psi-\Psi  \\
     \hat\Gamma-\Gamma 
\end{array} \right)\left( \begin{array}{c}
     \hat\Psi-\Psi  \\
     \hat\Gamma-\Gamma 
\end{array} \right)^\top  \right]\right\|_\infty \\
&\leq \left\|\dfrac{1}{n}\sumn \Wi W_{i\cdot}^\top\otimes \Wi W_{i\cdot}^\top\right\|_\infty \left(\|\hat\Psi-\Psi\|_1 + \|\hat\Gamma-\Gamma\|_1\right)^2 \\
&\lep \max_{j,k,\ell,m\in[p]}\left|\dfrac{1}{n}\sum_{i=1}^n  W_{ij}W_{ik}W_{i\ell}W_{im}  \right| \dfrac{s^2\log p}{n}.
\end{aligned}\]
where the last inequality applies Proposition \ref{prop:LassoTheta}. By sub-Gaussianity in Assumption \ref{as:DecisionMatrix}, the fourth-moment $|\mathbb{E}(W_{ij}W_{ik}W_{i\ell}W_{im})|$ is uniformly bounded by some absolute constant. Then by (\ref{eq:W4})
\[\begin{aligned}
&\ \ \ \left\|\dfrac{1}{n}\sumn {\rm vec}\left[\Wi W_{i\cdot}^\top \left( \begin{array}{c}
     \hat\Psi-\Psi  \\
     \hat\Gamma-\Gamma 
\end{array} \right)\left( \begin{array}{c}
     \hat\Psi-\Psi  \\
     \hat\Gamma-\Gamma 
\end{array} \right)^\top \Wi^\top \Wi\right]\right\|_\infty \\
&\lep  \max_{j,k,\ell,m\in[p]}\left|\dfrac{1}{n}\sum_{i=1}^n\left(  W_{ij}W_{ik}W_{i\ell}W_{im} - \mathbb{E}(W_{ij}W_{ik}W_{i\ell}W_{im} )\right) \right| \dfrac{s^2\log p}{n}\\
&\ \ \ \ +\max_{j,k,\ell,m\in[p]}|\mathbb{E}(W_{ij}W_{ik}W_{i\ell}W_{im})| \dfrac{s^2\log p}{n} \\
&\lep \left(1+\sqrt{\dfrac{s\log p}{n}}\right)\dfrac{s^2\log p}{n} \lesssim \dfrac{s^2\log p}{n}.
\end{aligned} \]
As for the last two terms of (\ref{eq:WWhate temp}), by (\ref{eq:W3epsilon}), (\ref{eq:W2epsilon2}) and Lemma \ref{lem:asymp}, 
\[\max_{j,k,h\in[p]}\left|\dfrac{1}{n}\sum_{i=1}^nW_{ij}W_{ik}W_{ih}\varepsilon_{i,Y}\right| \cdot \left(\|\hPsi-\Psi\|_1+\|\hGamma-\Gamma\|_1\right) \lep \sqrt{\dfrac{\log p}{n}}\dfrac{s\sqrt{\log p}}{\sqrt{n}} \leq \sqrt{\dfrac{\log p}{n}},\]
and 
\[\left\|\dfrac{1}{n}\sumn W_{i\cdot}W_{i\cdot}(\varepsilon_{i,Y}^2 - \sigma_{i,Y}^2)\right\|_\infty \lep  \sqrt{\dfrac{\log p}{n}}.\]
Then (\ref{eq:bound_hat_1}) follows. 
\end{proof}

%% file: propositions2.tex
\subsubsection{Essential Propositions}\label{app:b31}

Proposition \ref{prop:Xihat} provides Lasso estimation errors of the identified parameters $\pi_A$ that measure IV validity. 
\begin{proposition}\label{prop:Xihat}Under Assumptions \ref{as:DecisionMatrix}-\ref{as:tuning},  
\begin{equation}\label{eq:XihatError}
\begin{aligned}
\max\{\|\widehat{\pi}_A-\widecheck{\pi}_A\|_2, \|\widehat{\varphi}_A-\widecheck{\varphi}_A\|_2\}&\lep \left(1+\pigamma\right)\sqrt{\dfrac{s\log p}{n}}, \\
\max\{\|\widehat{\pi}_A-\widecheck{\pi}_A\|_1, \|\widehat{\varphi}_A-\widecheck{\varphi}_A\|_1\}&\lep \left(1+\pigamma\right) \sqrt{\dfrac{s^2 \log p}{n}}.
\end{aligned}
\end{equation}
\end{proposition} 
\begin{proof}[Proof of Proposition \ref{prop:Xihat}]
    By Lemma \ref{lem:Lasso} and (\ref{eq:RE_lower}), it suffices to show that 
    \[\|n^{-1}W^\top\widecheck e_A\|_\infty\lep \left(1+\pigamma\right)\sqrt{\dfrac{\log p}{n}}.\]
    By (\ref{eq:WeBound}), we have
\[\begin{aligned}
\left|\dfrac{u_\gamma^\top W^\top e_A}{\sqrt{n}}\right| &\leq\|\gamma\|_1\|A\|_1\|\Omega\|_1\cdot\left(\left\|\dfrac{W^\top \varepsilon_Y}{\sqrt{n}}\right\|_\infty+|\beta_A|\cdot\left\|\dfrac{W^\top \varepsilon_D}{\sqrt{n}}\right\|_\infty\right) \\
&\lep \|\gamma\|_2\cdot \left(m_\omega\sqrt{s\log p} + \pigamma m_\omega\sqrt{s\log p}\right) \\
&\lep (\|\pi\|_2  + \|\gamma\|_2)\cdot m_\omega\sqrt{s\log p}. \\
\end{aligned}\]
By (\ref{eq:consistent_Qgamma}), (\ref{eq: beta hat decom with pi}) and (\ref{eq:bound u pi A W e2}), 
\begin{equation}\label{eq:betahat_rate with pi}
\begin{aligned}
    \hat\beta_A - \beta_A &= O_p\left(\dfrac{(\|\pi\|_2 + \|\gamma\|_2)\cdot m_\omega\sqrt{s\log p}}{\sqrt{n}\QA(\gamma)}\right) +  \\
    &\ \ \ \ \dfrac{1}{\sqrt{n}\QA(\gamma)}O_p\left(\dfrac{m_\omega s\log p}{\sqrt{n}} + \dfrac{s_\omega m_{\omega}^{2-2q}\cdot s^{1/2}(\log p)^{(1-q)/2}}{n^{(1-q)/2}}(\|\gamma\|_2+\|\pi\|_2)\right) \\
    &= O_p\left(\dfrac{(\|\pi\|_2 + \|\gamma\|_2)\cdot m_\omega\sqrt{s\log p}}{\sqrt{n}\QA(\gamma)}\right) + \dfrac{1}{\sqrt{n}\QA(\gamma)}O_p\left(\|\gamma\|_2 + (\|\gamma\|_2+\|\pi\|_2)\right)\\
    &=O_p\left(\dfrac{(\|\pi\|_2 + \|\gamma\|_2)\cdot m_\omega\sqrt{s\log p}}{\sqrt{n}\QA(\gamma)}\right),
\end{aligned} 
\end{equation}
where the last two steps apply Assumption \ref{as:asym} and Lemma \ref{lem:asymp}.  
    Given that $n^{-1/2} = o(\|\gamma\|_2)$ implied by Lemma \ref{lem:asymp}, 
    \[\dfrac{(\|\pi\|_2 + \|\gamma\|_2)\cdot m_\omega\sqrt{s\log p}}{\sqrt{n}\QA(\gamma)} \lep \left(\dfrac{ m_\omega\sqrt{s\log p}}{\sqrt{n}\|\gamma\|_2}\right)\dfrac{\|\pi\|_2 + \|\gamma\|_2}{\|\gamma\|_2} \lep 1+\pigamma,\] 
    and hence 
 \begin{equation}\label{eq:betahat_bound with pi}
\begin{aligned}
    |\hat\beta_A| \lep |\beta_A| +  1+\pigamma \lep \left(|\beta|+\pigamma\right) +   1+\pigamma \lesssim 1 + \pigamma.
\end{aligned} 
\end{equation}   
  The above, together with (\ref{eq:WeBound}), implies  
    \[\begin{aligned}
        \|n^{-1}W^\top\widecheck e_A\|_\infty &\leq \|n^{-1}W^\top\varepsilon_Y\|_\infty + |\hbeta_A|\cdot \|n^{-1}W^\top\varepsilon_D\|_\infty \\
        &\lep \sqrt{\dfrac{\log p}{n}} + \left(1+\pigamma\right) \sqrt{\dfrac{\log p}{n}}\\
        &\lep  \left(1+\pigamma\right) \sqrt{\dfrac{\log p}{n}}.
    \end{aligned}\]
    This completes the proof of Proposition \ref{prop:Xihat}. 
\end{proof}


\begin{remark}
When $\pi\in\mathcal{H}_{A^*}(t)$, by (\ref{eq: pi 2 norm bound HM}) and Lemma \ref{lem:asymp} we have $\|\pi\|_2\lesssim \|\gamma\|_2$. Then the convergence rate becomes
\begin{equation}\label{eq:XihatError_simple}
\begin{aligned}
\max\{\|\widehat{\pi}_A-\widecheck{\pi}_A\|_2, \|\widehat{\varphi}_A-\widecheck{\varphi}_A\|_2\}&\lep  \sqrt{\dfrac{s\log p}{n}}, \\
\max\{\|\widehat{\pi}_A-\widecheck{\pi}_A\|_1, \|\widehat{\varphi}_A-\widecheck{\varphi}_A\|_1\}&\lep  \sqrt{\dfrac{s^2 \log p}{n}},
\end{aligned}
\end{equation}
as usual for Lasso estimators. 
\end{remark}

\begin{proposition} \label{prop:bound_sd_hat_for_pi}Under Assumptions \ref{as:DecisionMatrix}-\ref{as:tuning}, if $\pi\in\mathcal{H}_M(t)$ for any absolute constant $t$, 
\begin{align}
\left\|\dfrac{1}{n}\sumn \left(\Wi\Wi^\top\widehat{e}_{iA}^2-\mathbb{E}(\Wi\Wi^\top\sigma_{iA^*}^2)\right)\right\|_\infty \lep \dfrac{s^2\log p}{n} + \left(1+\dfrac{1}{\|\gamma\|_2}\right)\sqrt{\dfrac{\log p}{n}}. \label{eq:bound_hat_eA} 
\end{align} 
\end{proposition}
\begin{proof}[Proof of Proposition \ref{prop:bound_sd_hat_for_pi}] Note that 
\[\begin{aligned}
\dfrac{1}{n}\sumn \left(\Wi\Wi^\top\widehat{e}_{iA}^2-\mathbb{E}(\Wi\Wi^\top\sigma_{iA^*}^2)\right) &= \dfrac{1}{n}\sumn \Wi\Wi^\top (\widehat{e}_{iA}^2 - \sigma_{iA}^2) + \dfrac{1}{n}\sumn \Wi\Wi^\top(\sigma_{iA}^2 - \sigma_{iA^*}^2)   \\
&\ \ \ \ + \dfrac{1}{n}\sumn(\Wi\Wi^\top\sigma_{iA^*}^2 - \mathbb{E}(\Wi\Wi^\top \sigma_{iA^*}^2)).
\end{aligned} \]
We decompose $\widehat{e}_{iA}^2-\sigma_{iA}^2$ as 
\[\begin{aligned}
\widehat{e}_{iA}^2-\sigma_{iA}^2 &= \widehat{e}_{iA}^2-\widecheck{e}_{iA}^2+\widecheck{e}_{iA}^2-e_{iA}^2+e_{iA}^2-\sigma_{iA}^2 \\
&= \left[\Wi^\top\left(\begin{array}{c}
     \widehat{\varphi}_A-\widecheck{\varphi}_A  \\
      \widehat{\pi}_A-\widecheck{\pi}_A
\end{array}\right)+\widecheck{e}_{iA}\right]^2-\widecheck{e}_{iA}^2 + (\widehat{\beta}_A^2-\beta_A^2)\varepsilon_{i,D}^2 - 2(\widehat{\beta}_A-\beta_A)\varepsilon_{i,Y}\varepsilon_{i,D} + e_{iA}^2-\sigma_{iA}^2 \\ 
&= \left[\Wi^\top\left(\begin{array}{c}
     \widehat{\varphi}_A-\widecheck{\varphi}_A \\
      \widehat{\pi}_A-\widecheck{\pi}_A
\end{array}\right)\right]^2-2\Wi^\top\left(\begin{array}{c}
     \widehat{\varphi}_A-\widecheck{\varphi}_A  \\
      \widehat{\pi}_A-\widecheck{\pi}_A
\end{array}\right)(\varepsilon_{i,Y}-\widehat{\beta}_A\varepsilon_{i,D}) + (\widehat{\beta}^2_A-\beta^2_A)\varepsilon_{i,D}^2 \\
&\ \ \ - 2(\widehat{\beta}_A-\beta_A)\varepsilon_{i,Y}\varepsilon_{i,D} + (\varepsilon_{i,Y}^2-\sigma_{i,Y}^2) + \beta_A^2(\varepsilon_{i,D}^2-\sigma_{i,D}^2) - 2\beta_A(\varepsilon_{i,Y}\varepsilon_{i,Y} - \sigma_{i,YD}). 
\end{aligned}\]
Then 
\[\dfrac{1}{n}\sumn \left(\Wi\Wi^\top\widehat{e}_{iA}^2-\mathbb{E}(\Wi\Wi^\top\sigma_{iA}^2)\right)=\Delta^A_1+\Delta^A_2+\Delta^A_3+\Delta^A_4+\Delta^A_5,\]
where
\[\Delta^A_1=\dfrac{1}{n}\sumn \Wi\left[\Wi^\top\left(\begin{array}{c}
     \widehat{\varphi}_A-\widecheck{\varphi}_A  \\
     \widehat{\pi}_A-\widecheck{\pi}_A
\end{array}\right)\right]^2\Wi^\top,\]
\[\Delta^A_2=\dfrac{1}{n}\sumn \Wi\Wi^\top\left(\begin{array}{c}
     \widehat{\varphi}_A-\widecheck{\varphi}_A  \\
     \widehat{\pi}_A-\widecheck{\pi}_A
\end{array}\right)(\varepsilon_{i,Y}-\widehat{\beta}_A\varepsilon_{i,D})\Wi^\top,\]
\[\begin{aligned}
\Delta^A_3&=\dfrac{1}{n}\sumn \Wi\Wi^\top\left[(\widehat{\beta}_A^2-\beta_A^2)\varepsilon_{i,D}^2- 2(\widehat{\beta}_A-\beta_A)\varepsilon_{i,Y}\varepsilon_{i,D}\right], 
\end{aligned}\]
\[\Delta^A_4=\dfrac{1}{n}\sumn \Wi\Wi^\top\left[(\varepsilon_{i,Y}^2-\sigma_{i,Y}^2) + \beta_A^2(\varepsilon_{i,D}^2-\sigma_{i,D}^2) - 2\beta_A(\varepsilon_{i,Y}\varepsilon_{i,Y} - \sigma_{i,YD})\right],\]
and 
\[\Delta^A_5=\dfrac{1}{n}\sumn \Wi\Wi^\top(\sigma_{iA}^2 - \sigma_{iA^*}^2) + \dfrac{1}{n}\sumn(\Wi\Wi^\top\sigma_{iA^*}^2 - \mathbb{E}(\Wi\Wi^\top \sigma_{iA^*}^2)).\]
\par \underline{Bound $\Delta^A_1$.} Following similar arguments to show (\ref{eq:bound_hat_1}), by (\ref{eq:XihatError_simple})
\[\begin{aligned}
\|\Delta^A_1\|_\infty 
&\leq \max_{j,k,\ell,m\in[p]}\left|\dfrac{1}{n}\sum_{i=1}^n  W_{ij}W_{ik}W_{i\ell}W_{im}  \right|  (\|\widehat{\varphi}_A-\widecheck{\varphi}_A\|_1+\|\widehat{\pi}_A-\widecheck{\pi}_A\|_1)^2 \lep \dfrac{s^2\log p}{n}.
\end{aligned}\]
\par \underline{Bound $\Delta^A_2$.} By (\ref{eq: beta A bound}) and (\ref{eq:betahat_rate}), 
        $|\hbeta_A| \lep |\beta_A| + O_p\left(\dfrac{1}{\sqrt{n}\|\gamma\|_2}\right) \lep 1$. Then following similar arguments to show (\ref{eq:bound_hat_1}),
\[\begin{aligned}
\|\Delta^A_2\|_\infty 
&\lep \max_{(j,k,h)\in[p]^3,m\in\{1,2\}}\left|\dfrac{1}{n}\sum_{i=1}^nW_{ij}W_{ik}W_{ih}\varepsilon_{im}\right|(\|\widehat{\varphi}_A-\widecheck{\varphi}_A\|_1+\|\widehat{\pi}_A-\widecheck{\pi}_A\|_1)\cdot(1+|\hat\beta_A|)\\
&\lep \dfrac{s\log p}{n}.
\end{aligned}\]  
\par \underline{Bound $\Delta^A_3$.} Then by (\ref{eq:bound_hat_Weps}) and (\ref{eq:betahatsq_rate}),
\[\begin{aligned}
\|\Delta^A_3\|_\infty &\lep \left[|\widehat{\beta}_A^2-\beta_A^2|+|\widehat{\beta}_A-\beta_A|\right]\cdot 1 \lep \dfrac{\sqrt{\log p}}{\sqrt{n}\|\gamma\|_2}.
\end{aligned}\]
\par \underline{Bound $\Delta^A_4$.}  By  (\ref{eq:W2epsilon2}) and (\ref{eq: beta A bound}), 
\[\begin{aligned}
\|\Delta^A_4\|_\infty &\lep (1+|\beta_A|+|\beta_A|^2 )\sqrt{\dfrac{\log p}{n}} \lep \sqrt{\dfrac{\log p}{n}}.
\end{aligned}\]
\par \underline{Bound $\Delta^A_5$.} By (\ref{eq:bound_hat_Wsigma_star}) and (\ref{eq: DB WWsigma}), 
\[\|\Delta^A_5\|_\infty \lep \sqrt{\dfrac{\log p}{n}}.\]
Then we complete the proof of (\ref{eq:bound_hat_eA}) by summing up the upper bounds of $\Delta^A_1$, $\Delta^A_2$, $\Delta^A_3$, $\Delta^A_4$ and $\Delta^A_5$. 
\end{proof}

 Proposition \ref{prop:A} provides an intermediate result for lower bounded individual variances of the test statistic for the M test. 
 \begin{proposition}\label{prop:A}Let $A_{0,j}^{*\top}$ denote the $j$-th row of the matrix $A_0^*$ defined as (\ref{eq: def A0 star}). Suppose that Assumption \ref{as: two strong IV} holds. Then $\min_{j\in[p_z]}\|A_{0,j}^*\|_2^2 \gtrsim 1$. 
\end{proposition} 
\begin{proof}[Proof of Proposition \ref{prop:A}]Note that $I_{p_z} - \dfrac{\gamma\gamma^\top A^*}{\QAst(\gamma)}$ is idempotent and hence 
\[A_0^* A_0^{*\top} = A^{*1/2}\left(I_{p_z} - \dfrac{\gamma\gamma^\top A^*}{\QAst(\gamma)}\right)A^{*1/2}.\] 
For any $j\in[p_z]$,  
 $\|A_{0,j}^*\|_2^2$ is the $j$-th diagonal element of $A_0^* A_0^{*\top}$ given as 
 \[\begin{aligned}
     \|A_{0,j}^*\|_2^2 &= \sigma_{jz}^2\left(1-\dfrac{\gamma_j^2 \sigma_{jz}^2}{\QAst(\gamma)}\right) = \sigma_{jz}^2\left(1-\dfrac{\gamma_j^2 \sigma_{jz}^2}{\sum_{j\in[p_z]}\gamma_j^2 \sigma_{jz}^2}\right),
 \end{aligned}\]
 which is strictly bounded from below by $(1-C_\gamma)\sigma_{jz}^2$.  Proposition \ref{prop:A} then follows by the fact that $\sigma_{jz}^2$ is uniformly lower bounded for all $j\in[p_z]$ implied by the bounded eigenvalues of the population Gram matrix specified in Assumption \ref{as:DecisionMatrix}.
\end{proof}

Proposition \ref{prop:GaussianApprox} shows the Gaussian Approximation property for the key component in the test statistic, which is the key for the asymptotic size and power of the M test. Define 
\begin{equation}\label{eq:def eiAst}
    e_{iA^*} := \varepsilon_{i,Y} - \varepsilon_{i,D}\beta_{A^*}, 
\end{equation}
and $e_{A^*} = ( e_{1A^*}, e_{2A^*},\cdots, e_{nA^*})^\top$.  
\begin{proposition} \label{prop:GaussianApprox} Define $\xi_{i\cdot} = A_0^*\Omega_z W_{i\cdot}e_{iA^*}$ and $\xi_{ij}$ as the $j$-th element of $\xi_{i\cdot}$ for any $j\in[p_z]$. Suppose that $\pi\in\mathcal{H}_{A^*}(t)$. Under Assumptions \ref{as:DecisionMatrix}-\ref{as:asym} and \ref{as: two strong IV},
\begin{equation}\label{eq:GaussianApprox}
    \sup_{x\in\mathbb{R}} \left|\Pr\left( \max_{j\in[p_z]} \dfrac{\sumn \xi_{ij}} {\sqrt{n}} \leq x \right) - \Pr\left( \max_{j\in[p_z]} \dfrac{\sumn a_{ij}} {\sqrt{n}} \leq x \right)\right| \lesssim Cn^{-c},
\end{equation}
for some absolute constant $c$, where $\left \{a_{i\cdot} = (a_{i,Y},\cdots,a_{ip_z})^\top \right \}_{i=1}^{n}$ is a sequence of mean zero Gaussian vector with covariance matrix 
\begin{equation}\label{def:VAstar}
    {\rm V}_{A^*}  := A_0^*\Omega_z\mathbb{E}[W_{i\cdot}W_{i\cdot}\sigma_{iA^*}^2]\Omega_z^\top A_0^{*\top}.
\end{equation} 
\end{proposition}
\begin{proof}[Proof of Proposition \ref{prop:GaussianApprox}]By Corollary 2.1 of \citet{chernozhukov2013gaussian}, it suffices to show
\begin{enumerate}
        \item $c \leq n^{-1}\sumn \mathbb{E}[\xi_{ij}^2] \leq C$ for all $j\in[p_z]$. 
    \item $\max_{k=1,2}n^{-1}\sumn \mathbb{E}[|\xi_{ij}|^{2+k}/C^k] + \mathbb{E}\left[\exp(|\xi_{ij}|/C)\right]<4$ for some  large enough absolute constant $C$. Here the constant $C$ is a counterpart of ``$B_n$'' in  \citet{chernozhukov2013gaussian}. 
\end{enumerate}
Then (\ref{eq:GaussianApprox}) follows by  Corollary 2.1 of \citet{chernozhukov2013gaussian}, given that $B_n[\log (np)]^7/n = O\left(n^{-\nu/(7+\nu)}\right)$ implied by Assumption \ref{as:asym}. 
\par \textbf{\underline{Step 1.}} Show $c \leq n^{-1}\sumn \mathbb{E}[\xi_{ij}^2] \leq C$. By the law of iterated expectations 
\[\begin{aligned}
    \mathbb{E}(\xi_{i\cdot}\xi_{i\cdot}^\top) &= \mathbb{E}(a_{i\cdot}a_{i\cdot}^\top) \\
    &=  A_0^*\Omega_z\mathbb{E}\left[ W^{i\cdot}W_{i\cdot}^\top\mathbb{E}(e_{iA^*}^2|W) \right]\Omega_z^\top A_0^{*\top}\\
    &=  A_0^*\Omega_z\mathbb{E}\left[ W_{i\cdot}W_{i\cdot}^\top \sigma_{iA^*}^2\right] \Omega_z^\top A_0^{*\top}. 
\end{aligned}\]
Let $\delta_j$ be the $j$-th standard basis vector of $\mathbb{R}^{p_z}$. Then by (\ref{eq:bound_sigma_iA}) $\sigma_{iA^*}^2\asymp 1$. Hence, 
\[\begin{aligned}
\mathbb{E}[\xi_{ij}^2]  = \delta_j^\top \mathbb{E}(\xi_{i\cdot}\xi_{i\cdot}^\top) \delta_j &\gtrsim \cdot \delta_j^\top  A_0^*\Omega_z \Sigma \Omega_z^\top A_0^{*\top} \delta_j \\
&\gtrsim \delta_j^\top  A_0^* A_0^{*\top}\delta_j \\
&\gtrsim \min_{j\in[p_z]}\|A_{0,j}^*\|_2 \gtrsim 1. 
\end{aligned}\]
where the last inequality is deduced by Proposition \ref{prop:A}. Similarly, 
\[\begin{aligned}
\mathbb{E}[\xi_{ij}^2]  &\leq \sigma_{\max}^2 \cdot \delta_j^\top  A_0^*\Omega_z \Sigma \Omega_z^\top A_0^{*\top} \delta_j \\
&\lesssim \delta_j^\top  A_0^* A_0^{*\top}\delta_j \\
&\lesssim \lambda_{\max}(A^*) \leq C_{A^*}. 
\end{aligned}\]
\par \textbf{\underline{Step 2.}} It suffices to show that $\xi_{ij}$ is sub-exponential satisfying for any $\mu>0$, $\Pr(|\xi_{ij}|>\mu)\leq C\exp(-c\mu)$. Since $\Wi$ is a sub-Gaussian vector with bounded sub-Gaussian norm and $A_0^*\Omega_z$ has $L_2$ norm bounded from above, by Lemma \ref{lem:sub-Gaussian trans}, the entries of $A_0^*\Omega_z\Wi$ are sub-Gaussian variables. By Sub-Gaussianity of $\varepsilon_{i,D}$, $A_0^*\Omega_z\Wi\varepsilon_{i,D}$ is sub-exponential. It then turns out that $\xi_{ij}$ is  sub-exponential, since it is an element in the sub-exponential vector  $A_0^*\Omega_z\Wi\varepsilon_{i,D}$. This completes Step 2. 
\end{proof}

Proposition \ref{prop:debias} provides a decomposition of the debiased Lasso estimator $\tilde\pi_A$ of the target vector $\pi_{A^*}$. 
\begin{proposition}\label{prop:debias}Suppose that $\pi\in\mathcal{H}_{A^*}(t)$. Under Assumptions \ref{as:DecisionMatrix}-\ref{as:tuning}, the estimation error of $A^{1/2}\tilde{\pi}_A$ is decomposed as 
\begin{equation}\label{eq:DBLassoPiDecom}
     A^{1/2}\left(\tilde\pi_A-\pi_A\right) = \dfrac{A_0^*\Omega_z W^\top e_{A^*}}{n} +  \Delta_A, 
\end{equation}
with $A_0^*=A^{*1/2}\left(I_{p_z} - \dfrac{\gamma\gamma^\top A^*}{\QAst(\gamma)}\right)$ and  $\|\Delta_A\|_\infty = o_p\left(\dfrac{1}{\sqrt{n}\log p}\right)$. 
\end{proposition} 

\begin{proof}[Proof of Proposition \ref{prop:debias}]By definition of $\tilde\pi_A$, 
\begin{equation}\label{eq:DBLassoPiInitialDecom}
    \begin{aligned}
    A^{1/2}\left(\tilde\pi_A-\pi_A\right) &=A^{1/2}\left(\widecheck\pi_A - \pi_A\right) + A^{1/2}(\hOmega \hSigma - I_{p_z})_z \left(\begin{array}{c}
       \widecheck\varphi_A -  \hphi_A  \\
       \widecheck\pi_A  - \hpi_A  
    \end{array}\right) +  A^{1/2}\dfrac{\hOmega_z W^\top \widecheck e_A}{n} \\
    &=  A^{1/2}\left( \dfrac{\Omega_z W^\top e_A}{n} - \gamma(\hbeta_A-\beta_A) \right) + \\
    &\ \ \ \ \ A^{1/2} (\hOmega \hSigma - I_{p})_z \left(\begin{array}{c}
       \widecheck\varphi_A -  \hphi_A  \\
       \widecheck\pi_A  - \hpi_A  
    \end{array}\right) +  A^{1/2}\dfrac{\hOmega_z W^\top (\widecheck e_A - e_A)}{n} \\
    &=  \dfrac{A_0\Omega_z W^\top e_A}{n} + \dfrac{A^{1/2}\gamma\gamma^\top A}{n\QA(\gamma)} -  A^{1/2}\gamma(\hbeta_A-\beta_A)  + \\
    &\ \ \ \ \ A^{1/2} (\hOmega \hSigma - I_{p})_z \left(\begin{array}{c}
       \widecheck\varphi_A -  \hphi_A  \\
       \widecheck\pi_A  - \hpi_A  
    \end{array}\right) +  A^{1/2}\dfrac{\hOmega_z W^\top (\widecheck e_A - e_A)}{n} \\
    &= \dfrac{A_0^*\Omega_z W^\top e_{A^*} }{n} + \Delta_{1\pi} + \Delta_{2\pi} + \Delta_{3\pi} + \Delta_{4\pi},
\end{aligned}
\end{equation}
where $(\hOmega \hSigma - I_{p})_z$ is the $p_z\times p$ submatrix composed of the last $p_z$ rows of $\hOmega \hSigma - I_{p}$, and 
\[\Delta_{1\pi} = A^{1/2}(\hOmega \hSigma - I_{p})_z \left(\begin{array}{c}
       \widecheck\varphi_A -  \hphi_A  \\
       \widecheck\pi_A  - \hpi_A  
    \end{array}\right), \]
     \[\Delta_{2\pi} =  A^{1/2}\dfrac{\hOmega_z W^\top (\widecheck e_A - e_A)}{n},\]
     \begin{equation}\label{eq:Delta3pi}
         \Delta_{3\pi} = \dfrac{A^{1/2}\gamma\gamma^\top AW^\top e_A}{n\QA(\gamma)}  - A^{1/2}\gamma(\hat\beta_A - \beta_A),
     \end{equation}
      \[\Delta_{4\pi} =  \dfrac{A_0 \Omega_z W^\top e_{A} }{n} - \dfrac{A_0^* \Omega_z W^\top e_{A^*} }{n}.\]
\par \underline{\textbf{Bound $\Delta_{1\pi}$}}. By Assumption \ref{as:asym}, Lemmas \ref{lem:CLIME}, \ref{lem:asymp} and Proposition \ref{prop:Xihat}, 
\begin{equation}\label{eq:supNormRootN1}
\begin{aligned}
    \|\Delta_{1\pi}\|_\infty &\lep \|\Omega \|_1\cdot\|\hOmega \hSigma - I_{p}\|_\infty \left(\|\widecheck\varphi_A - \hphi_A\|_1 + \|\widecheck\pi_A - \hpi_A\|_1 \right) \\
    &\lep m_\omega \cdot s_\omega \dfrac{m_\omega^{2-2q}(\log p)^{(1-q)/2}}{n^{(1-q)/2}} \cdot \sqrt{\dfrac{s^2 \log p }{n} }  \\
    & = \dfrac{1}{\sqrt{n}\log p}\cdot\dfrac{s_\omega m_\omega^{3-2q} s (\log p)^{(4-q)/2}}{n^{(1-q)/2}} \\ 
    &= o(n^{-1/2}(\log p)^{-1}). 
\end{aligned}
\end{equation}
\par \underline{\textbf{Bound $\Delta_{2\pi}$}}. By (\ref{eq:betahat_rate}) $|\hbeta_A-\beta_A| = O_p(n^{-1/2}\|\gamma\|_2^{-1})$. Additionally by (\ref{eq:WeBound}) and Proposition \ref{prop:A and A*}, 
\begin{equation}\label{eq:supNormRootN2}
    \begin{aligned}
   \|\Delta_{2\pi}\|_\infty  &\leq \|A^{1/2}\|_1\|\Omega\|_1 \cdot |\hbeta_A-\beta_A| \cdot \left\|\dfrac{W^\top \varepsilon_D}{n}\right\|_\infty \\
    &= O_p\left(\dfrac{m_\omega\sqrt{\log p}}{n\|\gamma\|_2}\right) = O_p\left(\dfrac{1}{\sqrt{n}\log p}\cdot\dfrac{m_\omega (\log p)^{3/2}}{\sqrt{n}\|\gamma\|_2}\right) = o_p(n^{-1/2}(\log p)^{-1}).
\end{aligned}
\end{equation}
\par \underline{\textbf{Bound $\Delta_{3\pi}$}}. By (\ref{eq:consistent_Qgamma}), (\ref{eq:root_n_betahat_decom}), Assumption 4 and Lemma (\ref{lem:asymp}),
\begin{equation}\label{eq: want 1}
    \begin{aligned}
    \hbeta_A-\beta_A &= \dfrac{u_\gamma^\top W^\top e_A}{n\hQA(\gamma)} + \dfrac{1}{\sqrt{n}\hQA(\gamma)}o_p(\|\gamma\|_2)\\ 
    &= \dfrac{\gamma^\top A\Omega_z W^\top e_A}{n\cdot\hQA(\gamma)} + o_p((n\QA(\gamma))^{-1/2}(\log p)^{-1}) \\
    &= \dfrac{\gamma^\top A\Omega_z W^\top e_A}{n\QA(\gamma)} + \dfrac{\gamma^\top A\Omega_z W^\top e_A}{n}\left[\dfrac{1}{\hQA(\gamma)}-\dfrac{1}{\QA(\gamma)}\right]   + o_p((n\QA(\gamma))^{-1/2}(\log p)^{-1}). 
\end{aligned}
\end{equation}
 
Then by (\ref{eq:WeBound}), (\ref{eq:consistent_Qgamma}) and (\ref{eq: beta A bound}), 
\begin{equation}\label{eq: want 2}\begin{aligned}
    &\ \ \left|\dfrac{\gamma^\top A\Omega_z W^\top e_A}{n}\left[\dfrac{1}{\hQA(\gamma)}-\dfrac{1}{\QA(\gamma)}\right]\right|\\
    &\leq \|\Omega_z^\top A\gamma\|_1 \left\|\dfrac{W^\top e_A}{n}\right\|_\infty\cdot \dfrac{|\hQA(\gamma)-\QA(\gamma)|}{\hQA(\gamma)\QA(\gamma)} \\
    &\lep \dfrac{\|\gamma\|_1}{\QA^2(\gamma)} \sqrt{\dfrac{m_\omega^2\log p}{n}}\cdot  O_p\left(\dfrac{m_\omega s\log p}{n} + \left[\dfrac{s_\omega m_\omega^{2-2q}\cdot s^{1/2}(\log p)^{1-q/2}}{n^{1-q/2}}+\sqrt{\dfrac{m_\omega^2 s^2 \log p}{n}}\right]\|\gamma\|_2\right)\\
    &\lep \dfrac{1}{\|\gamma\|_2^3} \sqrt{\dfrac{m_\omega^2 s \log p}{n}}\cdot  O_p\left(\dfrac{m_\omega s\log p}{n} + \left[\dfrac{s_\omega m_\omega^{2-2q}\cdot s^{1/2}(\log p)^{1-q/2}}{n^{1-q/2}}+\sqrt{\dfrac{m_\omega^2 s^2 \log p}{n}}\right]\|\gamma\|_2\right). 
\end{aligned}\end{equation}
Then by (\ref{eq: want 1}), (\ref{eq: want 2}), Assumption \ref{as:asym}, Lemma \ref{lem:asymp} and Proposition \ref{prop:A and A*}, 
\[\begin{aligned}
    \|\Delta_{3\pi}\|_\infty 
    &= \left\|A^{1/2}\left[\gamma(\hat\beta_A - \beta_A) - \dfrac{\gamma\gamma^\top A\Omega_z W^\top e_A}{n\QA(\gamma)}\right]\right\|_\infty \\
    &= \|A^{1/2}\|_1\cdot\left\|\dfrac{\gamma\gamma^\top A\Omega_z W^\top e_A}{n}\left[\dfrac{1}{\hQA(\gamma)}-\dfrac{1}{\QA(\gamma)}\right]   + \gamma \cdot o_p((n\QA(\gamma))^{-1/2}(\log p)^{-1}) \right\|_\infty\\
    &\lep \dfrac{\|\gamma\|_\infty}{\sqrt{n}\|\gamma\|_2^3}O_p\left(\dfrac{m_\omega^2 s^{3/2} (\log p)^{3/2}}{n}+\left[\dfrac{s_\omega m_\omega^{3-2q}s (\log p)^{(5-q)/2}}{n^{1-q/2}}+\dfrac{m_\omega^2 s^{3/2} \log p}{\sqrt{n}}\right]\|\gamma\|_2\right) + \\
    &\ \ \ \ \ \ \ \|\gamma\|_\infty  o_p((n\QA(\gamma))^{-1/2}(\log p)^{-1}) \\
    &= n^{-1/2}O_p\left(\dfrac{m_\omega^2 s^{3/2} (\log p)^{3/2}}{n\|\gamma\|_2^2}+\dfrac{s_\omega m_\omega^{3-2q}s(\log p)^{(5-q)/2}}{n^{1-q/2}\|\gamma\|_2}+\dfrac{m_\omega^2 s^{3/2} \log p}{\sqrt{n}\|\gamma\|_2}\right) + o_p(n^{-1/2}(\log p)^{-1})  \\
    &=\dfrac{1}{\sqrt{n}\log p} O_p\left(\dfrac{m_\omega^2 s^{3/2} (\log p)^{5/2}}{n\|\gamma\|_2}\right) +  \\
    &\ \ \ \ \ \  \dfrac{1}{\sqrt{n}\log p}O_p\left[\dfrac{s_\omega m_\omega^{2-2q}s^{1/2}(\log p)^{(5-q)/2}}{n^{(1-q)/2}}\cdot \dfrac{m_\omega s^{1/2}\log p}{\sqrt{n}\|\gamma\|_2} + \dfrac{m_\omega^2 s^{3/2} (\log p)^2}{\sqrt{n}\|\gamma\|_2}\right] + o_p(n^{-1/2}(\log p)^{-1}) \\
    &=o_p(n^{-1/2}(\log p)^{-1}).
\end{aligned}\]
\par \underline{\textbf{Bound $\Delta_{4\pi}$}}. We first bound $\|A_0^*\|_1$. Since \[\|\gamma\gamma^\top A^*\|_1\leq \|\gamma\|_1\cdot\|\gamma\|_\infty\cdot\|A^*\|_1 \lesssim \|\gamma\|_2\cdot\sqrt{s}\|\gamma\|_2 \lesssim \sqrt{s}\QAst(\gamma), \]
we deduce that 
\begin{equation}\label{eq: Bound A0star 1}
 \|A_0^*\|_1 \leq \|A^{*1/2}\|_1\cdot \left\|I_{p_z} - \dfrac{\gamma\gamma^\top A^*}{\QAst(\gamma)}\right\|_1 \lesssim \|I_{p_z}\|_1 + \dfrac{\sqrt{s}\QAst(\gamma)}{\QAst(\gamma)} \lesssim \sqrt{s}. 
\end{equation}
\par Note that 
\[\|\Delta_{4\pi}\|_\infty \leq   \left\|\dfrac{A_0^*\Omega_z W^\top (e_A-e_{A^*})}{n}\right\|_\infty +  \left\|\dfrac{(A_0-A_0^*)\Omega_z W^\top e_A}{n}\right\|_\infty,\]
where the first term on the RHS is bounded by 
\[\begin{aligned}
\left\|\dfrac{A_0^*\Omega_z W^\top (e_A-e_{A^*})}{n}\right\|_\infty &\lesssim \|A_0^*\|_1 \|\Omega\|_1 \left\|\dfrac{W^\top \varepsilon_D}{n}\right\|_\infty\cdot |\beta_A - \beta_{A^*}| \\
&\lep \sqrt{s}\cdot \dfrac{m_\omega\sqrt{\log p}}{\sqrt{n}}\cdot \sqrt{\dfrac{\log p}{n}} = \dfrac{1}{\sqrt{n}\log p}\cdot\dfrac{m_\omega \sqrt{s}(\log p)^2}{\sqrt{n}} = o_p\left(\dfrac{1}{\sqrt{n}\log p}\right), 
\end{aligned}\]
where the second inequality applies (\ref{eq:WeBound}), (\ref{eq: betaA betaAst}) and (\ref{eq: Bound A0star 1}), and the last step applies Lemma \ref{lem:asymp}. 
It thus suffices to show that 
 \begin{equation*}
     \left\|\dfrac{(A_0-A_0^*)\Omega_z W^\top e_A}{n}\right\|_\infty = o_p\left(\dfrac{1}{\sqrt{n}\log p}\right).
 \end{equation*}
Note that by Proposition \ref{prop:A and A*}, 
\[\begin{aligned}
    |\QA(\gamma) - \QAst(\gamma)| \leq \|\gamma\|_1^2 \cdot \|A-A^*\|_\infty \lep \|\gamma\|_2^2\cdot \sqrt{\dfrac{\log p}{n}},
\end{aligned}\]
and hence $\QA(\gamma)/\QAst(\gamma) \convp 1 $ and 
\[\left|\dfrac{1}{\hQA(\gamma)}-\dfrac{1}{\QAst(\gamma)}\right| = \dfrac{|\QA(\gamma) - \QAst(\gamma)|}{\QAst(\gamma)\QA(\gamma)} \lep \dfrac{1}{\QA(\gamma)}\sqrt{\dfrac{\log p}{n}}.\]
Then by Proposition \ref{prop:A and A*}, 
\begin{equation}\label{eq:A0starA0L1}
    \begin{aligned}
    \|A_0^*-A_0\|_1 &\leq \|A^{1/2} - A^{*1/2}\|_1\cdot\left\|I_{p_z} - \dfrac{\gamma\gamma^\top A}{\QA(\gamma)}\right\|_1 + \|A^{1/2}\|_1 \cdot \left\|\dfrac{\gamma\gamma^\top A}{\QA(\gamma)} - \dfrac{\gamma\gamma^\top A^*}{\QAst(\gamma)}\right\|_1 \\
    &\lep \sqrt{\dfrac{\log p}{n}}\cdot \left(1+\dfrac{\|\gamma\|_1^2\cdot \|A\|_1}{\QA(\gamma)}\right) +\left\|\dfrac{\gamma\gamma^\top (A-A^*)}{\QA(\gamma)} \right\|_1 + \left\|\gamma\gamma^\top A^*\left[\dfrac{1}{\hQA(\gamma)}-\dfrac{1}{\QAst(\gamma)}\right]\right\|_1 \\
     &\lep \sqrt{\dfrac{\log p}{n}}\cdot \left(1+s\right) +\dfrac{\|\gamma\|_1^2\|A-A^*\|_1}{\QA(\gamma)} + \dfrac{\|\gamma\|_1^2\cdot\|A\|_1}{\QA(\gamma)}\sqrt{\dfrac{\log p}{n}} \\
     &\lep \sqrt{\dfrac{s^2\log p}{n}} + \dfrac{2\|\gamma\|_2^2}{\QA(\gamma)}\cdot s\cdot\sqrt{\dfrac{\log p}{n}}  \lep \sqrt{\dfrac{s^2\log p}{n}}.
\end{aligned}
\end{equation}
This implies 
 \begin{equation*}
     \left\|\dfrac{(A_0-A_0^*)\Omega_z W^\top e_A}{n}\right\|_\infty \leq \|A_0^*-A_0\|_1\cdot \|\Omega\|_1 \cdot \|n^{-1}W^\top e_A\|_\infty = O_p \left(\dfrac{m_\omega s\log p}{n}\right)= o_p\left(\dfrac{1}{\sqrt{n}\log p}\right),
 \end{equation*}
where the last inequality applies Lemma \ref{lem:asymp}. This completes the proof of Proposition \ref{prop:debias}. 
\end{proof}

Proposition \ref{prop:VA} provides a probability upper bound for the estimation error of ${\rm V}_A$. Define 
\begin{equation}\label{eq: def A0 star}
    A_0^* = A^{*1/2}\left(I_{p_z}-\dfrac{\gamma\gamma^\top A^*}{\QAst(\gamma)}\right), 
\end{equation} 
with $A^*$ defined as (\ref{eq:def Astar}). Recall that $\Omega_z$ is the $p_z\times p $ submatrix composed of the last $p_z$ rows of $\Omega := \Sigma^{-1}$.
\begin{proposition}\label{prop:VA}Suppose that $\pi\in\mathcal{H}_{A^*}(t)$. Under Assumptions \ref{as:DecisionMatrix}-\ref{as:tuning}, 
\begin{equation}\label{eq:hatVAErr}
    \|\hat {\rm V}_A - {\rm V}_{A^*} \|_\infty  = o_p\left(\dfrac{1}{(\log p)^3}\right),
\end{equation}
where $\hat{\rm V}_{A}$ is defined as (\ref{eq:hatVA}) and 
\begin{equation}\label{eq:def-VA}
     {\rm V}_{A^*} =A_0^*\Omega_z  \mathbb{E}\left[  W_{i\cdot}W_{i\cdot}\sigma_{iA^*}^2 \right] \Omega_z^\top A_0^{*\top}.
\end{equation}
\end{proposition} 
\begin{proof}[Proof of Proposition \ref{prop:VA}]We bound the estimation error of $\hat{\rm V}_A$ as 
\begin{equation}\label{eq: hatVA hatVAstar initial bound}\begin{aligned}
\left\|\hat{\rm V}_A -  {\rm V}_{A^*}\right\|_\infty &\leq \left|\hat A_0\hat\Omega_z\left(\dfrac{\sumn(\Wi\Wi^\top\hat e_{iA}^2 - \mathbb{E}(\Wi\Wi^\top\sigma_{iA^*}^2))}{n}\right)\hat\Omega_z^\top \hat A_0^\top\right| + \\
&\ \ \ \ \left|\left(\hat{A}_0\hat\Omega_z - A_0^*\Omega_z\right)\mathbb{E}(\Wi\Wi^\top\sigma_{iA^*}^2) \Omega_z^\top A_0^{*\top}\right| + \left| A_0^*\Omega_z \mathbb{E}(\Wi\Wi^\top\sigma_{iA^*}^2) \left(\hat{A}_0\hat\Omega_z - A_0^*\Omega_z\right)^\top\right| \\
&\leq  \|\hat{A}_0\|_1^2 \|\hat\Omega_z\|_1^2\cdot\left\|\dfrac{\sumn(\Wi\Wi^\top\hat e_{iA}^2 - \mathbb{E}(\Wi\Wi^\top\sigma_{iA^*}^2))}{n}\right\|_\infty + \\
&\ \ \ \ \ \ \ 2\|A_0^*\|_1\|\Omega_z\|_1\|\mathbb{E}(\Wi\Wi^\top\sigma_{iA^*}^2)\|_\infty\cdot\|\hat{A}_0\hat\Omega_z - A_0^*\Omega_z\|_1 \\
&\lep \|\hat{A}_0\|_1^2\cdot m_\omega^2 \cdot \left(\dfrac{s^2\log p}{n} + \left(1+\dfrac{1}{\|\gamma\|_2}\right)\sqrt{\dfrac{\log p}{n}}\right) + \sqrt{s} \cdot m_\omega\|\hat{A}_0\hat\Omega_z - A_0^*\Omega_z\|_1,
\end{aligned}\end{equation}
where the last inequality applies (\ref{eq:CLIMEestL1Bound}), (\ref{eq: Bound A0star 1}), Proposition \ref{prop:bound_sd_hat_for_pi} and that fact that $\|\mathbb{E}(\Wi\Wi^\top\sigma_{iA^*}^2)\|_\infty\lep 1$ follows the arguments above (\ref{eq: DB WWsigma}). 
It remains to bound $\|\hat A_0\|_1$, $\|A_{0}^*\|_1$ and $\|\hat{A}_0\hat\Omega_z - A_0^*\Omega_z\|_1$.

\par \underline{\textbf{Bound $\|\hat A_0\|_1$}.}We first bound $\|\hat A_0 - A_0^*\|_1$. Define 
\[A_0 := A^{1/2}\left(I_{p_z}-\dfrac{\gamma\gamma^\top A}{\QA(\gamma)}\right).\]
Note that 
\[\begin{aligned}
\|\hat A_0 - A_0^*\|_1 \leq \|\hat A_0 - A_0\|_1 + \| A_0 - A_0^*\|_1 \lep \|\hat A_0 - A_0\|_1 + \sqrt{\dfrac{s^2\log p}{n}}, 
\end{aligned}\]
where the second inequality follows by (\ref{eq:A0starA0L1}). We further bound the first term on the RHS that  
\[\begin{aligned}
\|\hat A_0 - A_0\|_1 &\leq \|A^{1/2}\|_1\cdot \left\|\dfrac{\hgamma\hgamma^\top A}{\hQA(\gamma)} - \dfrac{\gamma\gamma^\top A}{\QA(\gamma)}\right\|_1 \\
&\lep \dfrac{1}{\hQA(\gamma)}\left\|\hgamma\hgamma^\top A - \gamma\gamma^\top A\right\|_1 + \left\|\gamma\gamma^\top A\right\|_1\left|\dfrac{1}{\hQA(\gamma)}-\dfrac{1}{\QA(\gamma)}\right| \\
&\lep \dfrac{\|\hgamma\hgamma^\top - \gamma\gamma^\top\|_1\|A\|_1}{\QA(\gamma)} + \|\gamma\|_\infty\|\gamma\|_1\|A\|_1\cdot\dfrac{|\hQA(\gamma)-\QA(\gamma)|}{\hQA(\gamma)\QA(\gamma)} \\
&\lep  \dfrac{\|(\hgamma - \gamma)(\hgamma-\gamma)^\top\|_1}{\QA(\gamma)} + \dfrac{\|(\hgamma - \gamma)\gamma^\top\|_1}{\QA(\gamma)} + \dfrac{\|\gamma(\hgamma - \gamma)^\top\|_1}{\QA(\gamma)}   + \\
&\ \ \ \ \ \|\gamma\|_\infty\|\gamma\|_1\|A\|_1\cdot\dfrac{|\hQA(\gamma)-\QA(\gamma)|}{\hQA(\gamma)\QA(\gamma)}. \\
\end{aligned}\]
Since by Proposition \ref{prop:LassoTheta},
\[\dfrac{\|(\hgamma - \gamma)(\hgamma-\gamma)^\top\|_1}{\QA(\gamma)} \leq \dfrac{\|\hgamma - \gamma\|_\infty\|\hgamma - \gamma\|_1}{\QA(\gamma)} \lep \dfrac{s\log p}{n\QA(\gamma)},\]
\[\dfrac{\|(\hgamma - \gamma)\gamma^\top\|_1}{\QA(\gamma)} \leq \dfrac{\|\hgamma - \gamma\|_\infty\|\gamma\|_1}{\QA(\gamma)} \lep \dfrac{1}{\QA(\gamma)}\sqrt{\dfrac{s\log p}{n}}\cdot\sqrt{s}\|\gamma\|_2\lep\dfrac{1}{\|\gamma\|_2}\sqrt{\dfrac{s^2\log p}{n}},\]
\[\dfrac{\|\gamma(\hgamma-\gamma)^\top\|_1}{\QA(\gamma)} \leq \dfrac{\|\gamma \|_\infty\|\hgamma - \gamma\|_1}{\QA(\gamma)} \lep \dfrac{\|\gamma\|_2}{\QA(\gamma)}\sqrt{\dfrac{s^2\log p}{n}}\lep \dfrac{1}{\|\gamma\|_2}\sqrt{\dfrac{s^2\log p}{n}},\]
and by (\ref{eq:consistent_Qgamma}),
\[\begin{aligned}
&\ \ \ \ \ \|\gamma\|_\infty\|\gamma\|_1\|A\|_1\cdot\dfrac{|\hQA(\gamma)-\QA(\gamma)|}{\hQA(\gamma)\QA(\gamma)} \\
&\lep \dfrac{\|\gamma\|_2\cdot\sqrt{s}\|\gamma\|_2}{\QA^2(\gamma)}\left[\dfrac{m_\omega s\log p}{n}+\|\gamma\|_2\left(\dfrac{s_\omega m_\omega^{2-2q} \cdot s^{1/2}(\log p)^{1-q/2}}{n^{1-q/2}}+\sqrt{\dfrac{m_\omega^2 s^2 \log p}{n}}\right)\right] \\
&\lep \dfrac{m_\omega s^{3/2}\log p}{n\QA(\gamma)} + \dfrac{1}{\|\gamma\|_2}\left(\dfrac{s_\omega m_\omega^{2-2q} \cdot s^{1/2}(\log p)^{1-q/2}}{n^{1-q/2}}+\sqrt{\dfrac{m_\omega^2 s^2 \log p}{n}}\right).
\end{aligned}\]
We can deduce that 
\[\|\hat A_0 - A_0\|_1 \lep \dfrac{m_\omega s^{3/2}\log p}{n\QA(\gamma)} + \dfrac{1}{\|\gamma\|_2}\left(\dfrac{s_\omega m_\omega^{2-2q} \cdot s^{1/2}(\log p)^{1-q/2}}{n^{1-q/2}}+\sqrt{\dfrac{m_\omega^2 s^2 \log p}{n}}\right),\]
and thus, 
\begin{equation}\label{eq: hatA0 A0star L1}
    \|\hat A_0 - A_0^*\|_1 \lep \dfrac{m_\omega s^{3/2}\log p}{n\QA(\gamma)} + \dfrac{1}{\|\gamma\|_2}\left(\dfrac{s_\omega m_\omega^{2-2q} \cdot s^{1/2}(\log p)^{1-q/2}}{n^{1-q/2}}+\sqrt{\dfrac{m_\omega^2 s^2 \log p}{n}}\right)  + \sqrt{\dfrac{s^2\log p}{n}}.
\end{equation}
By Assumption \ref{as:asym} and Lemma \ref{lem:asymp}, $\|\hat A_0 - A_0^*\|_1 = o_p(1)$ and hence 

\begin{equation}\label{eq: hat A0 L1}
    \|\hat A_0\|_1 \leq \|\hat A_0 - A_0^*\|_1 + \|A_0^*\|_1 \lep \sqrt{s}.
\end{equation}


\par \underline{\textbf{Bound $\|\hat A_0\hOmega_z - A_0^*\Omega_z\|_1$}}. Note that by (\ref{eq: hatA0 A0star L1}) and Lemma \ref{lem:CLIME},
\[\begin{aligned}
\|\hat A_0\hOmega_z - A_0^*\Omega_z\|_1 &\leq \|\hat A_0 - A_0^*\|_1\|\hOmega\|_1 + \|A_0^*\|_1\|\hOmega - \Omega\|_1 \\
&\leq m_\omega\left[\dfrac{m_\omega s^{3/2}\log p}{n\QA(\gamma)} + \sqrt{\dfrac{s^2\log p}{n}}+ \dfrac{1}{\|\gamma\|_2}\left(\dfrac{s_\omega m_\omega^{2-2q} \cdot s^{1/2}(\log p)^{1-q/2}}{n^{1-q/2}}+\sqrt{\dfrac{m_\omega^2 s^2 \log p}{n}}\right)\right]  \\
&\ \ \ \ \ \ +\sqrt{s}\cdot\dfrac{s_\omega\cdot m_\omega^{2-2q}\cdot (\log p)^{(1-q)/2}}{n^{(1-q)/2}}. 
\end{aligned}\]

Then by (\ref{eq: hatVA hatVAstar initial bound}),
\[\begin{aligned}
\left\|\hat{\rm V}_A -  {\rm V}_{A^*}\right\|_\infty &\leq s\cdot m_\omega^2 \left(\dfrac{s^2\log p}{n} + \dfrac{1}{\|\gamma\|_2}\sqrt{\dfrac{\log p}{n}} \right) \\
&\ \ \ \ +\dfrac{m_\omega^3 s^{2}\log p}{n\QA(\gamma)} + m_\omega^2\sqrt{\dfrac{ s^3\log p}{n}}+ \dfrac{1}{\|\gamma\|_2}\left(\dfrac{s_\omega m_\omega^{3-2q} \cdot s(\log p)^{1-q/2}}{n^{1-q/2}}+m_\omega^3\cdot\sqrt{\dfrac{s^3 \log p}{n}}\right) \\
&\ \ \ \   + \dfrac{s_\omega\cdot m_\omega^{3-2q}\cdot s (\log p)^{(1-q)/2}}{n^{(1-q)/2}}. 
\end{aligned}\]
Then it follows by Assumption \ref{as:asym} and Lemma \ref{lem:asymp} that $\left\|\hat{\rm V}_A -  {\rm V}_{A^*}\right\|_\infty = o_p(1/(\log p)^3)$. 
\end{proof}